\documentclass[pra,twocolumn,floatfix,superscriptaddress,notitlepage,longbibliography]{revtex4-2}

\usepackage[T1]{fontenc}             % Ensures proper encoding for PDF output.
\usepackage{amsmath, amsfonts, amssymb, amsthm} % For advanced mathematical typesetting.
\usepackage{mathtools}               % Enhances amsmath, e.g., for \coloneqq and showonlyrefs.
\usepackage[many]{tcolorbox}         % For colored boxes around text, the 'many' option loads many libraries.
\makeatletter
\mathchardef\ordinarycolon="603A
\mathchardef\ordinaryequal="303D

\begingroup
  \lccode`~=`\:
\lowercase{\endgroup
  \def~{\futurelet\@let@token\colon@active}
}
\def\colon@active{%
  \ifx\@let@token=%
    \expandafter\colon@eq
  \else
    \ordinarycolon
  \fi
}
\def\colon@eq={\coloneqq}

\begingroup
  \lccode`~=`\=
\lowercase{\endgroup
  \def~{\futurelet\@let@token\equal@active}
}
\def\equal@active{%
  \ifx\@let@token:%
    \expandafter\equal@colon
  \else
    \ordinaryequal
  \fi
}
\def\equal@colon:{\eqqcolon}

\AtBeginDocument{%
  \mathcode`\:=\string"8000
  \mathcode`\==\string"8000
}
\makeatother

\usepackage{xcolor}
\usepackage{graphicx}                % For including images.
\usepackage{tikz}                    % For creating vector graphics in LaTeX.
\usetikzlibrary{shapes.geometric, arrows.meta, positioning, fit, backgrounds}
\usepackage{quantikz}
\usetikzlibrary{3d}                  % Adds 3D functionalities to TikZ.

\usepackage{enumitem}                % Control over itemize, enumerate, and description.

\usepackage{multirow}                % Allows for multi-row cells in tables.
\usepackage{longtable}
\usepackage{booktabs}                % For professional-looking tables
\renewcommand{\thetable}{\arabic{table}}

\usepackage{physics}                 % For simplifying notation in physics.
\usepackage{nicefrac}                % For creating "nice" fractions with \nicefrac.
\usepackage{verbatim}                % For including raw text or comments.
\usepackage{ascmac}                  % For creating boxed text.
\usepackage[linesnumbered, ruled, vlined]{algorithm2e} % For typesetting algorithms. This package should be loaded after ascmac.
\SetKwInput{kwInit}{Init}            % Custom keyword for initialization in algorithms.
\usepackage[normalem]{ulem}          % For underlining and strikeout (sout).
\usepackage{fontawesome}             % For including icons in the document.
\usepackage{wrapfig}                % For wrapping text around figures.
\usepackage[top=25truemm,bottom=20truemm,left=20truemm,right=20truemm]{geometry}                % For setting page dimensions and margins.

\allowdisplaybreaks[4]               % Allows page breaks in multiline equations.
\usepackage{float}                   % Improved control over figure and table placement.
\usepackage{lineno}                   % For adding line numbers to the document.

\usepackage{apptools}                % Extends the functionality of theorem environments in appendices.
\usepackage{appendix}                % Manage appendices in the document.
\usepackage{titlesec}                % Customization of section titles.

\usepackage[colorlinks=true,citecolor=green,linkcolor=blue]{hyperref}

\theoremstyle{definition}
\newtheorem{theorem}{Theorem}
\newtheorem{lemma}{Lemma}
\newtheorem{corollary}{Corollary}
\newtheorem{proposition}{Proposition}

\newtheorem{definition}{Definition}
\newtheorem{assumption}{Assumption}
\newtheorem{remark}{Remark}

\AtAppendix{%
    \counterwithin{theorem}{section}
    \counterwithin{corollary}{section}
    \counterwithin{lemma}{section}
    \counterwithin{proposition}{section}
    \counterwithin{definition}{section}
    \counterwithin{assumption}{section}
    \counterwithin{remark}{section}
    \counterwithin{example}{section}
    \counterwithin{property}{section}
    \counterwithin{problem}{section}
}

\tcolorboxenvironment{definition}{
    colback=gray!20!white,
    boxrule=0pt,
    boxsep=1pt,
    left=5pt,right=5pt,top=5pt,bottom=5pt,
    oversize=2pt,
    sharp corners,
    before skip=\topsep,
    after skip=\topsep,
    breakable
}
\tcolorboxenvironment{assumption}{
    colback=gray!20!white,
    boxrule=0pt,
    boxsep=1pt,
    left=5pt,right=5pt,top=5pt,bottom=5pt,
    oversize=2pt,
    sharp corners,
    before skip=\topsep,
    after skip=\topsep,
    breakable
}
\tcolorboxenvironment{proposition}{
    colback=gray!20!white,
    boxrule=0pt,
    boxsep=1pt,
    left=5pt,right=5pt,top=5pt,bottom=5pt,
    oversize=2pt,
    sharp corners,
    before skip=\topsep,
    after skip=\topsep,
    breakable
}
\tcolorboxenvironment{theorem}{
    colback=gray!20!white,
    boxrule=0pt,
    boxsep=1pt,
    left=5pt,right=5pt,top=5pt,bottom=5pt,
    oversize=2pt,
    sharp corners,
    before skip=\topsep,
    after skip=\topsep,
    breakable
}
\tcolorboxenvironment{corollary}{
    colback=gray!20!white,
    boxrule=0pt,
    boxsep=1pt,
    left=5pt,right=5pt,top=5pt,bottom=5pt,
    oversize=2pt,
    sharp corners,
    before skip=\topsep,
    after skip=\topsep,
    breakable
}
\tcolorboxenvironment{lemma}{
    colback=gray!20!white,
    boxrule=0pt,
    boxsep=1pt,
    left=5pt,right=5pt,top=5pt,bottom=5pt,
    oversize=2pt,
    sharp corners,
    before skip=\topsep,
    after skip=\topsep,
    breakable
}

\newcommand{\n}[1]{\null}

\newif\ifverbose
\verbosetrue                         % Toggle to show or hide corrections.
\newcommand{\iid}{\text{i.i.d.\ }}

\newcommand{\bbC}{\mathbb{C}}
\newcommand{\bbD}{\mathbb{D}}

\newcommand{\bbN}{\mathbb{N}}

\newcommand{\bbP}{\mathbb{P}}
\newcommand{\bbQ}{\mathbb{Q}}
\newcommand{\bbR}{\mathbb{R}}

\newcommand{\calB}{\mathcal{B}}

\newcommand{\calD}{\mathcal{D}}

\newcommand{\calF}{\mathcal{F}}

\newcommand{\calN}{\mathcal{N}}
\newcommand{\calO}{\mathcal{O}}

\newcommand{\calT}{\mathcal{T}}
\newcommand{\calU}{\mathcal{U}}
\newcommand{\calV}{\mathcal{V}}

\newcommand{\poly}{\mathrm{poly}}

\DeclareMathOperator{\diag}{diag}
\DeclareMathOperator{\E}{\mathbb{E}}
\DeclareMathOperator{\Var}{Var}

\renewcommand{\bar}[1]{\overline{#1}}

\newcommand{\ind}{\,\mathrm{d}}

\newcommand{\bx}{\boldsymbol{x}}
\newcommand{\by}{\boldsymbol{y}}
\newcommand{\bth}{\boldsymbol{\theta}}

\newcommand{\dg}{^\dagger}

\newcommand{\T}{^\mathsf{T}}

\newcommand{\bs}{\boldsymbol}
\newcommand{\kett}[1]{\ket*{#1}\!\rangle}

\newcommand{\evv}[2]{\langle\!\langle{#1}|{#2}\rangle\!\rangle}

\newcommand{\vep}{\varepsilon}
\usepackage{dsfont}

\NewDocumentCommand{\set}{s m o}{%
  \IfBooleanTF{#1}%
    {%
      \IfValueTF{#3}%
        {\{~ #2 \,~:~ #3 ~\}}%
        {\{~ #2 ~\}}%
    }%
    {%
      \IfValueTF{#3}%
        {\left\{~ #2 \,~:~ #3 ~\right\}}%
        {\left\{~ #2 ~\right\}}%
    }%
}
\graphicspath{{img/}}

\begin{document}
\title{Double Descent in Quantum Kernel Ridge Regression}

\author{Kensuke Kamisoyama}
    \email[]{kamisoyama-kensuke564(at)g.ecc.u-tokyo.ac.jp}
    \affiliation{
    Department of Physics, Graduate School of Science, The University of Tokyo, 7-3-1 Hongo, Bunkyo-ku, Tokyo 113-0033, Japan
    }

\author{Lento Nagano}
    % \email[]{lento(at)icepp.s.u-tokyo.ac.jp}
    \affiliation{
    International Center for Elementary Particle Physics (ICEPP),
    The University of Tokyo, 7-3-1 Hongo, Bunkyo-ku, Tokyo 113-0033, Japan
    }

\author{Koji Terashi}
    % \email[]{terashi(at)icepp.s.u-tokyo.ac.jp}
    \affiliation{
    International Center for Elementary Particle Physics (ICEPP),
    The University of Tokyo, 7-3-1 Hongo, Bunkyo-ku, Tokyo 113-0033, Japan
    }

% \date{}

\begin{abstract}
    Various classical machine learning models, including linear models, kernel methods, and deep neural networks, exhibit double descent, in which the test risk peaks near the interpolation threshold and then decreases in the overparameterized regime.
    However, this phenomenon has received less attention in the quantum setting.
    In this work, we investigate the double descent phenomenon in quantum kernel ridge regression (QKRR).
    By applying tools from random matrix theory (RMT), we rigorously derive a deterministic-equivalent formula for the test risk of QKRR in the high-dimensional regime where both the effective feature dimension and the number of training samples grow large.
    Our analysis characterizes the double-descent peak of QKRR and reveals how explicit regularization effectively suppresses it.
    Numerical experiments show close agreement between this formula and empirical test errors even for 3- and 5-qubit models.
\end{abstract}

\maketitle

\section{Introduction}
Rapid advancements in quantum computing technologies have spurred the development of quantum machine learning (QML)~\cite{cerezo2022challenges,chang2025primer} as a highly promising research frontier.
By harnessing inherently quantum phenomena, such as superposition and entanglement, QML aims to process information in ways that are fundamentally inaccessible to classical computers.
However, realizing this potential requires overcoming a critical hurdle that mirrors the evolution of classical machine learning (CML): understanding and controlling the generalization capabilities of quantum models.

\subsection{Double Descent in Classical Machine Learning}
Conventional capacity-based statistical learning theories---grounded in concepts such as covering numbers, Rademacher complexity, or VC dimension---typically give U-shaped test risk curve bounds as a function of model complexity.
As model complexity (e.g., number of model parameters~$p$) increases relative to the number of training samples~$N_{\mathrm{tr}}$, the test risk initially decreases due to reduced bias and then increases due to high variance (overfitting).
This is illustrated by the red dotted curve in Fig.~\ref{fig:double_descent_concept}.
This conventional bias--variance trade-off implies that overparameterized models ($p > N_{\mathrm{tr}}$) should perform poorly.

\begin{figure}[b]
    \centering
    \includegraphics[width=7cm]{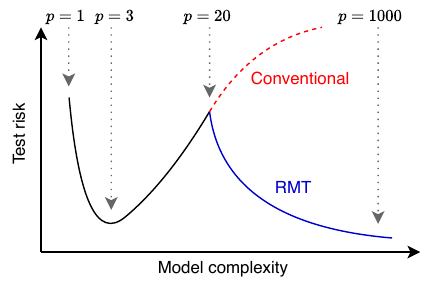}
    \caption{
        Conceptual illustration of double descent. The test risk is plotted against model complexity (e.g., number of model parameters $p$).
        The \textbf{red dotted curve} represents the bias--variance trade-off, where the test risk increases when $p > N_{\mathrm{tr}}$, as suggested by conventional capacity-based theories.
        The \textbf{blue solid curve} illustrates double descent, where the test risk decreases again beyond the interpolation threshold ($p = N_{\mathrm{tr}}$), captured by random matrix theory (RMT).
        The values $p \in \{1,3,20,1000\}$ correspond to the polynomial regression examples in Fig.~\ref{fig:double_descent_fit}.
    }
    \label{fig:double_descent_concept}
\end{figure}

\begin{figure*}[t]
    \centering
    \includegraphics[width=17cm]{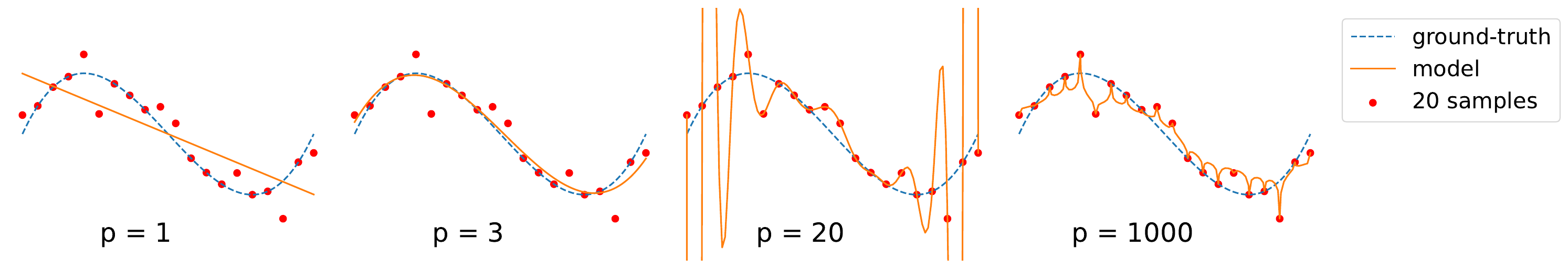}
    \caption{
        \textbf{Example of double descent (adapted from~\cite{nakkiran2019windows}):} Classical polynomial regression with varying numbers of parameters $p$ fitted to a training dataset of $N_{\mathrm{tr}} = 20$ samples. The training dataset is generated from a polynomial of degree $p_* = 3$ with added noise. As $p$ increases beyond $N_{\mathrm{tr}}$, the model fits the training dataset perfectly (interpolation), yet the test risk decreases again, demonstrating double descent. When $p = 1$, the model is underparameterized and underfits; at $p = 3$, it fits well due to correct specification; at $p = 20$, it interpolates the training dataset but with high variance; and at $p = 1000$, it again achieves low test risk despite being highly overparameterized.
    }\label{fig:double_descent_fit}
\end{figure*}

However, as computational resources have grown, it has become feasible to train large machine learning models operating in the overparameterized regime ($p > N_{\mathrm{tr}}$), where conventional theories predict failure due to overfitting.
Surprisingly, beyond the interpolation threshold ($p = N_{\mathrm{tr}}$), the models often regain generalization ability and sometimes outperform their underparameterized counterparts.
Since the test risk curve exhibits two distinct descending regions---one before and one after a peak (see the blue solid curve in Fig.~\ref{fig:double_descent_concept})---this phenomenon is known as \emph{double descent}~\cite{belkin2019reconciling,belkin2020two,nakkiran2021deep}.
Fig.~\ref{fig:double_descent_fit} illustrates how the model behavior changes as the number of model parameters $p$ increases in a simple polynomial regression task.
Double descent appears in diverse model classes, including linear regression~\cite{schaeffer2023double,rocks2022memorizing,hastie2022surprises,bach2024highdimensional,atanasov2024scaling}, kernel methods~\cite{canatar2021spectral,MisiakiewiczSaeed2024,atanasov2024scaling}, random features~\cite{mei2022generalization,gerace2021generalisation,liao2021random,rocks2022memorizing,schroder2024deterministic,atanasov2024scaling}, deep neural networks~\cite{nakkiran2021deep}, and others~\cite{poggio2020double,holzmuller2023universality}. %~\cite{gedonno2024nodouble}

To explain why overparameterized models can generalize, modern statistical learning theory has moved beyond purely capacity-based generalization bounds.
For instance, PAC-Bayesian frameworks provide data-dependent generalization bounds that can remain informative in the overparameterized regime (compatible with double-descent behavior) and can help explain the observed generalization behavior of overparameterized models~\cite{alquier2024userfriendly,hellstrom2024generalization,lotfi2022pac}, while conventional capacity-based theories often yield vacuous worst-case upper bounds over a hypothesis class.
% In a standard PAC-Bayesian bound, the population risk of a randomized predictor drawn from a data-dependent posterior distribution over hypotheses is controlled in terms of its empirical test error and a complexity term, typically involving the Kullback--Leibler divergence from a data-independent prior to the posterior.
% Depending on the choice of posterior and prior, such bounds can reflect properties of the learned solution---for example, robustness to perturbations, flatness, norm control, or compressibility---rather than the overall capacity of the hypothesis class.
% Nevertheless, limitations of PAC-Bayesian bounds have also been identified~\cite{nagarajan2019uniform,livni2020limitation}.
Bound-based approaches, including PAC-Bayesian approaches, are applicable to various models and are valuable for certifying or explaining generalization, but they are not designed to predict the detailed shape of the test-risk curve.

Alongside these approaches, tools from random matrix theory (RMT) and the replica method from statistical mechanics have elucidated the mechanisms of double descent and the effect of explicit regularization on the double-descent peak in analytically tractable models such as linear regression, kernel methods, and random features.
In particular, RMT and the replica method allow us to analyze the \emph{typical} behavior of the test risk in the high-dimensional regime ($p, N_{\mathrm{tr}} \to \infty$ with $p/N_{\mathrm{tr}} \to \gamma \in (0,\infty)$), thereby reproducing double-descent curves that cannot be captured by purely capacity-based worst-case bounds.

We emphasize two points regarding double descent.
First, double descent is not intrinsically a problem to be ``solved'', but rather a phenomenon to be understood in order to make principled decisions in model design and training. Indeed, even without explicit regularization, some models can still generalize better than their regularized counterparts in highly overparameterized regimes despite the higher peak~\cite{nakkiran2021optimal,liang2020just}.
Second, double descent is fundamentally a continuous phenomenon governed by the interplay among continuous hyperparameters, such as the ratio of feature dimension to sample size, the regularization parameter, and the magnitude of noise in the data~\cite{nakkiran2021deep,liu2023understanding}. %  and the learning rate

A critical aspect of the double-descent curve is the relative depth of the two minima surrounding the double-descent peak. In some tasks, the test risk in the highly overparameterized regime (the ``second dip'') is lower than the best test risk achieved in the underparameterized regime, providing a strong motivation for using large models.
Such models can perfectly interpolate noisy training data without sacrificing generalization---a phenomenon known as \emph{benign overfitting}~\cite{bartlett2020benign,mallinar2207benign}.
However, the existence of a double-descent peak does not guarantee that the second dip is lower than the first~\cite{hastie2022surprises}.
In CML, whether the second dip is lower than the first heavily depends on the relationship between the data-generating process and the model specification.
While well-specified linear models with isotropic features fail to achieve a second dip lower than the first, complex setups---such as misspecified models, latent space models~\cite{hastie2022surprises}, random projection models~\cite{bach2024highdimensional}, and random feature models~\cite{liao2021random}---can reproduce a second dip that is lower than the first.
Investigating whether QML models exhibit this ordering of the two minima is therefore as important as studying the double-descent peak itself.
However, accurately quantifying the second dip requires more fine-grained analysis than the double-descent peak.
Consequently, this work leaves the detailed analysis of the second dip for future research and focuses primarily on characterizing the double-descent peak itself. In particular, we aim to quantify how explicit regularization suppresses this peak.

\subsection{Generalization and Double Descent in Quantum Machine Learning}
\begin{figure*}[t]
    \centering
    \begin{tikzpicture}[
        box/.style={rectangle, draw=blue!60, thick, fill=blue!5, text width=5.5cm, align=center, rounded corners=3pt, minimum height=1.7cm, font=\small},
        tool/.style={rectangle, draw=gray!80, thick, fill=gray!10, text width=5cm, align=center, rounded corners=3pt, minimum height=1.7cm, font=\small},
        result/.style={rectangle, draw=red!60, thick, fill=red!5, text width=4cm, align=center, rounded corners=3pt, minimum height=1.7cm, font=\small},
        arrow/.style={-{Stealth[scale=1.2]}, thick, draw=black!70},
        label_text/.style={font=\footnotesize, align=center},
        dash_blue/.style={rectangle, draw=blue!80, dashed, thick, inner sep=4mm, rounded corners=5pt},
        dash_gray/.style={rectangle, draw=gray!90, dashed, thick, inner sep=4mm, rounded corners=5pt},
        dash_red/.style={rectangle, draw=red!80, dashed, thick, inner sep=4mm, rounded corners=5pt}
    ]

    % --- Column 1: Setup (Left) ---
    \node (qmap) at (0, 0) [box] {
        \textbf{Effective Quantum Feature Map}\\
        $\bs{u} \mapsto \bs{r}(\bs{u}) \in \bbR^{p\leq 4^n}$\\
        (Lipschitz Continuous)\\
        \textit{Section~\ref{subsec:feature_vector_generating_process}, \ref{subsec:qkm_setting}}
    };
    \node (risk) at (0, -2) [box] {
        \textbf{Bias--Variance Decomposition}\\
        $R_{\lambda, \hat{\Sigma}} = \mathcal{B}_{\lambda, \hat{\Sigma}} + \mathcal{V}_{\lambda, \hat{\Sigma}} + \sigma^2$\\
        \textit{Section~\ref{subsec:test_risk_decomposition}}
    };

    % --- Column 2: RMT Tools (Middle) ---
    \node (rmt) at (6.5, 0) [tool] {
        \textbf{Random Matrix Theory}\\
        DE of Resolvent (\textit{Prop.~\ref{prop:resolvent_deterministic_equivalent}}~\cite{louart2021spectral})\\
        $\lambda(\hat{\Sigma} + \lambda I_p)^{-1} \asymp_F \kappa_\lambda(\Sigma + \kappa_\lambda I_p)^{-1}$ \\
        \textit{Section~\ref{subsec:DE_intro}}
    };
    \node (deriv) at (6.5, -2) [tool] {
        \textbf{Derivative Trick \& Vitali's Theorem}\\
        DE of Second-Order Resolvents\\
        \textit{Appendix~\ref{appdx:sec:derivative_method}}
    };

    % --- Column 3: Main Results (Right) ---
    \node (main) at (12.4, -1) [result] {
        \textbf{Test-Risk DE}\\
        $R_{\lambda, \hat{\Sigma}} \asymp R^{\mathrm{DE}}_{\lambda,\Sigma}$\\
        \textit{Theorem~\ref{thm:test_risk_DE}, Section~\ref{subsec:asymptotic_test_risk}}
    };

    % --- Background boxes for grouping ---
    \begin{scope}[on background layer]
        \node (group1) [dash_blue, fit=(qmap)(risk), label={[font=\small\bfseries, text=blue!80]above:Problem Setup (Section \ref{sec:setting})}, inner sep=1mm] {};
        \node (group2) [dash_gray, fit=(rmt)(deriv), label={[font=\small\bfseries, text=gray!90]above:Mathematical Tools (Section \ref{sec:DE_test_risk})}, inner sep=1mm] {};
        \node (group3) [dash_red, fit=(main), label={[font=\small\bfseries, text=red!80]above:Main Result (Section \ref{sec:DE_test_risk})}, inner sep=1mm] {};
    \end{scope}

    % --- Draw Arrows between the dashed group boundaries ---
    \draw [arrow] (group1.east) -- node[above, label_text] {} node[below, label_text] {} (group2.west);
    \draw [arrow] (group2.east) -- node[above, label_text] {} node[below, label_text] {} (group3.west);

    \end{tikzpicture}
    \caption{
        \textbf{Logical flow of the theoretical framework.} In the problem setup (blue dashed box), we first show that the quantum feature map produces concentrated feature vectors due to its Lipschitz continuity. Then, we consider the effective feature map and decompose the test risk into bias and variance terms, which involve higher-order resolvents.
        In the mathematical tools section (gray dashed box), concentrated feature vectors permit the use of random matrix theory (RMT) to derive the deterministic equivalent (DE) of the first-order resolvent. We then apply the derivative trick with Vitali's theorem to evaluate the DEs of the higher-order terms in the bias and variance.
        Integrating these tools yields our main theoretical contribution (red dashed box): the test-risk DE $R_{\lambda, \Sigma}^{\mathrm{DE}}$ of QKRR. This formula explicitly reveals the mechanism of the double descent and characterizes the double-descent peak.
    }
    \label{fig:logical_flow}
\end{figure*}

Double descent is increasingly well understood in CML, but its manifestation and mechanisms in QML remain less explored.
Early studies on generalization performance in QML~\cite{caro2021encodingdependent,caro2022generalization,caro2022outofdistribution,banchi2024statistical,bu2023effects,du2022efficient,du2023problemdependent,khanal2024generalization} largely relied on conventional theories.
As in the classical domain, these capacity-based frameworks give standard U-shaped risk curve bounds, known as the bias--variance trade-off.
Indeed, some works have observed U-shaped curves in quantum settings~\cite{du2023problemdependent,strashko2022generalization}.

However, emerging numerical evidence suggests that this picture is incomplete. In particular, certain overparameterized quantum models show better generalization performance~\cite{larocca2021theory,peters2023generalization,haug2024generalization}, and Gil-Fuster \emph{et al.}~\cite{gil-fuster2024understanding} highlighted the inadequacy of using conventional theories to explain generalization in quantum machine learning, mirroring similar findings in CML~\cite{zhang2021understanding}.
Motivated by these observations, PAC-Bayesian generalization bounds have been derived for quantum models~\cite{pablo2026pacbayes}.
Moreover, double-descent peaks have also been observed in quantum kernel methods (QKMs)~\cite{kempkes2025double} and quantum reservoir computing~\cite{hu2024generalization}; these settings were analyzed using RMT and the replica method, respectively.

In this work, we focus on non-variational supervised QKMs~\cite{schuld2019quantum,schuld2021supervised,havlicek2019supervised,tanner2026nonvariational}, where the quantum circuit defines a fixed feature map and the subsequent learning problem is solved classically.
Our work aims to bridge the following gaps in the current literature on generalization in QKMs:

\paragraph{Regularization\,:} Ref.~\cite{kempkes2025double} shows the existence of the double-descent peak in quantum kernel \textit{ridgeless} regression (i.e., without explicit regularization: $\lambda = 0$) based on RMT and provides supporting numerical experiments.
While the authors theoretically elucidate the mechanism of the double-descent peak, they do not provide an explicit formula for predicting the test risk.
In contrast, our work complements the ridgeless analysis of Ref.~\cite{kempkes2025double} by considering quantum kernel \textit{ridge} regression (QKRR) (i.e., with explicit regularization: $\lambda > 0$).
We derive an explicit test-risk formula for QKRR that accounts for the interplay between the effective feature dimension $p \leq 4^n$, the number of training samples $N_{\mathrm{tr}}$, the population covariance matrix $\Sigma$, and the fixed regularization parameter $\lambda$.
This allows us to quantify how explicit regularization modifies the magnitude of the double-descent peak.

\paragraph{Mathematical Rigor\,:} Ref.~\cite{canatar2023bandwidth} also gives an explicit test-risk formula for QKRR using a result based on the replica method from statistical physics~\cite{canatar2021spectral}, but this method relies on heuristic assumptions regarding replica symmetry.
In contrast, we obtain a test-risk formula for QKRR by leveraging mathematically provable results on \emph{deterministic equivalent} (DE) regarding concentrated feature vectors from RMT~\cite{louart2021spectral}, exploiting the fact that quantum feature maps realized via parameterized quantum circuits (PQCs) are Lipschitz continuous~\cite{kempkes2025double}.
Moreover, whereas Ref.~\cite{canatar2023bandwidth} investigates the impact of the bandwidth of simple quantum kernels on spectral decay (i.e., the rate at which the eigenvalues decrease) and generalization, our analysis explicitly investigates double descent and the effect of explicit regularization on the double-descent peak, and provides clear numerical validation.

Understanding the precise shape of the test risk curve in QML is not merely of academic interest: it dictates the optimal strategy for model design.
If the curve is U-shaped, one must carefully limit model complexity. Conversely, if it exhibits double descent, increasing the model complexity beyond the interpolation threshold may yield superior performance.

\subsection{Our Contributions}
Our key contributions are summarized as follows:
\begin{enumerate}[leftmargin=15pt, label=\textbf{\arabic*.}, itemsep=5pt]
    \item \textbf{Rigorous QML-to-RMT Framework (Section~\ref{sec:setting}):}
    We establish a rigorous mathematical bridge that connects QML models to RMT.
    Inspired by the ``linearly independent dimension'' ($p \leq 4^n$) introduced in Ref.~\cite{kempkes2025double} (which we refer to as the \textit{effective feature dimension}), we explicitly formalize an \textit{effective quantum feature map} based on the eigenfunctions of the quantum kernel.
    Using a result from Ref.~\cite{kempkes2025double}, we prove that this map is norm-bounded and Lipschitz continuous, as required to apply the RMT results.
    \item \textbf{Test-Risk DE for QKRR (Section~\ref{sec:DE_test_risk}):}
    Addressing the limitations of ridgeless ($\lambda = 0$) analyses in Ref.~\cite{kempkes2025double}, we derive a rigorous deterministic-equivalent formula for the test risk of QKRR at each fixed $\lambda > 0$ using RMT (Theorem~\ref{thm:test_risk_DE}). This captures the double descent and explicitly quantifies how the regularization parameter~$\lambda$ suppresses the magnitude of the double-descent peak.
    \item \textbf{Numerical Validation (Section~\ref{sec:experiment}):}
    While our theoretical result becomes exact in the high-dimensional limit, we provide numerical validation demonstrating that the test-risk DE can capture the behavior of the empirical test error even for small quantum models (e.g., 3- and 5-qubit models).
\end{enumerate}

The logical flow of our theoretical framework, from the problem setup to the main asymptotic results, is conceptually summarized in Fig.~\ref{fig:logical_flow}.
Guided by this flow, the remainder of the paper is organized as follows.
Section~\ref{sec:setting} formulates the problem setup and introduces the effective quantum feature map.
Section~\ref{sec:DE_test_risk} applies random matrix theory to derive a deterministic equivalent of the test risk of QKRR.
Section~\ref{sec:experiment} provides numerical validation of our theoretical predictions.
Finally, Section~\ref{sec:conclusion} summarizes our findings and outlines future research directions.
The appendices contain a table of notations (Table~\ref{tab:notation_table}) and supplementary details supporting the main text.
\section{Problem Setup}\label{sec:setting}
In this section, we formalize the mathematical framework.
In Section~\ref{subsec:feature_vector_generating_process}, we specify the feature-vector generating process using Lipschitz maps, which can capture quantum feature maps.
In Section~\ref{subsec:ridge_regression}, we review classical ridge regression in its primal formulation.
In Section~\ref{subsec:kernel_regression}, we present the kernel formulation of ridge regression.
In Section~\ref{subsec:eigendecomposition}, we present the eigendecomposition of the kernel function to simplify the analysis of the test risk.
In Section~\ref{subsec:qkm_setting}, we explicitly map the QKRR model into the classical framework, showing that quantum feature maps satisfy the Lipschitz continuity condition.
Finally, in Section~\ref{subsec:test_risk_decomposition}, we define the test risk and perform a bias--variance decomposition, setting the stage for analyzing the double-descent peak.

\subsection{Feature-Vector Generating Process}\label{subsec:feature_vector_generating_process}
\vspace{-7pt}
\begin{figure}[htb]
    \centering
    \begin{tikzpicture}[>=stealth, auto, thick]
        % Draw sets as ellipses representing mathematical spaces
        \draw[fill=gray!5, draw=black!80] (0,0) ellipse (1 and 1);
        \draw[fill=gray!5, draw=black!80] (3,0) ellipse (1 and 1);
        \draw[fill=gray!5, draw=black!80] (6,0) ellipse (1 and 1);

        % Titles for the spaces
        \node[align=center] at (0, 1.8) {Latent Space\\ $\bbR^q$};
        \node[align=center] at (3, 1.8) {Input Space\\ $\calU \subset \bbR^d$};
        \node[align=center] at (6, 1.8) {Feature Space\\ $\bbR^{p_a}$};

        % Elements (Points within the sets)
        \node[circle, fill, inner sep=1.5pt, label={[align=center]left:$\bs{z}_i$}] (z) at (0, 0.3) {};
        \node[circle, fill, inner sep=1.5pt, label={[align=center]below:$\bs{u}_i$}] (u) at (3, 0.3) {};
        \node[circle, fill, inner sep=1.5pt, label={[align=center]right:$\bx^{a}_i$}] (x) at (6, 0.3) {};

        % Mappings (Arrows between elements)
        \draw[->, shorten >=4pt, shorten <=4pt] (z) to[bend left=30] node[midway, above] {$\chi$} (u);
        \draw[->, shorten >=4pt, shorten <=4pt] (u) to[bend left=30] node[midway, above] {$\phi_{a}$} (x);

        % Composite Mapping
        \draw[->, shorten >=4pt, shorten <=4pt] (z) to[bend right=80] node[midway, below] {$\Omega_a = \phi_{a} \circ \chi$} (x);
    \end{tikzpicture}
    \vspace{-10pt}
    \caption{The two-stage feature-vector generating process. A latent standard Gaussian vector $\bs{z}_i$ is transformed into an input vector $\bs{u}_i$ by the data-generating map \(\chi\). The input vector is then mapped to the ambient feature vector $\bx^{a}_i$ by the ambient feature map \(\phi_a\). Their composition is the feature-generating map \(\Omega_a = \phi_a \circ \chi\).}
    \label{fig:feature_generation}
\end{figure}

In typical machine learning scenarios, including QML, the feature vectors used for training and prediction are not completely arbitrary; rather, they are the result of a structured generating process, as detailed in Chapter~8 of Ref.~\cite{couillet2022random}. To make our mathematical analysis clear, we first explicitly define the generating process for both input vectors and feature vectors before introducing the formal assumptions.

As illustrated in Fig.~\ref{fig:feature_generation}, we model the generation of a feature vector $\bx^{a}_i \in \bbR^{p_a}$ in two consecutive stages:
\begin{enumerate}[leftmargin=15pt]
    \item \textbf{Data Generation ($\chi \colon \bbR^q \to \calU$):} We assume an \textit{input vector}~$\bs{u}_i \in \calU \subset \bbR^d$ is generated from an independent latent vector~$\bs{z}_i \sim \calN(0, I_q)$ via a data-generating map~$\chi \colon \bbR^q \to \calU$ as $\bs{u}_i = \chi(\bs{z}_i)$.
    For instance, the generator in Generative Adversarial Networks (GANs)~\cite{goodfellow2014generative}, which maps latent Gaussian noise to realistic data samples (e.g., images), has been shown to be Lipschitz continuous~\cite{seddik2020random}, serving as a practical example of this stage.
    \item \textbf{Feature Mapping ($\phi_{a} \colon \calU \to \bbR^{p_a}$):} This input vector $\bs{u}_i$ is then mapped into a high-dimensional feature space via a feature map $\phi_{a} \colon \calU \to \bbR^{p_a}$, yielding the \textit{feature vector} $\bx^{a}_i = \phi_{a}(\bs{u}_i) \in \bbR^{p_a}$. In the context of QKRR considered in this paper, $\phi_{a}$ represents a quantum feature map implemented by a PQC (discussed further in Section~\ref{subsec:qkm_setting}).
\end{enumerate}

Consequently, the end-to-end generation of the feature vector from the latent vector can be described by the composite map $\Omega_a := \phi_{a} \circ \chi$, such that $\bx^{a}_i = \Omega_a(\bs{z}_i)$.
If both the data-generating map~$\chi$ and the feature map~$\phi_{a}$ are Lipschitz continuous, their composition~$\Omega_a$ preserves the Lipschitz property.
We call the feature vectors generated in this manner \textit{concentrated feature vectors}.
We formalize this process with the following assumption, adapted from Assumptions~0.4--0.6 in Ref.~\cite{louart2021spectral} (reviewed in Appendix~\ref{appdx:sec:definitions_assumptions_louart2021}):
\begin{assumption}\label{assum:feature_generation}
    Suppose there exist universal constants~$C_1, C_2 > 0$. For any feature dimension~$p_a \in \bbN$ and sample size~$N_{\mathrm{tr}} \in \bbN$, we consider a set of $N_{\mathrm{tr}}$ independent random feature vectors~$\bx^{a}_1, \ldots, \bx^{a}_{N_{\mathrm{tr}}} \in \bbR^{p_a}$. These vectors are generated as
    \begin{align}
        \forall i \in [N_{\mathrm{tr}}], \quad \bx^{a}_i = \Omega_a(\bs{z}_i) \,,
    \end{align}
    subject to the following conditions:
    \begin{itemize}[leftmargin=15pt]
        \item The latent vectors $\bs{z}_1, \ldots, \bs{z}_{N_{\mathrm{tr}}}$ are \iid as $\calN(0, I_q)$ for some $q \in \bbN$.
        \item The composite map $\Omega_a \colon \bbR^q \to \bbR^{p_a}$ (representing $\phi_{a} \circ \chi$) is $C_1$-Lipschitz continuous with respect to the $\ell_2$ norm.
        \item The resulting feature vectors satisfy $\|\E[\bx^{a}_i]\|_2 \leq C_2$.
    \end{itemize}
\end{assumption}
This assumption allows us to capture a wide range of feature vectors and to apply powerful tools from RMT to analyze the test risk of ridge regression in the high-dimensional regime.
For quantum feature maps, the Lipschitz constant might not be universal as required in Assumption~\ref{assum:feature_generation}; however, under Assumption~\ref{assum:exponential_effective_dimension} and the polynomial-depth circuit condition introduced below, it satisfies $C_1 = \order{\poly(\log{p})}$, which is sufficient for RMT results used in this work, as shown in Remark~\ref{remark:lipschitz_constant}.
Thus, we call Assumption~\ref{assum:feature_generation} with this relaxed $C_1$ as Assumption~\ref{assum:feature_generation}(q).
While the theoretical framework in Ref.~\cite{louart2021spectral} permits a different map $\Omega_a^{(i)}$ for each $\bs{z}_i$, we assume a common map $\Omega_a$ for all $i \in [N_{\mathrm{tr}}]$ throughout our analysis for simplicity.
The detailed correspondence between our setting and that of Ref.~\cite{louart2021spectral} is provided in Appendix~\ref{appdx:sec:correspondence_feature_generation}.
We further note that Assumption~0.4 in Ref.~\cite{louart2021spectral} is formulated directly as a concentration-of-measure condition and is therefore more general than the Gaussian-Lipschitz construction imposed in Assumption~\ref{assum:feature_generation}.

\subsection{Ridge Regression in Primal Formulation}\label{subsec:ridge_regression}
Following the setting detailed in Refs.~\cite{MisiakiewiczSaeed2024,defilippis2024dimension}, we consider a standard ridge regression framework involving a training dataset of $N_{\mathrm{tr}}$ \iid pairs, denoted by~$\{(\bs{u}_i, y_i)\}_{i=1}^{N_{\mathrm{tr}}}$, drawn from a joint probability distribution over $\calU \times \bbR$.
We assume that each label~$y_i$ is generated by an unknown target function $f_*$ with additive noise, such that $y_i = f_*(\bs{u}_i) + \vep_i$, where the noise term $\vep_i$ is independent of $\bs{u}_i$ and has zero mean and variance $\sigma^2$.
Our objective is to construct a \textit{predictor}~$\hat{f} \colon \calU \to \bbR$ that minimizes the test risk for a new test sample $(\bs{u}_{\mathrm{ts}}, y_{\mathrm{ts}}) \in \calU \times \bbR$ drawn from the same joint distribution, i.e., $y_{\mathrm{ts}} = f_*(\bs{u}_{\mathrm{ts}}) + \vep_{\mathrm{ts}}$.
The test risk is defined as
\begin{align}
    \E_{\bs{u}_{\mathrm{ts}},\vep_{\mathrm{ts}},\bs{\vep}_{\mathrm{tr}}} \left[ (y_{\mathrm{ts}} - \hat{f}(\bs{u}_{\mathrm{ts}}))^2 \right] \,,
    \label{eq:test_risk1}
\end{align}
where the expectation is taken over the test input $\bs{u}_{\mathrm{ts}}$, the test noise $\vep_{\mathrm{ts}}$, and the training label noise vector $\bs{\vep}_{\mathrm{tr}} = (\vep_1, \dots, \vep_{N_{\mathrm{tr}}})\T$.

Using a feature map $\phi_{a} \colon \calU \to \bbR^{p_a}$ that maps an input vector $\bs{u}$ to a feature vector $\bx^{a} := \phi_{a}(\bs{u})$, we consider a linear model in this feature space,
\begin{align}
    f(\bs{u}) = \bth_a\T \bx^{a} = \bth_a\T \phi_{a}(\bs{u}) \,,
\end{align}
where $\bth_a \in \bbR^{p_a}$ is the parameter vector to be learned from the training dataset.
We assume the target function is also linear in this feature space, such that $f_*(\bs{u}) = \bth_{*,a}\T \bx^{a}$ for a fixed but unknown target vector $\bth_{*,a} \in \bbR^{p_a}$.

We employ a ridge loss function to estimate the target vector $\bth_{*,a}$ from the training dataset. The \textit{estimator}~$\hat{\bth}_a$ is defined as
\begin{align}
    \hat{\bth}_a := \underset{\bth_a \in \bbR^{p_a}}{\arg \min} \left\{ \frac{1}{N_{\mathrm{tr}}}\sum_{i=1}^{N_{\mathrm{tr}}} (y_i - \bth_a\T \bx^{a}_i)^2 + \lambda \|\bth_a\|_2^2 \right\} \,.
\end{align}
The addition of the regularization term $\lambda \norm{\bth_a}_2^2$, where~$\lambda > 0$ is the regularization parameter, helps suppress the double-descent peak.
Defining the training feature matrix as $X_a := [\bx^{a}_1, \dots, \bx^{a}_{N_{\mathrm{tr}}}]\T \in \bbR^{N_{\mathrm{tr}} \times p_a}$, the sample covariance matrix as $\hat{\Sigma}_a := \frac{1}{N_{\mathrm{tr}}} X_a\T X_a \in \bbR^{p_a \times p_a}$, and the label vector as $\by := (y_1, \dots, y_{N_{\mathrm{tr}}})\T$, the estimator is given by
\begin{align}
    \hat{\bth}_a
    = (\hat{\Sigma}_a + \lambda I_{p_a})^{-1} \frac{1}{N_{\mathrm{tr}}} X_a\T\by \,.
    \label{eq:ridge_regression_estimator}
\end{align}

The computational cost of this estimator is dominated by the inversion of the $p_a \times p_a$ matrix $\hat{\Sigma}_a + \lambda I_{p_a}$, which typically scales as $\calO((p_a)^3)$.
This becomes prohibitive when the ambient feature dimension~$p_a$ is very large.
The predictor is given by~$\hat{f}(\bs{u}) := \hat{\bth}_a\T \phi_{a}(\bs{u})$.

\subsection{Ridge Regression in Kernel Formulation}\label{subsec:kernel_regression}
The kernel formulation of ridge regression provides an alternative path to the predictor~$\hat{f}$ that is particularly advantageous when the number of training samples $N_{\mathrm{tr}}$ is smaller than the ambient feature dimension~$p_a$~\cite{shawe2004kernel}.
The kernel function $k \colon \calU \times \calU \to \bbR$ is defined as the inner product in the feature space:
\begin{align}
    k(\bs{u}, \bs{u}') := \phi_{a}(\bs{u})\T \phi_{a}(\bs{u}') \,.
\end{align}
The kernel matrix $K \in \bbR^{N_{\mathrm{tr}} \times N_{\mathrm{tr}}}$ is defined as $K := X_aX_a\T$, whose components are $K_{ij} = k(\bs{u}_i, \bs{u}_j)$.
The kernel formulation hinges on the representer theorem~\cite{scholkopf2001generalized}, which states that the estimator~$\hat{\bth}_a$ can be expressed as a linear combination of the training feature vectors for some parameter~$\hat{\bs{\alpha}} \in \bbR^{N_{\mathrm{tr}}}$:
\begin{align}
    \hat{\bth}_a = \sum_{i=1}^{N_{\mathrm{tr}}} \hat{\bs{\alpha}}_i \bx^{a}_i = X_a\T \hat{\bs{\alpha}} \,.
    \label{eq:dual_estimator_theta1}
\end{align}
On the other hand, using the \textit{push-through identity} $A(I + BA)^{-1} = (I + AB)^{-1}A$ in Eq.~\eqref{eq:ridge_regression_estimator}, the estimator can be rewritten as
\begin{align}
    \hat{\bth}_a = X_a\T(K + N_{\mathrm{tr}}\lambda I_{N_{\mathrm{tr}}})^{-1}\by \,.
    \label{eq:dual_estimator_theta2}
\end{align}
Thus, comparing the two expressions~\eqref{eq:dual_estimator_theta1} and~\eqref{eq:dual_estimator_theta2}, one finds that the \textit{dual estimator}~$\hat{\bs{\alpha}} \in \bbR^{N_{\mathrm{tr}}}$ is given by
\begin{align}
    \hat{\bs{\alpha}} = (K + N_{\mathrm{tr}}\lambda I_{N_{\mathrm{tr}}})^{-1}\by \,.
    \label{eq:dual_estimator}
\end{align}
The computational cost to obtain this dual estimator is dominated by the inversion of the $N_{\mathrm{tr}} \times N_{\mathrm{tr}}$ matrix $K + N_{\mathrm{tr}} \lambda I_{N_{\mathrm{tr}}}$, which scales as~$\calO(N_{\mathrm{tr}}^3)$.
This is more efficient than the primal estimator when $N_{\mathrm{tr}} < p_a$.

Once $\hat{\bs{\alpha}}$ is found, a prediction for a new test sample~$\bs{u}_{\mathrm{ts}}$ is made not by computing the estimator $\hat{\bth}_a$ explicitly, but by leveraging the following prediction formula:
\begin{align}
    \hat{f}(\bs{u}_{\mathrm{ts}}) = (X_a\T \hat{\bs{\alpha}})\T \phi_{a}(\bs{u}_{\mathrm{ts}}) = \sum_{i=1}^{N_{\mathrm{tr}}} \hat{\bs{\alpha}}_i k(\bs{u}_i, \bs{u}_{\mathrm{ts}}) \,.
    \label{eq:kernel_prediction}
\end{align}
Crucially, both the training (solving for $\hat{\bs{\alpha}}$) and prediction phases depend only on the kernel function $k$.
This allows us to use kernel functions that implicitly define high-dimensional or even infinite-dimensional feature spaces without explicitly computing the feature vectors~$\bx^{a} = \phi_{a}(\bs{u})$ themselves, enabling efficient learning in complex spaces. % (e.g., Gaussian kernels)

\subsection{Eigendecomposition of the Kernel Function}\label{subsec:eigendecomposition}
Let $\mu$ be the marginal distribution of the input vector $\bs{u}$.
We assume that the ambient feature map satisfies $\phi_a \in L^2(\mu\,;\,\bbR^{p_a})$.

We define the population covariance matrix of the ambient feature vector
as~\footnote{Precisely, $\Sigma_a$ is the uncentered second-moment matrix. We loosely call it the population covariance matrix.}
\begin{align}
    \Sigma_a
    :=
    \int_{\calU}
    \phi_a(\bs{u})\phi_a(\bs{u})\T
    \ind\mu(\bs{u})
    \in
    \bbR^{p_a\times p_a} \,.
    \label{eq:ambient_population_covariance}
\end{align}
This matrix is symmetric positive semidefinite. Let its eigendecomposition be
\begin{align}
    \Sigma_a
    =
    V_a\Lambda_aV_a\T
    =
    \sum_{i=1}^{p_a} \xi_i v_i v_i\T \,,
    \label{eq:ambient_covariance_eigendecomposition}
\end{align}
where $\Lambda_a := \diag(\xi_1, \ldots, \xi_{p_a}) \in \bbR^{p_a \times p_a}$ is a diagonal matrix of eigenvalues and $V_a := (v_1, \ldots, v_{p_a}) \in \bbR^{p_a \times p_a}$ is the orthonormal matrix of eigenvectors.
Let $p := \rank(\Sigma_a)$. We order the positive eigenvalues as $\xi_1 \geq \cdots \geq \xi_p > 0$; if $p < p_a$, the remaining eigenvalues satisfy $\xi_{p+1} = \cdots = \xi_{p_a} = 0$.
Then, we define the following functions:
\begin{align}
    e_i(\bs{u}) := \phi_a(\bs{u})\T v_i \quad (i = 1,\dots,p_a) \,.
    \label{eq:unnormalized_eigenfunctions}
\end{align}
These functions satisfy the eigenvalue equation for the integral operator $\calT_k$ associated with the kernel, as verified by the following calculation:
\begin{align}
    (\calT_k e_i)(\bs{u})
    &:= \int_\calU k(\bs{u},\bs{u}')e_i(\bs{u}') \ind\mu(\bs{u}') \nonumber\\
    &= \phi_{a}(\bs{u})\T \left( \int_\calU \phi_{a}(\bs{u}')\phi_{a}(\bs{u}')\T \ind\mu(\bs{u}') \right) v_i \nonumber\\
    &= \phi_{a}(\bs{u})\T \Sigma_a v_i \nonumber\\
    &= \xi_i e_i(\bs{u}) \,.
\end{align}
These functions are orthogonal with respect to $\mu$:
\begin{align}
    \int_{\calU}
    e_i(\bs{u})e_j(\bs{u})
    \ind\mu(\bs{u})
    =
    v_i\T\Sigma_a v_j
    =
    \xi_i\delta_{ij} \,.
    \label{eq:ei_inner_products}
\end{align}
For every $i > p$, Eq.~\eqref{eq:ei_inner_products} gives $\int_{\calU} e_i(\bs{u})^2 \ind\mu(\bs{u}) = 0$.
Hence, $e_i=0$ $\mu$-almost everywhere for every $i>p$.
Since there are only finitely many such indices, there exists a measurable set $A\subseteq\calU$ with $\mu(A) = 1$ such that $e_i(\bs{u}) = 0$ for every $\bs{u} \in A$ and every $i > p$.
On the other hand, for each $i \le p$, $e_i$ is a nonzero eigenfunction of the kernel integral operator $\calT_k$ with eigenvalue $\xi_i$.
For $i = 1,\ldots,p$, define the normalized eigenfunctions $\psi_i(\bs{u}) := e_i(\bs{u})/\sqrt{\xi_i}$.
Eq.~\eqref{eq:ei_inner_products} shows that these functions are orthonormal with respect to $\mu$.

Let
\begin{align}
    V
    &:=
    (v_1,\ldots,v_p)
    \in
    \bbR^{p_a\times p} \,,
    \\
    \Lambda
    &:=
    \diag(\xi_1,\ldots,\xi_p)
    \in
    \bbR^{p\times p} \,,
    \\
    \psi(\bs{u})
    &:=
    [\psi_1(\bs{u}),\ldots,\psi_p(\bs{u})]\T \,.
\end{align}
We define the following feature map by
\begin{align}
    \phi(\bs{u})
    &:=
    V\T\phi_a(\bs{u})
    \nonumber\\
    &=
    [e_1(\bs{u}),\ldots,e_p(\bs{u})]\T
    \nonumber\\
    &=
    \Lambda^{\frac12}\psi(\bs{u}) \in \bbR^p \,.
    \label{eq:effective_feature_map}
\end{align}
Since $\{v_i\}_{i=1}^{p_a}$ is an orthonormal basis and $e_i(\bs{u}) = 0$ for every $\bs{u} \in A$ and every $i>p$, we have
\begin{align}
    \phi_a(\bs{u})
    =
    \sum_{i=1}^{p_a} e_i(\bs{u})v_i
    =
    \sum_{i=1}^{p} e_i(\bs{u})v_i
    =
    V\phi(\bs{u})
    \label{eq:ambient_effective_feature_relation}
\end{align}
for every $\bs{u}\in A$, and hence $\phi_a(\bs{u})=V\phi(\bs{u})$ $\mu$-almost everywhere.
Consequently, the kernel has the exact finite-rank spectral representation
\begin{align}
    k(\bs{u},\bs{u}')
    &=
    \phi_a(\bs{u})\T\phi_a(\bs{u}')
    \nonumber\\
    &=
    \phi(\bs{u})\T\phi(\bs{u}')
    \nonumber\\
    &=
    \sum_{i=1}^{p}
    \xi_i
    \psi_i(\bs{u})
    \psi_i(\bs{u}')
    \label{eq:kernel_spectral_expansion}
\end{align}
for $\mu\otimes\mu$-almost every $(\bs{u},\bs{u}')\in\calU\times\calU$.
If the restriction of $\phi_a$ to $\operatorname{supp}(\mu)$ is continuous, then Eq.~\eqref{eq:ambient_effective_feature_relation} holds pointwise on $\operatorname{supp}(\mu)$, and Eq.~\eqref{eq:kernel_spectral_expansion} holds pointwise on $\operatorname{supp}(\mu)\times\operatorname{supp}(\mu)$.
This applies, in particular, to the Lipschitz quantum feature maps considered below.

We refer to $\phi(\cdot)$ as the \textit{effective feature map}, $p$ as the \textit{effective feature dimension}, and $p_a$ as the \textit{ambient feature dimension}.

Due to the orthogonality of the eigenfunctions $\{e_i\}_{i=1}^p$, the population covariance matrix of the effective feature vector is diagonalized as
\begin{align}
    \Sigma
    :=
    \E_{\calU} [\phi(\bs{u})\phi(\bs{u})\T]
    =
    V\T \Sigma_a V
    =
    \Lambda \succ 0 \,.
    \label{eq:effective_population_covariance}
\end{align}

For any ambient parameter vector $\bth_a\in\bbR^{p_a}$, we define the effective parameter vector $\bth := V\T \bth_a$.
Equation~\eqref{eq:ambient_effective_feature_relation} then gives $f(\bs{u})=\bth_a\T\phi_a(\bs{u})=\bth\T\phi(\bs{u})$ for $\mu$-almost every $\bs{u}\in\calU$.
Similarly, defining the effective target vector by $\bth_* := V\T\bth_{*,a}$, we have $f_*(\bs{u})= \bth_{*,a}\T \phi_a(\bs{u}) = \bth_*\T \phi(\bs{u})$ for $\mu$-almost every $\bs{u} \in \calU$.

For the effective representation, we define
\begin{align}
    \Omega
    &:=
    V\T\Omega_a \,,
    \\
    \bx
    &:=
    \phi(\bs{u}) \in \bbR^p \,,
    \quad
    \bx_i
    = V\T\bx_i^a = \Omega(\bs{z}_i) \,,
    \\
    X
    &:=
    [\bx_1,\ldots,\bx_{N_{\mathrm{tr}}}]\T
    =
    X_aV \in \bbR^{N_{\mathrm{tr}}\times p} \,,
    \\
    \hat{\Sigma}
    &:=
    \frac{1}{N_{\mathrm{tr}}}
    X\T X
    =
    V\T \hat{\Sigma}_aV \in \bbR^{p\times p} \,,
    \\
    K
    &:=
    XX\T
    \,.
    \label{eq:effective_sample_quantities}
\end{align}
For any finite i.i.d.\ training sample from $\mu$, all training inputs belong to $A$ almost surely. Therefore, $X_a=XV\T$ and $K = XX\T = X_aX_a\T$ almost surely.
Consequently, the ambient and effective representations produce the same random Gram matrix, the same kernel-ridge predictions, and the same population test risk almost surely.
In this precise distributional sense, the ambient- and effective-space ridge-regression formulations are equivalent under i.i.d.\ sampling from $\mu$.
Hence, for the distributional analysis considered here, we can restrict our attention to the effective feature space without loss of generality.
Since $V\T$ has operator norm one, the effective feature-generating map $\Omega$ inherits the Lipschitz and mean-norm bounds in Assumption~\ref{assum:feature_generation}.
The corresponding effective-space ridge estimator is
\begin{align}
    \hat{\bth}
    :=
    (\hat{\Sigma} + \lambda I_p)^{-1}
    \frac{1}{N_{\mathrm{tr}}} X\T\by \,.
\end{align}

The primary motivation for introducing this effective space is twofold.
First, it ensures that the population covariance matrix $\Lambda$ is diagonal and positive definite, which simplifies the theoretical analysis.
Second, it clarifies that the interpolation threshold is fundamentally determined by the effective feature dimension~$p$ rather than the ambient feature dimension~$p_a$.
For a discussion of how the kernel formulation reflects overparameterization ($p > N_{\mathrm{tr}}$) despite involving only $N_{\mathrm{tr}}$ degrees of freedom, see Appendix~\ref{appdx:sec:kernel_overparameterization}.

For simplicity, we assume that $f_*$, viewed as an element of $L^2(\mu)$, belongs to $\mathrm{span}\{\psi_1,\psi_2,\ldots,\psi_p\}$, so that
\begin{align}
    f_*(\bs{u})
    &= \bs{\beta}_*\T \psi(\bs{u})
    = (\Lambda^{-\frac12} \bs{\beta}_*)\T \phi(\bs{u})
    \quad \mu\text{-a.e.} \,.
\end{align}
We refer to $\bs{\beta}_*$ as the projected target vector.
Since $f_*(\bs{u})=\bth_*\T\phi(\bs{u})$ $\mu$-almost everywhere, the projected target vector is related to the effective target vector as $\bs{\beta}_* = \Lambda^{\frac12}\bth_*$.
Even though the target function $f_*$ for some datasets may not be fully contained in the span of $\{\psi_i\}_{i=1}^p$, the unrepresented part of $f_*$ orthogonal to the feature space acts as additional noise (model misspecification error), which gets absorbed into the constant $\sigma^2$ term in the risk with some additional technical assumptions~\cite{canatar2021spectral}.

\begin{table*}[ht]
    \centering
    \caption{Correspondence between classical and quantum ridge regression frameworks.}
    \label{tab:qml_classical_correspondence}
    \renewcommand{\arraystretch}{1.4}
    \begin{tabular}{@{}lcc@{}}
        \toprule
        \textbf{Framework} & \textbf{Classical (Section~\ref{subsec:ridge_regression} \& \ref{subsec:kernel_regression})} & \textbf{Quantum (Section~\ref{subsec:qkm_setting})} \\
        \midrule
        % \multicolumn{3}{@{}l}{\textit{\textbf{Ambient Space Representation}}} \\
        Ambient feature dimension & $p_a$ & $4^n$ \\
        Ambient feature vector & $\bx^{a} := \phi_{a}(\bs{u}) \in \bbR^{p_a}$  & $\bs{r}_a(\bs{u}) \in \bbR^{4^n}$ (Pauli coeff. of $\rho(\bs{u})$) \\
        Ambient parameter vector & $\bth_a \in \bbR^{p_a}$ & $\bs{w}_a \in \bbR^{4^n}$ (Pauli coeff. of $O$) \\
        Ambient target vector & $\bth_{*,a} \in \bbR^{p_a}$ & $\bs{w}_{*,a} \in \bbR^{4^n}$ (Pauli coeff. of $O_*$) \\
        Model $f(\bs{u})$ / Target $f_*(\bs{u})$ & $\bth_a\T \phi_{a}(\bs{u})$ / $\bth_{*,a}\T \phi_{a}(\bs{u})$ & $\Tr[\rho(\bs{u})O] = \bs{w}_a\T \bs{r}_a(\bs{u})$ / $\Tr[\rho(\bs{u})O_*] = \bs{w}_{*,a}\T \bs{r}_a(\bs{u})$ \\

        \midrule
        % \multicolumn{3}{@{}l}{\textit{\textbf{Effective Space Representation}}} \\
        Effective feature dimension & $p \leq p_a$ & $p \leq 4^n$ for HEA, $p \leq 3^n$ for TPA \\
        Effective feature vector & $\phi(\bs{u}) = V\T \phi_{a}(\bs{u}) \in \bbR^p$ & $\bs{r}(\bs{u}) = V_r\T \bs{r}_a(\bs{u}) \in \bbR^p$ \\
        Effective parameter vector & $\bth = V\T \bth_a \in \bbR^p$ & $\bs{w} = V_r\T \bs{w}_a \in \bbR^p$ \\
        Effective target vector & $\bth_* = V\T \bth_{*,a} \in \bbR^p$ & $\bs{w}_* = V_r\T \bs{w}_{*,a} \in \bbR^p$ \\
        Effective model $f(\bs{u})$ / Target $f_*(\bs{u})$ & $\bth\T \phi(\bs{u})$ / $\bth_*\T \phi(\bs{u})$ & $\bs{w}\T \bs{r}(\bs{u})$ / $\bs{w}_*\T \bs{r}(\bs{u})$ \\

        \midrule
        Ambient population covariance $\Sigma_a$ & $\E_{\calU}[\phi_{a}(\bs{u})\phi_{a}(\bs{u})\T] \in \bbR^{p_a \times p_a}$ & $\E_{\calU}[\bs{r}_a(\bs{u})\bs{r}_a(\bs{u})\T] \in \bbR^{4^n \times 4^n}$ \\
        Effective population covariance $\Sigma$ & $\E_{\calU}[\phi(\bs{u})\phi(\bs{u})\T] \in \bbR^{p \times p}$ & $\E_{\calU}[\bs{r}(\bs{u})\bs{r}(\bs{u})\T] \in \bbR^{p \times p}$ \\
        Kernel function $k(\bs{u}, \bs{u}')$ & $\phi_{a}(\bs{u})\T \phi_{a}(\bs{u}') = \phi(\bs{u})\T \phi(\bs{u}')$ & $\Tr[\rho(\bs{u})\rho(\bs{u}')] = \bs{r}_a(\bs{u})\T \bs{r}_a(\bs{u}') = \bs{r}(\bs{u})\T \bs{r}(\bs{u}')$ \\
        \bottomrule
    \end{tabular}
\end{table*}

\subsection{Quantum Kernel Ridge Regression}\label{subsec:qkm_setting}
To apply the classical ridge regression framework to QML, we first express the quantum model in its primal formulation to identify the feature space, and then demonstrate how it is solved using the quantum kernel.

We consider the quantum dataset $\{(\rho(\bs{u}_i), y_i)\}_{i=1}^{N_{\mathrm{tr}}}$, where each pair consists of an $n$-qubit quantum state~$\rho(\bs{u}_i) \in \bbC^{2^n \times 2^n}$ and a corresponding real-valued label $y_i \in \bbR$.
Each quantum state $\rho(\bs{u}_i)$ is generated from a classical input vector $\bs{u}_i \in \calU \subset \bbR^d$ through a data-encoding parameterized quantum circuit $U(\bs{u})$ as
\begin{equation}
    \rho(\bs{u}_i) = U(\bs{u}_i)\dyad{\bs{0}}U\dg(\bs{u}_i)\,.
\end{equation}
The QML model $f(\bs{u})$ is defined as the expectation value of a trainable observable $O$ with respect to $\rho(\bs{u})$:
\begin{equation}
    f(\bs{u}) := \Tr[\rho(\bs{u}) O]\,.
\end{equation}

We assume that for each quantum state $\rho(\bs{u}_i)$, the label $y_i$ is generated from a target function given by the expectation value of some observable $O_*$ with respect to~$\rho(\bs{u}_i)$:
\begin{align}
    y_i = \Tr[\rho(\bs{u}_i)O_*] + \vep_i, \quad \vep_i \stackrel{\iid}{\sim} \calN(0, \sigma^2) \,.
\end{align}

To connect this QML model to the classical primal regression framework with real-valued feature vectors and parameter vectors described in Section~\ref{subsec:ridge_regression}, we express both the quantum state $\rho(\bs{u})$ and the observable $O$ in the normalized Pauli basis $\{P_i\}_{i=1}^{4^n}$ (satisfying $\Tr[P_i P_j] = \delta_{ij}$) with real coefficient vectors $\bs{r}_a(\bs{u}) := (r_i^a(\bs{u}))_{i=1}^{4^n} \in \bbR^{4^n}$ and $\bs{w}_a := (w_i^a)_{i=1}^{4^n} \in \bbR^{4^n}$ as follows:
\begin{align}
    \rho(\bs{u}) = \sum_{i=1}^{4^n} r_i^a(\bs{u}) P_i
    \quad\text{and} \quad
    O = \sum_{i=1}^{4^n} w_i^a P_i \,.
\end{align}
Then, the QML model can be rewritten as
\begin{align}
    f(\bs{u}) = \Tr[\rho(\bs{u}) O] = \sum_{i=1}^{4^n} w_i^a r_i^a(\bs{u}) = \bs{w}_a\T \bs{r}_a(\bs{u}) \,.
\end{align}
We also define $\bs{w}_{*,a} \in \bbR^{4^n}$ as the Pauli coefficients of the unknown observable $O_*$, so that the target function can be expressed as $f_*(\bs{u}) = \Tr[\rho(\bs{u}) O_*] = \bs{w}_{*,a}\T \bs{r}_a(\bs{u})$.
This shows that the quantum model can be written as the regression model introduced in Section~\ref{subsec:ridge_regression}.
Table~\ref{tab:qml_classical_correspondence} summarizes this correspondence between the classical and quantum frameworks.

This feature map $\bs{u} \mapsto \bs{r}_a(\bs{u})$ can be shown to be Lipschitz continuous with respect to the $\ell_2$ norm when using PQCs for data encoding.
The formal proof of this property is provided in Appendix~\ref{appdx:sec:lipschitz_qml_feature_map} based on a result from Ref.~\cite{kempkes2025double}.
The feature vector $\bs{r}_a(\bs{u})$ satisfies $\norm{\E_{\calU}[\bs{r}_a(\bs{u})]}_2 \leq \E_{\calU}[\norm{\bs{r}_a(\bs{u})}_2] = \E_{\calU}[\sqrt{\Tr[\rho(\bs{u})^2]}] = 1$.
These properties align with Assumption~\ref{assum:feature_generation}(q).

In this quantum model, the ambient feature dimension is $p_a=4^n$. For a given input distribution $\mu$, however, the effective feature dimension is
\begin{align}
    p
    :=
    \rank(
    \E_{\bs{u} \sim \mu}[
        \bs{r}_a(\bs{u})\bs{r}_a(\bs{u})\T
    ])
        \leq 4^n.
\end{align}
Thus, the ambient feature dimension is determined by the number of qubits, whereas the effective dimension depends jointly on the quantum feature map $\bs{r}_a(\cdot)$ and the input distribution $\mu$.
In particular, $p = 4^n$ holds if and only if the $4^n$ Pauli-coefficient functions are linearly independent in $L^2(\mu)$, namely,
\begin{align}
    \bs{c}\T\bs{r}_a(\bs{u})=0, \quad \mu\text{-a.e.}
    \quad\Longrightarrow\quad
    \bs{c}=\bs{0}.
\end{align}
The left-hand side of this implication means
\begin{align}
    \mu(\{\bs{u} \in \calU\,:\, \bs{c}\T\bs{r}_a(\bs{u}) \neq 0 \}) = 0 \,.
\end{align}
A sufficiently expressive circuit architecture, such as the Hardware-Efficient Ansatz (HEA) discussed in Section~\ref{sec:experiment}, may permit $p = 4^n$, but does not guarantee it for an arbitrary input distribution $\mu$. Thus, $p \leq 4^n$ for HEA in general.

Conversely, if the quantum circuit is not sufficiently expressive, such as the Tensor Product Ansatz (TPA) consisting only of $R_X$ gates, the entries of $\bs{r}_a(\bs{u})$ are linearly dependent. In this case, as shown in Section~\ref{subsec:eigendecomposition}, we can project $\bs{r}_a(\bs{u})$ onto a lower-dimensional vector~$\bs{r}(\bs{u}) \in \bbR^{p}$ with $p < p_a = 4^n$.
The effective feature dimension for the TPA is $p \leq 3^n$, as we show in Appendix~\ref{appdx:sec:effective_dim_TPA}.

\begin{assumption}[Exponential growth of the effective feature dimension]
    \label{assum:exponential_effective_dimension}
    We consider quantum feature-map families for which the effective feature
    dimension grows at least exponentially with the number of qubits.
    Specifically, we assume that there exists a constant \(c > 1\) such that $p = \Omega(c^n)$.
    Equivalently, this implies
    \begin{align}
        n = \order{\log{p}}.
        \label{eq:qubits_log_effective_dimension}
    \end{align}
\end{assumption}

Specifically, $\bs{r}(\bs{u})$ is obtained via a linear transformation:
\begin{align}
    \bs{r}(\bs{u}) = V_r\T \bs{r}_a(\bs{u}) \,,
\end{align}
where $V_r \in \bbR^{4^n \times p}$ is the matrix of $p$ eigenvectors corresponding to the nonzero eigenvalues of the population covariance matrix~$\Sigma_a = \E_{\calU}[\bs{r}_a(\bs{u})\bs{r}_a(\bs{u})\T]$.
% We also define $V_{r,a}$ as the matrix of all $4^n$ eigenvectors.
We call $\bs{r}(\cdot)$ the \textit{effective quantum feature map}.
As we have established for the effective feature map $\phi(\cdot)$ in Section~\ref{subsec:eigendecomposition}, the entries of $\bs{r}(\bs{u})$ are also orthogonal functions of $\bs{u}$ and are therefore linearly independent.
Moreover, $\bs{r}(\bs{u})$ yields the same kernel as $\bs{r}_a(\bs{u})$.
Because the linear transformation $V_r$ has orthonormal columns, we have
\begin{align}
    \norm{\bs{r}(\bs{u}) - \bs{r}(\bs{v})}_2
    &= \norm*{V_r\T(\bs{r}_a(\bs{u}) - \bs{r}_a(\bs{v}))}_2 \nonumber\\
    &\leq \norm{\bs{r}_a(\bs{u}) - \bs{r}_a(\bs{v})}_2 \,.
\end{align}
Therefore, the effective quantum feature map is also Lipschitz continuous with respect to the $\ell_2$ norm and satisfies $\norm{\E_{\calU}[\bs{r}(\bs{u})]}_2 \leq \E_{\calU}[\norm{\bs{r}(\bs{u})}_2] \leq \E_{\calU}[\norm{\bs{r}_a(\bs{u})}_2] = 1$, thus satisfying Assumption~\ref{assum:feature_generation}(q).

Similarly, defining the effective parameter vector~$\bs{w} := V_r\T \bs{w}_a$ and the effective target vector $\bs{w}_* := V_r\T \bs{w}_{*,a}$, we have $f(\bs{u}) = \bs{w}\T \bs{r}(\bs{u})$ and $f_*(\bs{u}) = \bs{w}_*\T \bs{r}(\bs{u})$.

Finally, computing the estimator~\eqref{eq:ridge_regression_estimator} of the quantum model directly in the $4^n$-dimensional primal space is generally intractable for large $n$.
Instead, quantum kernel methods (QKMs) bypass the explicit computation of the feature vector~$\bs{r}_a(\bs{u})$. Since the Pauli matrices are orthonormal ($\Tr[P_i P_j] = \delta_{ij}$), the dot product in the primal feature space corresponds exactly to the Hilbert-Schmidt inner product of the quantum states. Thus, we compute the kernel directly on the quantum computer as
\begin{align}
    k(\bs{u}, \bs{u}') = \bs{r}_a(\bs{u})\T \bs{r}_a(\bs{u}') = \Tr[\rho(\bs{u})\rho(\bs{u}')] \,.
\end{align}
Using this quantum kernel, we solve for the dual estimator $\hat{\bs{\alpha}}$ classically via Eq.~\eqref{eq:dual_estimator}.
One can construct various quantum kernels tailored to different tasks by choosing different encoding strategies~$U(\bs{u})$, as detailed in Ref.~\cite{schuld2021supervised}.

Since $\norm{\bx}_2^2 = \norm{\bx^{a}}_2^2 = k(\bs{u},\bs{u}) = \Tr[\rho(\bs{u})^2] = 1$ holds for the pure-state fidelity kernel used in this work, the following normalization holds:
\begin{align}
    \Tr[\Sigma]
    = \Tr[\Sigma_a]
    = \E_{\calU}[\norm{\bx^{a}}_2^2]
    = 1.
    \label{eq:pure_state_kernel_normalization}
\end{align}

\subsection{Bias--Variance Decomposition of the Risk}\label{subsec:test_risk_decomposition}
As shown in Sections~\ref{subsec:eigendecomposition} and~\ref{subsec:qkm_setting}, the test risk depends only on the kernel function, which remains invariant under the reduction to the effective feature space. Therefore, we hereafter focus exclusively on this effective feature space.
Henceforth, we only use the unadorned effective feature space quantities $\phi$, $\bx$, $\bth$, $\bth_*$, $X$, $\hat{\Sigma}$, and $\Sigma$, both in the classical and quantum settings for notational simplicity.
This choice of feature map allows us to work with a diagonal and positive definite population covariance matrix~$\Sigma = \Lambda = \diag(\xi_1, \xi_2, \ldots, \xi_p)$.

The ultimate goal of the regression task is to minimize the test risk, which is conditional on the training features:
\begin{align}
    R_{\lambda,\hat{\Sigma}}
    &:= \E_{\bx_{\mathrm{ts}},\vep_{\mathrm{ts}},\bs{\vep}_{\mathrm{tr}}} \left[ (y_{\mathrm{ts}} - \hat{\bth}\T \bx_{\mathrm{ts}})^2 \right] \nonumber\\
    &= \E_{\bs{\vep}_{\mathrm{tr}}}
    [(\bth_* - \hat{\bth})\T \Sigma (\bth_* - \hat{\bth})] + \sigma^2 \,.
    \label{eq:test_risk2}
\end{align}
Here, the expectation is taken over the test feature~$\bx_{\mathrm{ts}}$, the test noise $\vep_{\mathrm{ts}}$ associated with $y_{\mathrm{ts}}$, and the training label noise~$\bs{\vep}_{\mathrm{tr}}$.
Note that the constant term $\sigma^2$ in Eq.~\eqref{eq:test_risk2} represents the irreducible error arising from the test noise $\vep_{\mathrm{ts}}$.
We consider the bias--variance decomposition with respect to the training label noise $\bs{\vep}_{\mathrm{tr}}$, on which the estimator $\hat{\bth}$ depends. By defining the mean vector $\bar{\bth} := \E_{\bs{\vep}_{\mathrm{tr}}}[\hat{\bth}]$, we have
\begin{align}
    R_{\lambda,\hat{\Sigma}}
    &= (\bth_* - \bar{\bth})\T \Sigma (\bth_* - \bar{\bth}) \nonumber\\
    &\quad+ \E_{\bs{\vep}_{\mathrm{tr}}} \left[ (\hat{\bth} - \bar{\bth})\T \Sigma (\hat{\bth} - \bar{\bth}) \right]
    + \sigma^2 \nonumber\\
    &=: \calB_{\lambda,\hat{\Sigma}} + \calV_{\lambda,\hat{\Sigma}} + \sigma^2 \,.
    \label{eq:bias_variance_decomposition}
\end{align}
As derived in Appendix~\ref{appdx:sec:qml_analysis}, the bias and variance terms are given explicitly by
\begin{align}
    \calB_{\lambda,\hat{\Sigma}}
    &= \lambda^2 \bth_*\T (\hat{\Sigma} + \lambda I_p)^{-1} \Sigma (\hat{\Sigma} + \lambda I_p)^{-1} \bth_*\nonumber\\
    &= \lambda^2 \bs{\beta}_*\T \Sigma^{-\frac12} (\hat{\Sigma} + \lambda I_p)^{-1} \Sigma (\hat{\Sigma} + \lambda I_p)^{-1} \Sigma^{-\frac12}\bs{\beta}_* \,, \label{eq:bias_term_in_risk}\\
    \calV_{\lambda,\hat{\Sigma}}
    &= \frac{\sigma^2}{N_{\mathrm{tr}}}
    \Tr[
        \Sigma \hat{\Sigma} (\hat{\Sigma} + \lambda I_p)^{-2}
    ] \,.\label{eq:variance_term_in_risk}
\end{align}
The effective target vector~$\bth_*$ only appears in the bias term.
To derive the deterministic equivalent of the bias term, we impose a standard condition on $\bth_*$~\cite{MisiakiewiczSaeed2024}:
\begin{assumption}[Bounded effective target vector]\label{assum:target_regularity}
    We assume that there exists a constant $C_3 > 0$, independent of $p$, such that
    \begin{align}
        \sup_{p \in \bbN}\|\bth_*\|_2^2 < C_3 \,.
        \label{eq:bounded_effective_target}
    \end{align}
\end{assumption}

As we will see in the next section, the variance term is the origin of the double-descent peak: it becomes large when the regularization parameter $\lambda$ is small near the interpolation threshold, whereas the bias term becomes small in the same regime.
\section{Deterministic-Equivalent Analysis of the Test Risk}\label{sec:DE_test_risk}
Having defined the effective quantum feature map and decomposed the risk, we can apply RMT to analyze the test risk of QKRR.

Since the test risk~\eqref{eq:bias_variance_decomposition} depends on the sample covariance matrix $\hat{\Sigma} := \frac{1}{N_{\mathrm{tr}}} X\T X$, it is a random quantity.

However, in the proportional high-dimensional regime ($p, N_{\mathrm{tr}} \to \infty$ with $p/N_{\mathrm{tr}} \to \gamma \in (0,\infty)$) and under the assumptions stated below, the random fluctuations of the sample covariance matrix $\hat{\Sigma}$ become negligible when observed through the test risk.
In this section, we theoretically characterize the asymptotic behavior of the test risk of QKRR by deriving the deterministic equivalent of the test risk with the help of recent advances from RMT.
Using this equivalent, we can predict the typical behavior of the test risk in the high-dimensional regime without repeatedly retraining the model on independently sampled training datasets.

Although the deterministic equivalents of the test risks of various models have been derived in CML settings~\cite{atanasov2024scaling,mei2022generalization}, they often rely on restrictive assumptions, such as simple Gaussian features~\cite{atanasov2024scaling} or random features~\cite{mei2022generalization}. These assumptions are insufficient for QKRR, where the entries of the feature vectors are nonlinearly dependent due to the encoding by PQCs, and there is no explicit randomization in the feature generation.
In this section, we derive the deterministic equivalent of the test risk for a general class of feature vectors generated by Lipschitz functions of latent standard Gaussian vectors, referred to as \emph{concentrated feature vectors}, as specified in Section~\ref{subsec:feature_vector_generating_process}.
Crucially, this class includes feature vectors generated by quantum feature maps implemented by PQCs.
The derived test-risk DE makes explicit that quantum-specific feature geometry---including effects associated with entanglement and circuit expressibility---is encoded by the population covariance matrix $\Sigma$, while task dependence enters through the alignment of $\bs{\beta}_*$ with its eigendirections.

This section is organized as follows.
In Section~\ref{subsec:DE_intro}, we introduce the notion of deterministic equivalence and state the core result from Ref.~\cite{louart2021spectral} regarding the deterministic equivalent of the resolvent of sample covariance matrices for concentrated feature vectors.
In Section~\ref{subsec:asymptotic_test_risk}, we apply the \emph{derivative trick}~\cite{zavatone-vethlecture,atanasov2024scaling,dobriban2018high} to derive the test-risk DE for QKRR, expressing it as a function of the population covariance matrix $\Sigma$ and the effective regularization ($\kappa_\lambda$ defined in Proposition~\ref{prop:resolvent_deterministic_equivalent}), thereby clarifying the mechanisms behind the double descent phenomenon.

\subsection{Deterministic Equivalents of Random Matrices}\label{subsec:DE_intro}
We begin with the definition of deterministic equivalence for random matrices (``a.s.'' = almost surely).
\begin{definition}[Frobenius-dual deterministic equivalent]
\label{def:frobenius_deterministic_equivalent}
    Let $\{A_p\}$ be a sequence of random $p\times p$ matrices and $\{B_p\}$ a sequence of deterministic $p\times p$ matrices.
    We say that $B_p$ is a \textbf{Frobenius-dual deterministic equivalent} of $A_p$, denoted by
    \begin{align}
        A_p \asymp_F B_p \,,
    \end{align}
    if, for every deterministic matrix sequence $C_p \in \bbR^{p\times p}$ satisfying $\sup_{p \in \bbN}\norm{C_p}_F < \infty$, we have
    \begin{align}
        \lim_{p \to \infty} \Tr[C_p(A_p - B_p)]
        = 0
        \quad \text{a.s.}
        \label{eq:def:frobenius_deterministic_equivalent}
    \end{align}
\end{definition}
For a scalar random sequence $\{a_p\}$ and a deterministic sequence $\{b_p\}$, we say $b_p$ is a deterministic equivalent of $a_p$, denoted by $a_p \asymp b_p$, if $\lim_{p \to \infty} (a_p - b_p) = 0$ a.s.

The double-descent behavior of the test risk is governed by terms involving the resolvent $(\hat{\Sigma}+\lambda I_p)^{-1}$ of the sample covariance matrix.
The following proposition gives a Frobenius-dual deterministic equivalent of the scaled resolvent for concentrated feature vectors generated according to Assumption~\ref{assum:feature_generation}(q), which is satisfied by the effective quantum feature map $\bs{r}(\cdot)$ as discussed in Section~\ref{subsec:qkm_setting}.
This result follows from Theorem~0.9 in Ref.~\cite{louart2021spectral} by setting $z = -\lambda < 0$ and $\Sigma_i := \E[\bx_i\bx_i\T] = \Sigma$ in the resolvent formula; Appendix~\ref{appdx:sec:theorem_0.9_to_deterministic_equivalent} gives the detailed derivation.

\begin{proposition}[Deterministic Equivalent of the Resolvent for Concentrated Feature Vectors]\label{prop:resolvent_deterministic_equivalent}
    Assume that the input-vector dimension $d$ and the depth of the quantum circuit $L$ are $\order{\poly(n)}$, and that Assumption~\ref{assum:exponential_effective_dimension} holds.
    Suppose further that the feature vectors $\bx_i$ are generated independently as described in Assumption~\ref{assum:feature_generation}(q) with \(\Omega\) in place of \(\Omega_a\).
    For fixed $\lambda > 0$, we consider the proportional regime; $p/N_{\mathrm{tr}} \to \gamma \in (0,\infty)$ as $p, N_{\mathrm{tr}} \to \infty$.
    Let $\kappa_\lambda > 0$ be the unique positive solution of the self-consistent equation
    \begin{align}
        \frac{1}{N_{\mathrm{tr}}}
        \Tr\left[(\Sigma + \kappa_\lambda I_p)^{-1}\Sigma\right]
        +
        \frac{\lambda}{\kappa_\lambda}
        =1 \,.
        \label{eq:self-consistent_eq}
    \end{align}
    % This equation is also referred to as a fixed-point equation in the following equivalent form:
    % \begin{align}
    %     \kappa_\lambda
    %     =
    %     \lambda
    %     +
    %     \frac{1}{N_{\mathrm{tr}}}
    %     \Tr\left[
    %         \Sigma
    %         \left(I_p+\frac{\Sigma}{\kappa_\lambda}\right)^{-1}
    %     \right] \,.
    %     \label{eq:self-consistent_eq_2}
    % \end{align}
    Define
    \begin{align}
        Q_p(\lambda)
        &:=
        \lambda(\hat{\Sigma}+\lambda I_p)^{-1},
        \\
        \overline Q_p(\lambda)
        &:=
        \kappa_\lambda
        (\Sigma+\kappa_\lambda I_p)^{-1} \,.
    \end{align}
    Then
    \begin{align}
        Q_p(\lambda)
        &\asymp_F
        \overline Q_p(\lambda),
        \label{eq:frobenius_resolvent_deterministic_equivalent}
    \end{align}
\end{proposition}

The following corollary records the two scalar interfaces of Proposition~\ref{prop:resolvent_deterministic_equivalent} that are used in the bias and variance calculations below.
\begin{corollary}[Scalar Interfaces of the Resolvent Deterministic Equivalent]
\label{cor:resolvent_scalar_interfaces}
    Under the assumptions of Proposition~\ref{prop:resolvent_deterministic_equivalent}, the following two interfaces hold:
    \begin{enumerate}[label=(\roman*),leftmargin=25pt]
        \item \textbf{Bounded-bilinear-form interface.}
        For every pair of deterministic vector sequences
        $u_p,v_p \in \bbR^p$ satisfying $\sup_p \norm{u_p}_2 \norm{v_p}_2 < \infty$, we have
        \begin{align}
            u_p\T
            \{Q_p(\lambda)-\overline Q_p(\lambda)\}
            v_p
            \longrightarrow 0
            \qquad\text{a.s.}
            \label{eq:resolvent_bilinear_interface}
        \end{align}
        \item \textbf{Normalized-trace interface.}
        For every deterministic matrix sequence $D_p$ satisfying $\sup_p\norm{D_p}_\infty < \infty$,
        \begin{align}
            \frac1p
            \Tr\left[
                D_p\{Q_p(\lambda)-\overline Q_p(\lambda)\}
            \right]
            \longrightarrow0
            \qquad\text{a.s.}
            \label{eq:resolvent_trace_interface}
        \end{align}
    \end{enumerate}
\end{corollary}
Indeed, these statements follow directly from Definition~\ref{def:frobenius_deterministic_equivalent} by taking $C_p=v_pu_p\T$ for the bilinear-form interface and $C_p=D_p/p$ for the normalized-trace interface.

We refer to the scalar $\kappa_\lambda$ as the \emph{effective regularization parameter} and it satisfies $\kappa_\lambda > \lambda > 0$ by Lemma~\ref{lemma:unique_positive_fixed_point}. We refer the reader to Appendix~\ref{appdx:sec:properties_of_kappa} for its detailed properties.

\subsection{Deterministic Equivalent of the Test Risk}\label{subsec:asymptotic_test_risk}
The test risk $R_{\lambda,\hat{\Sigma}}$ decomposes into bias and variance terms involving second-order resolvents (e.g., $(\hat{\Sigma} + \lambda I_p)^{-1} \Sigma (\hat{\Sigma} + \lambda I_p)^{-1}$ in Eq.~\eqref{eq:bias_term_in_risk} and $(\hat{\Sigma} + \lambda I_p)^{-2}$ in Eq.~\eqref{eq:variance_term_in_risk}), whereas Proposition~\ref{prop:resolvent_deterministic_equivalent} involves only the first-order resolvent $(\hat{\Sigma} + \lambda I_p)^{-1}$.
By observing that these second-order terms can be expressed as derivatives of the first-order resolvent with respect to $\lambda$ or an auxiliary parameter (\emph{derivative trick}~\cite{zavatone-vethlecture,atanasov2024scaling,dobriban2018high}) and using Vitali's convergence theorem, we can derive their deterministic equivalents from Proposition~\ref{prop:resolvent_deterministic_equivalent}.
In particular, the quadratic specialization $u_p = v_p$ of the bounded-bilinear-form interface~\eqref{eq:resolvent_bilinear_interface} is used for the bias, and the normalized-trace interface~\eqref{eq:resolvent_trace_interface} for the variance; Appendix~\ref{appdx:sec:derivative_method} gives the proof.
Using this approach, we obtain our main theoretical result.

\begin{theorem}[Deterministic Equivalent of Test Risk]\label{thm:test_risk_DE}
    Fix $\lambda > 0$.
    Suppose that Assumption~\ref{assum:target_regularity} holds on top of the assumptions of Proposition~\ref{prop:resolvent_deterministic_equivalent}.
    Then, the test risk $R_{\lambda,\hat{\Sigma}}$ of QKRR has the following deterministic equivalent:
    \begin{align}
        R_{\lambda,\hat{\Sigma}}
        \asymp
        R^{\mathrm{DE}}_{\lambda,\Sigma}
        &:=
        \underbrace{\frac{\kappa_\lambda^2}{1 - \eta_\lambda} \bs{\beta}_*\T (\Sigma + \kappa_\lambda I_p)^{-2} \bs{\beta}_*}_{\text{DE of Bias}} \nonumber\\
        &\quad+ \underbrace{\sigma^2 \frac{\eta_\lambda}{1 - \eta_\lambda}}_{\text{DE of Variance}}
        +\;\; \sigma^2 \,,
        \label{eq:test_risk_DE}
    \end{align}
    where ``DE'' stands for deterministic equivalent, $\kappa_\lambda$ is the unique solution to the self-consistent equation~\eqref{eq:self-consistent_eq}, and $\eta_\lambda$ is the normalized (second-order) effective degrees of freedom defined by
    \begin{align}
        \eta_\lambda
        :=
        \frac{1}{N_{\mathrm{tr}}}
        \Tr[\Sigma^2 (\Sigma + \kappa_\lambda I_p)^{-2}] \,.
    \end{align}
\end{theorem}
The same deterministic-equivalent machinery also yields the deterministic equivalent of the training error; see Appendix~\ref{appdx:sec:training_error_DE}.

Since Theorem~\ref{thm:test_risk_DE} holds pointwise for each fixed $\lambda > 0$, the normalization in Eq.~\eqref{eq:pure_state_kernel_normalization} has the following direct consequence of Theorem~\ref{thm:test_risk_DE}.
\begin{corollary}[Fixed-ridge Consequence of Pure-state Kernel Normalization]
\label{cor:fixed_ridge_normalized_quantum_kernel}
    Consider the normalization in Eq.~\eqref{eq:pure_state_kernel_normalization}. For every fixed $\lambda > 0$,
    \begin{align}
        \kappa_\lambda \longrightarrow \lambda \,,
        \quad
        \eta_\lambda \longrightarrow 0
        \quad\text{as} \;\; p \to \infty.
    \end{align}
    As a result, the deterministic equivalent of the test risk satisfies
    \begin{align}
        R^{\mathrm{DE}}_{\lambda,\Sigma}
        -
        \lambda^2 \bs{\beta}_*\T (\Sigma + \lambda I_p)^{-2} \bs{\beta}_* - \sigma^2
        \longrightarrow 0
        \;\;\text{as} \;\; p \to \infty.
    \end{align}
\end{corollary}
The proof is given in Appendix~\ref{appdx:subsec:fixed_ridge_normalized_quantum_kernel}.
For the test-risk DE to reach the Bayes risk $\sigma^2$ in QKRR, a sufficient scaling is $\lambda = \lambda_p \downarrow 0$ and $N_{\mathrm{tr}}\lambda_p \to \infty$ as $p \to \infty$, but this scaling is beyond the scope of Theorem~\ref{thm:test_risk_DE}, where $\lambda$ is fixed.

We now discuss the mechanisms of double descent by analyzing the behavior of the test-risk DE $R^{\mathrm{DE}}_{\lambda,\Sigma}$ at finite $(p,N_{\mathrm{tr}})$ as $\lambda \downarrow 0$ and the implications of the theorem; this does not assert a deterministic equivalent along a sequence $\lambda = \lambda_p \downarrow 0$ as $p \to \infty$.

\paragraph{Effective Degrees of Freedom\,:}
To quantify the model complexity more precisely than the effective feature dimension~$p$, we can consider the \emph{effective degrees of freedom (DOF)}~\cite{hastie2022surprises,bach2024highdimensional,atanasov2024scaling}.
For a generic parameter~$t$, the order-$s$ effective DOF is defined based on the eigenvalues~$\{\xi_i\}_{i=1}^p$ of the population covariance~$\Sigma$:
\begin{align}
    \mathrm{df}_s(t) = \sum_{i=1}^p \qty(\frac{\xi_i}{\xi_i + t})^s = \Tr[\Sigma^s(\Sigma + t I_p)^{-s}] \,.\label{eq:effective_DOF}
\end{align}
This quantity measures the number of eigenmodes of the feature vectors that are effectively learned by the model given the parameter~$t$.
Since $\eta_\lambda$ can be expressed as $\eta_\lambda = \frac{1}{N_{\mathrm{tr}}}\mathrm{df}_2(\kappa_\lambda)$, the term $\eta_\lambda$ captures the effective DOF of the model for $\kappa_\lambda$ relative to the sample size, which plays a central role in characterizing the double-descent peak.
From the definition of the effective DOF and the self-consistent equation~\eqref{eq:self-consistent_eq}, we have $\mathrm{df}_s(\kappa_\lambda) \leq \mathrm{df}_1(\kappa_\lambda) = N_{\mathrm{tr}}(1 - \lambda/\kappa_\lambda) < N_{\mathrm{tr}}$ for every $s\geq1$ and every $\lambda > 0$, and thus $0 \leq \eta_\lambda = \frac{1}{N_{\mathrm{tr}}}\mathrm{df}_2(\kappa_\lambda) < 1$. More quantitatively, Proposition~\ref{prop:effective_regularization_stability} gives, for each fixed $\lambda > 0$, a constant $a_\lambda > 0$ independent of $p$ such that $1 - \eta_\lambda \geq a_\lambda$ uniformly along the high-dimensional sequence; see Eq.~\eqref{eq:uniform_stability_eta}. This guarantees that the denominators in Theorem~\ref{thm:test_risk_DE} stay uniformly away from zero along the high-dimensional sequence at fixed $\lambda$, avoiding any divergence in the test-risk DE. However, as $\lambda \downarrow 0$, the denominator can approach zero at $p = N_{\mathrm{tr}}$ as shown below (c), leading to the double-descent peak.

\paragraph{Ridgeless Limit of the Effective Regularization Parameter\,:}
We show the following in Proposition~\ref{prop:kappa_zero_positive}.
For each fixed finite pair $(p,N_{\mathrm{tr}})$ and $\Sigma \succ 0$, the positive-$\lambda$ solution branch has the following ridgeless limit:
\begin{align}
    \kappa_0
    :=
    \lim_{\lambda \downarrow 0}\kappa_\lambda
    \begin{cases}
        = 0, & p \leq N_{\mathrm{tr}} \,,\\
        > 0, & p > N_{\mathrm{tr}} \,.
    \end{cases}
\end{align}
The solution for $p > N_{\mathrm{tr}}$ is the unique positive solution of
\begin{align}
    \mathrm{df}_1(\kappa_0)
    =
    \Tr[\Sigma(\Sigma + \kappa_0I_p)^{-1}]
    =
    N_{\mathrm{tr}} \,.
    \label{eq:kappa_zero_characterization}
\end{align}
Moreover, $\kappa_\lambda > \kappa_0$ for every $\lambda > 0$.
This positive endpoint in the overparameterized regime $p > N_{\mathrm{tr}}$ acts as an implicit regularization in the ridgeless limit $\lambda \downarrow 0$.

\paragraph{Mechanism behind the Interpolation Peak\,:}
Within the finite-$(p,N_{\mathrm{tr}})$ test-risk DE, the double-descent peak corresponds to a singularity in the variance term of the test risk.
As shown in Theorem~\ref{thm:test_risk_DE}, the variance is proportional to $\frac{\eta_\lambda}{1 - \eta_\lambda}$, which diverges when $\eta_\lambda \to 1$.
Since the ridgeless endpoint is $\kappa_0 = 0$ at $p = N_{\mathrm{tr}}$, $\mathrm{df}_2(\kappa_\lambda)\uparrow p = N_{\mathrm{tr}}$ and hence $\eta_\lambda \uparrow 1$.
More precisely, Appendix~\ref{appdx:sec:properties_of_kappa} shows that
\begin{align}
    \kappa_\lambda
    \propto
    \sqrt{\lambda},
    \qquad
    1 - \eta_\lambda
    \propto
    \sqrt{\lambda} \,.
\end{align}
Consequently, at fixed $(p,N_{\mathrm{tr}})$, the variance contribution grows as $\lambda^{-1/2}$ for $\sigma^2>0$, while the bias prefactor $\kappa_\lambda^2/(1 - \eta_\lambda)$ decays as $\lambda^{1/2}$, yielding the interpolation-peak scaling observed in classical kernel regression models~\cite{defilippis2024dimension}.

\paragraph{Task--model alignment and Bias Term\,:}
For the bias term to be small in the overparameterized regime, the large squared projected target coefficients $\beta_{*,i}^2$ should mainly correspond to eigendirections with large eigenvalues $\xi_i$.
This relationship is referred to as \emph{task--model alignment} in Ref.~\cite{canatar2021spectral,hastie2022surprises}.
Together with the spectral decay discussed below, task--model alignment is a crucial factor in determining the generalization performance of overparameterized models.

\paragraph{Spectral Decay\,:}
Complementary to task--model alignment, the eigenvalue distribution of $\Sigma$ also determines the effective degrees of freedom and thereby controls the variance contribution in Theorem~\ref{thm:test_risk_DE}.
In the context of CML, Refs.~\cite{bartlett2020benign,mallinar2207benign} demonstrate that spectral decay is essential for ridgeless models to generalize well in the overparameterized regime, i.e., to exhibit \emph{benign overfitting}.
In the quantum setting, however, highly expressive circuits approximating Haar-random unitaries tend to produce a flat spectrum with exponentially small eigenvalues.
Such behavior indicates a lack of spectral inductive bias, which makes it difficult to achieve generalization advantage over classical models~\cite{kubler2021inductive,canatar2023bandwidth,tomasi2025benign}.
This motivates designing structured quantum feature maps so that their spectra decay appropriately and their eigendirections align with the target function, thereby enhancing generalization performance.

\section{Experiments}\label{sec:experiment}
In this section, we compare the test-risk DE $R^{\mathrm{DE}}_{\lambda,\Sigma}$ derived in Theorem~\ref{thm:test_risk_DE} against the empirical test error in two complementary numerical studies. First, we use teacher--student models in which datasets are synthetically generated from the same quantum feature maps as the student models, providing a controlled validation of the test-risk DE $R^{\mathrm{DE}}_{\lambda,\Sigma}$ from Theorem~\ref{thm:test_risk_DE}.
Second, we apply the same formula to the Fashion-MNIST dataset as a robustness check; it is impossible to verify the assumptions of Theorem~\ref{thm:test_risk_DE} unless the data-generating process is known, but we can still examine whether the test-risk DE $R^{\mathrm{DE}}_{\lambda,\Sigma}$ provides an accurate prediction of the empirical test error.
We used a Hardware-Efficient Ansatz (HEA) and a Tensor Product Ansatz (TPA) for the quantum feature map.
We performed statevector simulations using \texttt{PennyLane}~\cite{bergholm2022pennylane}, neglecting shot noise in the estimation of the kernel matrix.
The key parameters and settings used across all experiments are summarized in Table~\ref{tab:experiment_settings}.

\begin{table*}[t]
\centering
\setlength{\tabcolsep}{10pt}
\caption{Settings used in the experiments. The number of training samples $N_{\mathrm{tr}}$ is varied to sweep the model complexity ratio $\tilde{\gamma} := p/N_{\mathrm{tr}}$, and the regularization parameter takes values in $\{10^{-10},10^{-8},10^{-6},10^{-4},10^{-2}\}$.}
\label{tab:experiment_settings}
\renewcommand{\arraystretch}{1.05}
\begin{tabular}{llcccccccc}
    \toprule
    Dataset & Ansatz & $n$ & $L$ & $d$ & $p$ & $N_{\mathrm{reps}}$ & $N_{\mathrm{ts}}$ & $N_{\mathrm{est}}$ & $\sigma$ \\
    \midrule
    Fashion-MNIST & HEA & 3 & 3 & 6 & 64 & 100 & 1,000 & 3,000 & $0.5$ (calibrated) \\
    Synthetic & HEA & 3 & 3 & 6 & 64 & 100 & 1,000 & 3,000 & $0.3$ \\
    Synthetic & HEA & 5 & 5 & 10 & 1024 & 30 & 1,000 & 10,000 & $0.3$ \\
    \midrule
    Fashion-MNIST & TPA & 3 & 3 & 6 & 27 & 100 & 1,000 & 3,000 & $0.5$ (calibrated) \\
    Synthetic & TPA & 3 & 3 & 6 & 27 & 100 & 1,000 & 3,000 & $0.3$ \\
    Synthetic & TPA & 5 & 5 & 10 & 243 & 30 & 1,000 & 10,000 & $0.3$ \\
    \bottomrule
\end{tabular}
\end{table*}

\subsection{Experimental Setup}
\paragraph{Quantum Feature Map\,:}
A quantum feature map maps an input vector $\bs{u} \in \bbR^d$ to a quantum state~$\rho(\bs{u}) = U(\bs{u})\dyad{\bs{0}}U\dg(\bs{u})$ in an $n$-qubit Hilbert space.
The kernel is defined as $k(\bs{u}, \bs{u}') = \Tr[\rho(\bs{u}) \rho(\bs{u}')] = |\bra{\bs{0}} U\dg(\bs{u}) U(\bs{u}') \ket{\bs{0}}|^2$.
We used the HEA and the TPA to construct the unitary~$U(\bs{u})$.

The employed HEA consists of $L = n$ layers, where each layer applies a sequence of single-qubit rotation gates ($R_X^{(j)}$ and $R_Z^{(j)}$) followed by a linear chain of CX gates for entanglement. Indices of the rotation gates indicate the qubits on which they act.
Specifically, the unitary~$U(\bs{u})$ is given by
\begin{align}
    U(\bs{u})
    =
    &\Bigg[
        \Bigg( \prod_{j=1}^{n-1} \mathrm{CX}_{j, j+1} \Bigg)
        \Bigg( \prod_{j=1}^d R^{(\omega_j)}_Z(u_j)R^{(\omega_j)}_X(u_j) \Bigg)
    \Bigg]^{L = n} \,,
\end{align}
where $\omega_j := ((j-1) \bmod n) + 1$.
For input vectors with dimension $d > n$, the entries of $\bs{u}$ are cyclically mapped to the qubits, meaning that the $j$-th entry of~$\bs{u}$ is mapped to the $\omega_j$-th qubit.
Fig.~\ref{fig:hardware_efficient_ansatz} illustrates the HEA for $n = 4$ qubits.

The employed TPA consists only of single-qubit rotation gates $R_X^{(j)}$ as follows:
\begin{align}
    U(\bs{u})
    =
    \Bigg[
        \prod_{j=1}^d R^{(\omega_j)}_X(u_j)
    \Bigg]^{L = n} \,.
\end{align}

As discussed in Section~\ref{sec:setting}, the effective feature dimension is $p \leq 4^n$ for the HEA and $p \leq 3^n$ for the TPA.
While $p$ can be smaller than these upper bounds depending on the input distribution, we observed that for the synthetic datasets and Fashion-MNIST with our specific encodings, the effective feature dimension reaches these upper bounds. Thus, we set $p = 4^n$ for the HEA and $p = 3^n$ for the TPA in our experiments.

\begin{figure}[t]
\centering
\begin{quantikz}[row sep=0.1cm, column sep=0.16cm]
\lstick{$q_1$}
    & \qw & \gate{R_X(u_1)}\gategroup[wires=4,steps=7,style={dashed}]{$\times L$}
                            & \gate{R_Z(u_1)} & \gate{R_X(u_5)} & \gate{R_Z(u_5)} & \ctrl{1} & \qw      & \qw       & \qw & \qw \\
\lstick{$q_2$}
    & \qw & \gate{R_X(u_2)} & \gate{R_Z(u_2)} & \gate{R_X(u_6)} & \gate{R_Z(u_6)} & \targ{}  & \ctrl{1} & \qw       & \qw & \qw \\
\lstick{$q_3$}
    & \qw & \gate{R_X(u_3)} & \gate{R_Z(u_3)} & \gate{R_X(u_7)} & \gate{R_Z(u_7)} & \qw      & \targ{}  & \ctrl{1}  & \qw & \qw \\
\lstick{$q_4$}
    & \qw & \gate{R_X(u_4)} & \gate{R_Z(u_4)} & \gate{R_X(u_8)} & \gate{R_Z(u_8)} & \qw      & \qw      & \targ{}   & \qw & \qw
\end{quantikz}
\caption{
    HEA for $n = 4$ qubits. The circuit consists of $L = n$ layers, where each layer applies single-qubit rotations ($R_X$ and $R_Z$) on each qubit, followed by a linear chain of CX gates. The entries of $\bs{u} \in \bbR^{d=2n}$ are cyclically mapped to the qubits. The TPA is obtained by removing the $R_Z$ and CX gates from this circuit, leaving only the $R_X$ gates.
}\label{fig:hardware_efficient_ansatz}
\end{figure}

\paragraph{Synthetic Dataset Generation\,:}
To provide a controlled validation of Theorem~\ref{thm:test_risk_DE}, we generated each synthetic dataset from the same quantum feature map as each student model.
The input vector was drawn independently as $\bs{u} \sim \calN(0,I_d)$ with $d = 2n$; hence, the data-generating map $\chi$ in Assumption~\ref{assum:feature_generation} is the identity.
For each ansatz and qubit number $n$, a coefficient vector $\bs{w}_{*,a}$ was independently drawn from the unit sphere in $\bbR^{4^n}$ and used as the normalized-Pauli coefficients of the teacher observable $O_*$. Reusing the notation of Section~\ref{subsec:qkm_setting}, we generated the label as
\begin{align}
    y
    &= \Tr[\rho(\bs{u})O_*] + \vep
    = \bs{w}_{*,a}\T\bs{r}_a(\bs{u}) + \vep
    = \bs{w}_*\T\bs{r}(\bs{u}) + \vep
\end{align}
where the noise $\vep$ was drawn independently from $\calN(0,\sigma^2 = 0.3^2)$.
This is the well-specified quantum regression model introduced in Section~\ref{subsec:qkm_setting}. For the TPA, $\bs{w}_{*,a}$ still lies in the $4^n$-dimensional ambient Pauli-coordinate space, but only its projection $\bs{w}_*$ onto the $p=3^n$ effective feature space contributes to the target. Hence, $\|\bs{w}_*\|_2 \leq \|\bs{w}_{*,a}\|_2 = 1$.
This synthetic dataset generation process satisfies the assumptions of Theorem~\ref{thm:test_risk_DE}.
We performed experiments with $n = 3, 5$ qubits.

\paragraph{Fashion-MNIST Dataset\,:}
We also evaluated the test-risk DE $R^{\mathrm{DE}}_{\lambda,\Sigma}$ on the Fashion-MNIST dataset to examine whether it is informative for real-world structured data. This experiment is treated as an empirical robustness check rather than a direct validation of the theorem, since the data-generating process is unknown.
We selected a binary classification subset consisting of classes 0 (T-shirt) and 1 (Trouser).
The data preprocessing pipeline involved the following steps:
\begin{enumerate}[leftmargin=20pt]
    \item Standardizing the pixel values to zero mean and unit variance.
    \item Reducing the dimensionality via principal component analysis (PCA) to $d = 2n$ components.
\end{enumerate}
We used $n = 3$ qubits for this task. Although this is a classification problem, we treated it as a regression task on the class labels $\{0, 1\}$ for the purpose of analyzing the double descent profile.
Unlike the synthetic dataset, the true noise level is unknown for real-world data. Since we needed to specify the noise variance $\sigma^2$ to plot the test-risk DE $R^{\mathrm{DE}}_{\lambda,\Sigma}$, we calibrated $\sigma$ to $0.5$ so that the test-risk DE curve matches the empirical test error at $\lambda = 10^{-2}$, and fixed it for all values of $\lambda$.

\paragraph{Training and Evaluation\,:}
For each combination of dataset and ansatz, we trained the QKRR model for various numbers of training samples $N_{\mathrm{tr}}$ to sweep the model complexity ratio $\tilde{\gamma} := p/N_{\mathrm{tr}}$ ($p = 4^n$ for HEA and $p = 3^n$ for TPA).
We also varied the regularization parameter~$\lambda$ across a range of values to observe its effect on the double-descent peak.
The empirical test error was calculated as the mean squared error (MSE) over a test dataset of size $N_{\mathrm{ts}} = 1,000$ for the synthetic datasets and Fashion-MNIST.
We computed the empirical test error for each setting by averaging over $N_{\mathrm{reps}} = 100$ different training datasets for synthetic data with $n = 3$ qubits and Fashion-MNIST with $n = 3$ qubits, and $N_{\mathrm{reps}} = 30$ for the synthetic data with $n = 5$ qubits.

\begin{figure*}[t]
    \centering
    \includegraphics[width=0.32\textwidth]{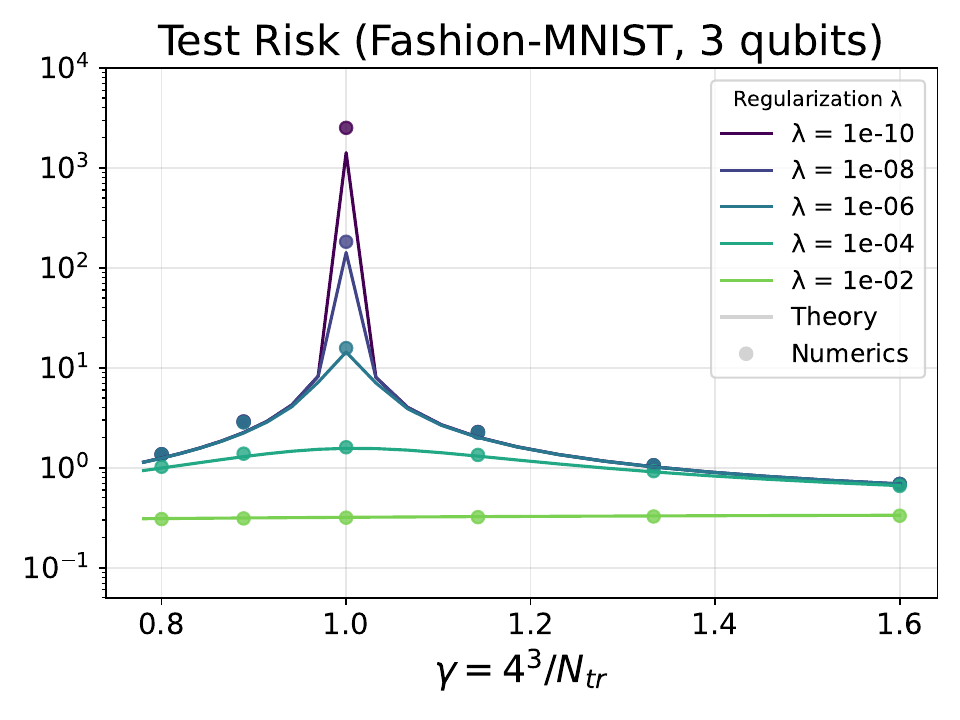} \hfill
    \includegraphics[width=0.32\textwidth]{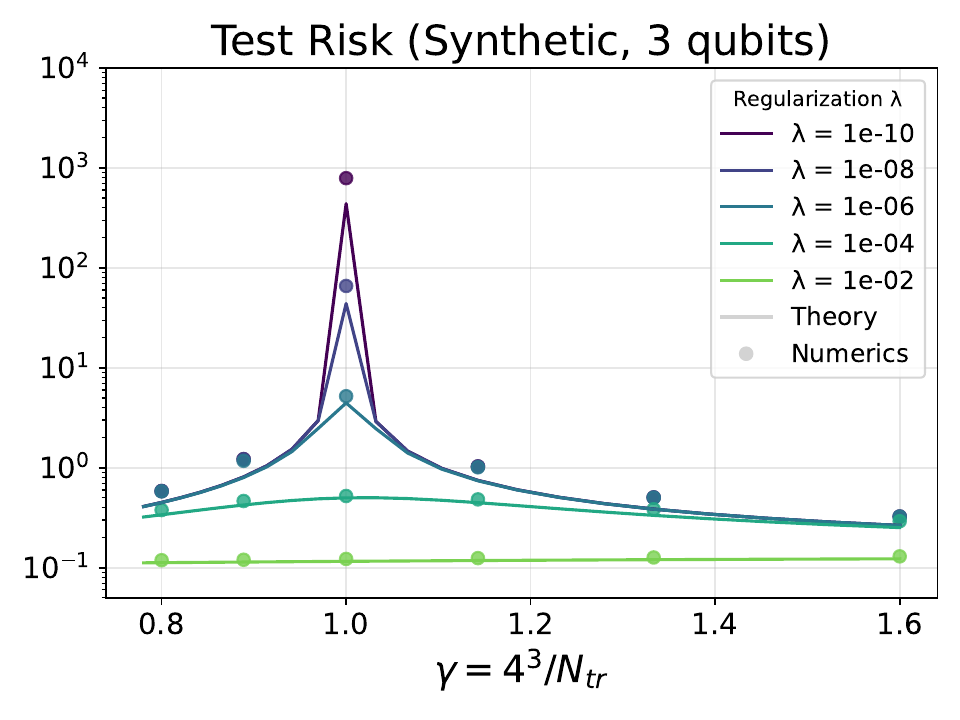} \hfill
    \includegraphics[width=0.32\textwidth]{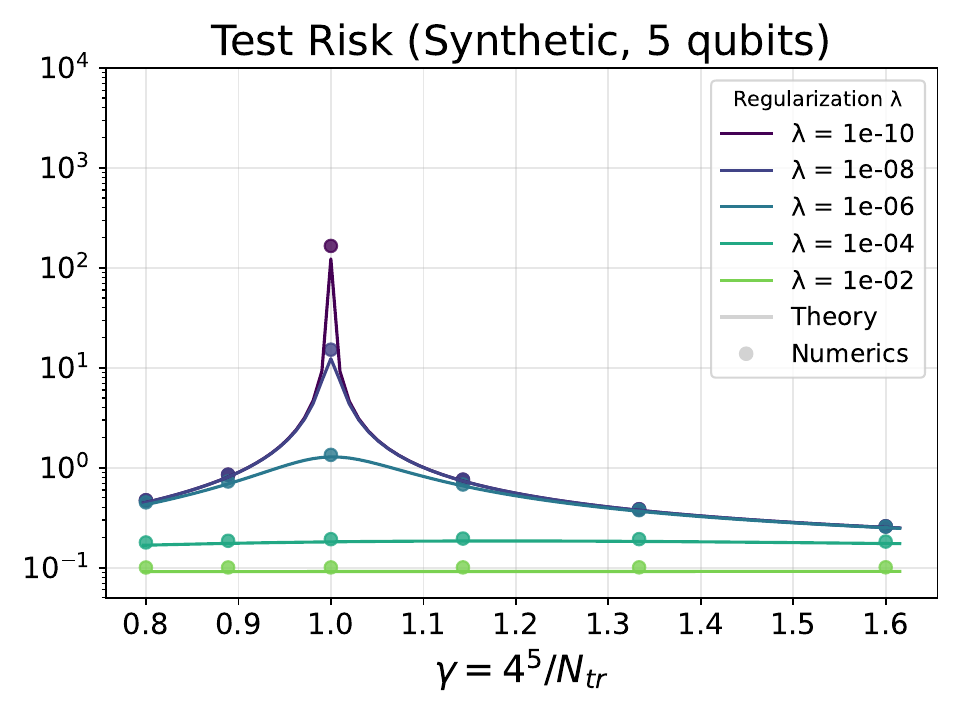}
    \caption{
        Test-risk DE and empirical test error of QKRR with HEA.
        The plots show the test-risk DE (solid lines) and empirical test error (dots) as a function of the model complexity ratio $\tilde{\gamma} := p/N_{\mathrm{tr}} = 4^n/N_{\mathrm{tr}}$.
        \textbf{Left:} Fashion-MNIST dataset (classes 0 vs 1) with $n = 3$ qubits ($p = 64$).
        \textbf{Middle:} Synthetic dataset with $n = 3$ qubits ($p = 64$).
        \textbf{Right:} Synthetic dataset with $n = 5$ qubits ($p = 1024$).
        Solid lines represent the test-risk DE $R^{\mathrm{DE}}_{\lambda,\Sigma}$. Different colors correspond to different regularization-parameter values $\lambda$. The interpolation threshold corresponds to $N_{\mathrm{tr}} = 4^n$.
        Dots represent empirical test errors averaged over multiple trials: $N_{\mathrm{reps}} = 100$ for Fashion-MNIST dataset and 3-qubit synthetic dataset, $N_{\mathrm{reps}} = 30$ for 5-qubit synthetic dataset.
        % Error bars represent the standard deviation across independent training datasets.
    }
    \label{fig:test_risk_HE}
\end{figure*}

\subsection{Estimation of Population Statistics}\label{subsec:estimation_sigma_beta}
To evaluate the test-risk DE $R^{\mathrm{DE}}_{\lambda,\Sigma}$, one requires explicit knowledge of the population covariance matrix $\Sigma$ and the projected target vector $\bs{\beta}_*$ (i.e., the representation of the target function $f_*$ in the kernel eigenbasis). For complex data distributions such as Fashion-MNIST or for quantum feature maps, these quantities are not analytically available.

To address this issue, we estimated $\Sigma$ and $\bs{\beta}_*$ using a procedure similar to Algorithm 1 in Ref.~\cite{defilippis2024dimension} with a large \textit{estimation dataset} $\calD_{\mathrm{est}} = \{(\bs{u}_j, y_j)\}_{j=1}^{N_{\mathrm{est}}}$ that is independent of the training and test data.
For the synthetic dataset tasks, we used $N_{\mathrm{est}} = 3,000$ samples for $3$-qubit tasks or $N_{\mathrm{est}} = 10,000$ samples for $5$-qubit tasks. For the Fashion-MNIST $3$-qubit classification tasks, we used $N_{\mathrm{est}} = 3,000$ samples.
The estimation procedure is based on the eigendecomposition of the kernel matrix $K_{\mathrm{est}}$ computed on the estimation dataset.
We present the final estimators here, and their derivations are detailed in Appendix~\ref{appdx:sec:estimation_Sigma_beta}.
For the experiments reported here, we numerically verified that the rank of $K_{\mathrm{est}}$ agrees with the effective dimension $p$. Since $N_{\mathrm{est}} > p$, we use the compact rank-$p$ eigendecomposition
\begin{align}
    \frac{1}{N_{\mathrm{est}}} K_{\mathrm{est}}
    &= H_p \Lambda_p^{(\mathrm{est})} H_p\T \,, \\
    \Lambda_p^{(\mathrm{est})}
    &:= \diag(\xi_1^{(\mathrm{est})},\ldots,\xi_p^{(\mathrm{est})}) \,,
\end{align}
where $\xi_i^{(\mathrm{est})}$ are the eigenvalues of $K_{\mathrm{est}}/N_{\mathrm{est}}$ and $H_p \in \bbR^{N_{\mathrm{est}} \times p}$ is the matrix whose orthonormal columns are the corresponding eigenvectors.
The population covariance matrix $\Sigma$ was estimated by the diagonal matrix $\Lambda_p^{(\mathrm{est})}$.
To estimate the projected target vector $\bs{\beta}_*$, we need the noise-free target values $f_*(\bs{u}_j)$.
For synthetic data, we use the available noise-free teacher values $f_*(\bs{u}_j)$; for Fashion-MNIST, the observed class labels are used as a proxy for the unknown noise-free target values.
With $\bs{t}_{\mathrm{est}}=(t_1,\ldots,t_{N_{\mathrm{est}}})\T$, where $t_j=f_*(\bs{u}_j)$ for synthetic data and $t_j=y_j$ for Fashion-MNIST, the projected target vector is estimated as
\begin{align}
    \bs{\beta}^{(\mathrm{est})}_*
    := \frac{1}{\sqrt{N_{\mathrm{est}}}}
    H_p\T\bs{t}_{\mathrm{est}}
    \in \bbR^{p} \,.
\end{align}
These estimates $\Lambda_p^{(\mathrm{est})}$ and $\bs{\beta}^{(\mathrm{est})}_*$ were then used in place of the true $\Sigma$ and $\bs{\beta}_*$ to plot the test-risk DE $R^{\mathrm{DE}}_{\lambda,\Sigma}$.

We note that Appendix~\ref{appdx:sec:training_error_DE} discusses the generalized cross-validation (GCV) risk estimator, which is expected to be asymptotically equivalent to the test risk under the same assumptions and is computed only from the training dataset without requiring an additional estimation dataset.

\subsection{Results and Discussion}
\begin{figure*}[t]
    \centering
    \includegraphics[width=0.32\textwidth]{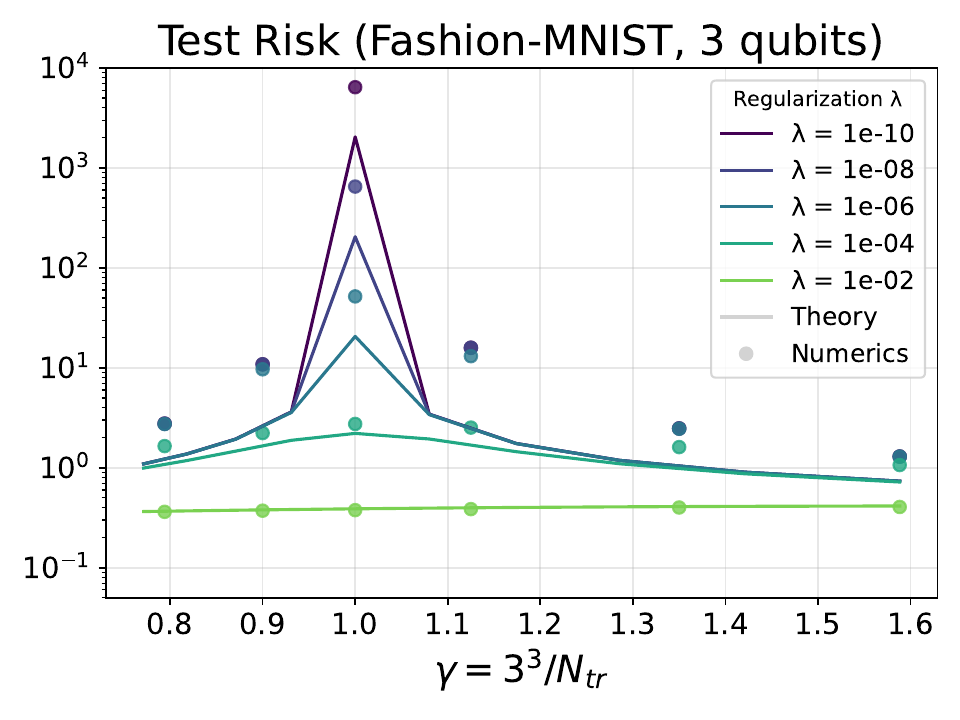} \hfill
    \includegraphics[width=0.32\textwidth]{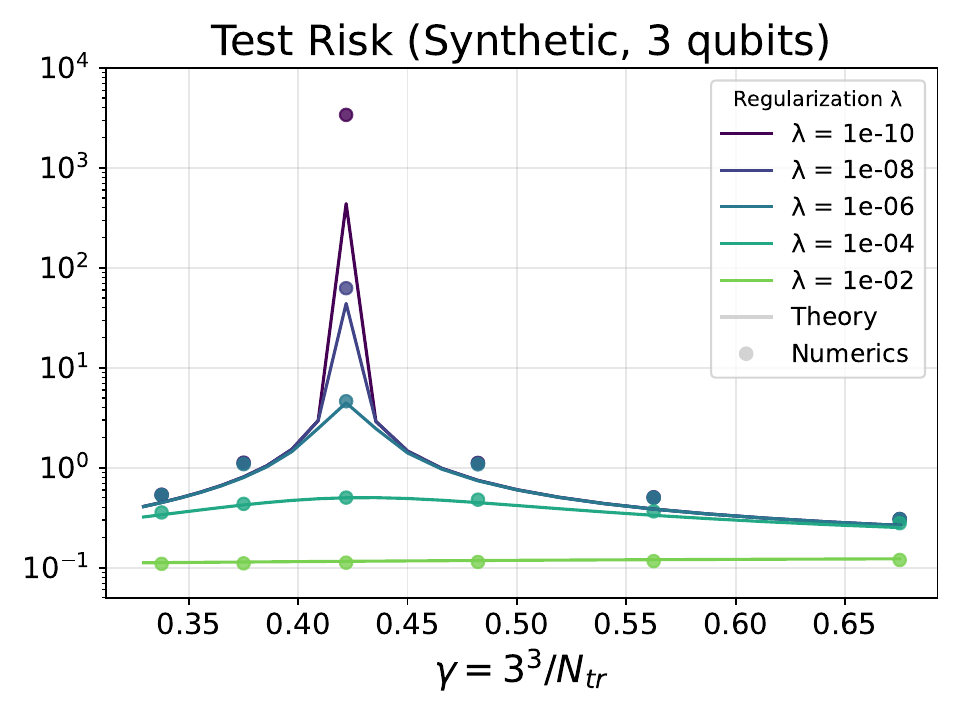} \hfill
    \includegraphics[width=0.32\textwidth]{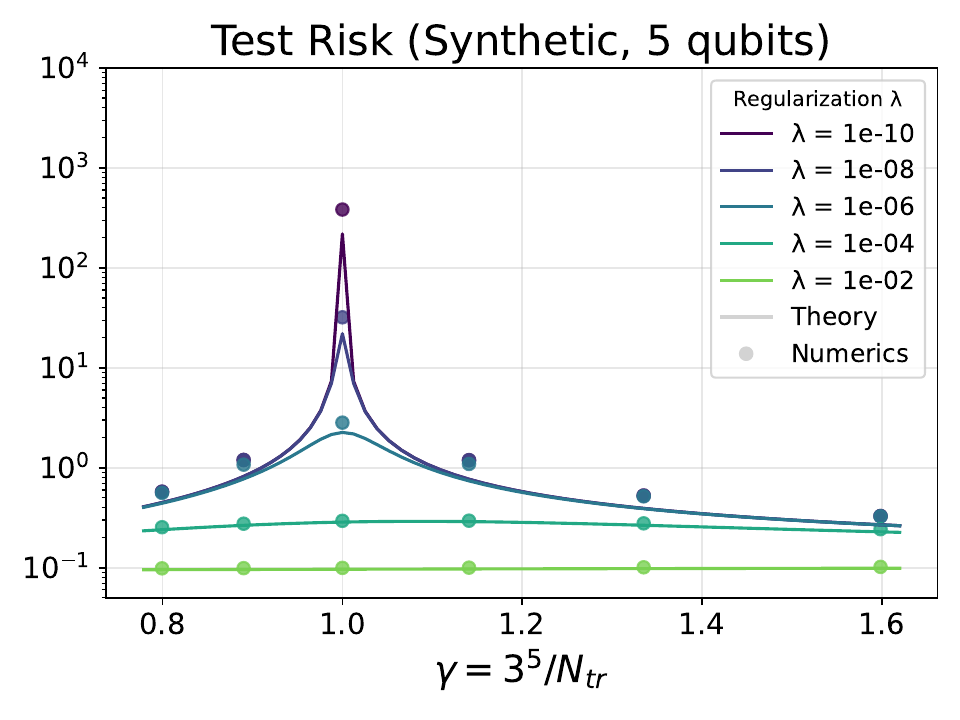}
    \caption{
        Test-risk DE and empirical test error of QKRR with TPA.
        The plots show the test-risk DE (solid lines) and empirical test error (dots) as a function of the model complexity ratio $\tilde{\gamma} := p/N_{\mathrm{tr}} = 3^n/N_{\mathrm{tr}}$.
        \textbf{Left:} Fashion-MNIST dataset (classes 0 vs 1) with $n = 3$ qubits ($p = 27$).
        \textbf{Middle:} Synthetic dataset with $n = 3$ qubits ($p = 27$).
        \textbf{Right:} Synthetic dataset with $n = 5$ qubits ($p = 243$).
        Solid lines represent the test-risk DE $R^{\mathrm{DE}}_{\lambda,\Sigma}$. Different colors correspond to different regularization-parameter values $\lambda$. The interpolation threshold corresponds to $N_{\mathrm{tr}} = 3^n$.
        Dots represent empirical test errors averaged over multiple trials: $N_{\mathrm{reps}} = 100$ for Fashion-MNIST dataset and 3-qubit synthetic dataset, $N_{\mathrm{reps}} = 30$ for 5-qubit synthetic dataset.
        % Error bars represent the standard deviation across independent training datasets.
    }
    \label{fig:test_risk_TP}
\end{figure*}

Fig.~\ref{fig:test_risk_HE} presents the comparison between the test-risk DE $R^{\mathrm{DE}}_{\lambda,\Sigma}$ (solid lines) and the empirical test error (dots) for the HEA. Fig.~\ref{fig:test_risk_TP} shows the same comparison for the TPA.
Different colors indicate different regularization-parameter values $\lambda \in \{10^{-10}, 10^{-8}, 10^{-6}, 10^{-4}, 10^{-2}\}$.
Each figure covers three scenarios: (Left) Fashion-MNIST classification (0 vs 1) using 3 qubits; (Middle) synthetic dataset regression using 3 qubits; and (Right) synthetic dataset regression using 5 qubits.
The effective feature dimension~$p$ is determined by the choice of ansatz and the input distribution: $p = 4^n$ for the HEA and $p = 3^n$ for the TPA in our experiments.
For the HEA, we have $p = 64$ for $n = 3$ qubits and $p = 1024$ for $n = 5$ qubits. For the TPA, we have $p = 27$ for $n = 3$ qubits and $p = 243$ for $n = 5$ qubits.
The $x$-axis represents the ratio of the effective feature dimension to the number of training samples, $\tilde{\gamma} := p/N_{\mathrm{tr}}$.
Since the number of qubits (and thus $p$) is fixed for a given circuit, varying~$\tilde{\gamma}$ corresponds to varying the training sample size $N_{\mathrm{tr}}$.

\paragraph{Double Descent and Regularization\,:}
In all scenarios, both the empirical test error and the test-risk DE exhibit the characteristic non-monotonic behavior of double descent. A peak in the test risk is observed at the interpolation threshold $N_{\mathrm{tr}} = p$, particularly for small regularization-parameter values (e.g., $\lambda = 10^{-10}$). Within the test-risk DE, this double-descent peak arises from the divergence of the variance term when the normalized effective DOF $\eta_\lambda$ approaches~$1$. As predicted by our theory, increasing the explicit regularization $\lambda$ effectively suppresses the double-descent peak, smoothing the transition between the underparameterized and overparameterized regimes.
We provide a plot of the normalized effective DOF $\eta_\lambda$ for the synthetic 3-qubit dataset in Appendix~\ref{appdx:sec:effective_DOF}, which shows that $\eta_\lambda$ approaches~$1$ near the interpolation threshold for small $\lambda$.

\paragraph{Agreement between Test-Risk DE and Empirical Test Errors\,:}
The test-risk DE curves show close agreement with the empirical test errors across all settings, capturing both the location and the magnitude of the double-descent peak.
Notably, the agreement improves as the effective feature dimension increases. Comparing the middle plot ($n = 3, p = 64$) with the right plot ($n = 5, p = 1024$) in Fig.~\ref{fig:test_risk_HE} for the HEA, we observe that the 5-qubit case exhibits a more accurate prediction. Such improved agreement is also seen in Fig.~\ref{fig:test_risk_TP} for the TPA, where the 5-qubit case ($p = 243$) shows a closer match than the 3-qubit case ($p = 27$). For the HEA, even at the modest scale of $n = 5$ qubits, the effective feature dimension seems to be sufficiently large for the test-risk DE to provide a highly accurate description of the empirical test error. For the TPA, the agreement in the 5-qubit case is not as accurate as in the HEA, which may be attributed to the smaller effective feature dimension.

\paragraph{Robustness Check on Fashion-MNIST\,:}
The left panels of Figs.~\ref{fig:test_risk_HE} and~\ref{fig:test_risk_TP} provide a robustness check on Fashion-MNIST.
After calibrating an effective noise parameter $\sigma$, the test-risk DE curves reproduce several qualitative features of the empirical test error profile, including the double-descent peak and its suppression by explicit regularization.
This suggests that the test risk behavior identified by the theory can remain informative beyond the controlled setting.

\vspace{-7pt}
\section{Conclusion}\label{sec:conclusion}
\paragraph{Summary of Contributions\,:}
In this work, we derived a rigorous deterministic equivalent of the test risk of quantum kernel ridge regression (QKRR) at fixed $\lambda > 0$ using the effective quantum feature map and deterministic equivalents for concentrated feature vectors from random matrix theory (RMT).
This analytical expression accurately captures the double-descent curve, explaining how the interplay between the effective feature dimension $p\leq 4^n$, the number of training samples $N_{\mathrm{tr}}$, the population covariance matrix $\Sigma$, and the regularization parameter $\lambda$ determines the model's performance.

We performed numerical experiments using HEA and TPA data-encoding circuits on synthetic data with 3 and 5 qubits and on Fashion-MNIST with 3 qubits to evaluate how accurately the test-risk DE predicts the empirical test error.
The first experiment is a controlled teacher--student experiment where the targets are synthetically generated from the same quantum feature maps as the student models, and the second is a robustness check on Fashion-MNIST, where the assumptions in Theorem~\ref{thm:test_risk_DE} cannot be rigorously verified.
Close agreement between the test-risk DE and the empirical test errors was observed even for 3- and 5-qubit circuits, with particularly strong agreement for the HEA.
This underscores the practical utility of RMT in evaluating QKRR models without the need for computationally expensive exhaustive training runs.

\paragraph{Limitations\,:}
A central limitation is that Theorem~\ref{thm:test_risk_DE} is pointwise in each fixed $\lambda > 0$.
Under the pure-state normalization in Eq.~\eqref{eq:pure_state_kernel_normalization} and Assumption~\ref{assum:target_regularity}, for both the bias and variance contributions to vanish, we will need to consider the regime where $\lambda_p$ decays with $p$, as is considered in the CML literature~\cite{atanasov2024scaling,defilippis2024dimension,MisiakiewiczSaeed2024}.
In our analysis, we have focused on the case of fixed positive regularization so that we can rigorously derive a deterministic equivalent of the test risk.
Characterizing dimension-dependent regularization is important for understanding benign overfitting and \emph{scaling laws} in QKRR.
Future work should aim to extend our analysis to cover this regime, potentially by developing new techniques in random matrix theory that can handle vanishing regularization parameter sequences.

While our framework offers significant insights, it assumes that the kernel matrix can be estimated without noise, which is not the case in real quantum hardware.
In practice, deep quantum circuits often suffer from quantum kernel concentration, where kernel values concentrate exponentially around their mean as the number of qubits increases~\cite{thanasilp2022exponential}, requiring an exponential number of shots to resolve. This complicates the numerical estimation of the kernel in high-qubit regimes. Future work should integrate the effects of quantum kernel concentration and measurement noise into the RMT analysis to provide a more realistic account of scalability.

\paragraph{Future Directions\,:}
An important avenue for future research is to evaluate the effect of quantum-circuit properties (e.g., entanglement, circuit expressibility) on the spectral decay of the population covariance matrix. Since the test-risk DE is governed by this spectrum, understanding how these properties correlate with spectral decay would enable the principled design of quantum kernels that might lead to benign overfitting.
This direction has been particularly investigated in terms of the bandwidth of quantum kernels in~\cite{shaydulin2022importance,canatar2023bandwidth,tomasi2025benign}.
Pursuing this line of inquiry using our framework could ultimately enable a comparative analysis of benign overfitting in quantum versus classical kernels and help delineate conditions under which quantum kernels offer a genuine generalization advantage in the overparameterized regime.

Other promising directions for future work include the study of variational quantum machine learning (VQML)~\cite{mitarai2018quantum,farhi2018classification,benedetti2019parameterized,schuld2020circuit}.
In VQML, the inherent unitarity of trainable quantum gates may impose implicit constraints distinct from those of classical regularization. Understanding how this \emph{unitarity-induced regularization}~\cite{chen2021expressibility} interacts with double descent is a compelling open question.
Additionally, investigating \emph{regularization-wise}~\cite{yilmaz2022regularizationwise} and \emph{epoch-wise}~\cite{nakkiran2021deep} double descent in the quantum context might yield further insights into the generalization dynamics of QML models.

\section*{Acknowledgments}
We thank Cosme Louart for fruitful discussions on random matrix theory.
We thank Pablo Rodriguez-Grasa for bringing Ref.~\cite{pablo2026pacbayes} to our attention and for helpful comments on PAC-Bayesian bounds.
K.K. was supported by JST SPRING, Grant Number JPMJSP2108.
This work was supported by the Center of Innovation for Sustainable Quantum AI, JST Grant Number JPMJPF2221.

\section*{Data Availability}
The code for the numerical experiments is available at \href{https://github.com/kkensuke/DD_in_QKRR_code}{github.com/kkensuke/DD\_in\_QKRR\_code}.

% \clearpage
\makeatletter
\def\bibsection{\section*{\refname}}
\makeatother
\bibliography{ref}

\onecolumngrid
\appendix
\newpage
\newpage
\section{Table of Notations}\label{appdx:sec:notation_table}
\begin{longtable}{|c|l|}
\caption{Summary of key notations used in this work.}\label{tab:notation_table}\\
\hline
\textbf{Symbol} & \textbf{Description} \\
\hline
\endfirsthead
\multicolumn{2}{c}
{{\bfseries \tablename\ \thetable{} -- continued from previous page}} \\
\hline
\textbf{Symbol} & \textbf{Description} \\
\hline
\endhead
\hline \multicolumn{2}{r}{{Continued on next page}} \\
\endfoot
\hline
\endlastfoot

$[m]$ & The set $\{1,2,\ldots,m\}$ for a positive integer $m$. \\

$\norm{\cdot}_2$ & Euclidean ($\ell_2$) norm for finite-dimensional vectors. \\

$\norm{\cdot}_1$ & Trace norm (Schatten $1$ norm) for matrices. \\

$\norm{\cdot}_F$ & Frobenius norm (Schatten $2$ norm) for matrices. \\

$\norm{\cdot}_\infty$ & Spectral norm (Schatten $\infty$ norm) for matrices. \\

$n$ & Number of qubits. \\

$N_{\mathrm{tr}}$ & Number of training samples. \\

$\mu$ & Probability distribution of the input vector $\bs{u}$ on $\calU$. \\

$\{(\bs{u}_i,y_i)\}_{i=1}^{N_{\mathrm{tr}}}$ & Training dataset, where $\bs{u}_i$ is the $i$-th input vector and $y_i$ is its label. \\

$\rho(\bs{u})$ & Density matrix produced from the input vector $\bs{u}$. \\

$\kett{\rho(\bs{u})}$ & Vectorized form of the density operator $\rho(\bs{u})$. \\

$\bs{r}_a(\bs{u})$ & Ambient Pauli feature vector, i.e., the coefficient vector of $\rho(\bs{u})$ in the normalized Pauli basis; $\bs{r}_a(\bs{u})\in\bbR^{4^n}$. \\

$\phi_a(\bs{u})$ & Ambient feature vector produced by the ambient feature map $\phi_a \colon \calU \to \bbR^{p_a}$. For QKRR, $\phi_a(\bs{u}) = \bs{r}_a(\bs{u})$. \\

$\phi(\bs{u})$ & Effective feature vector produced by the effective feature map $\phi \colon \calU \to \bbR^p$. \\
& It reproduces the same kernel as $\phi_a$. For QKRR, $\phi(\bs{u}) = \bs{r}(\bs{u})$. \\

$p_a$ & Ambient feature dimension. For the normalized-Pauli representation used in QKRR, $p_a = 4^n$. \\

$p$ & Effective feature dimension, defined as $p := \rank\E_{\bs{u}\sim\mu}[\phi_a(\bs{u})\phi_a(\bs{u})\T] \leq p_a$. \\

$\bth_a,\bth$ & Model-parameter vectors in the ambient and effective feature spaces, respectively; $\bth_a \in \bbR^{p_a}$ and $\bth \in \bbR^p$. \\

$\bth_{*,a},\bth_*$ & Target vectors in the ambient and effective feature spaces, respectively; $\bth_{*,a} \in \bbR^{p_a}$ and $\bth_* \in \bbR^p$. \\

$y,\vep,\sigma^2$ & Label model $y = f_*(\bs{u}) + \vep$, where $f_*(\bs{u}) = \bth_{*,a}\T\phi_a(\bs{u}) = \bth_*\T\phi(\bs{u})$, $\E[\vep] = 0$, and $\operatorname{Var}(\vep) = \sigma^2$. \\

$\bx$ & Effective feature vector $\bx := \phi(\bs{u}) \in \bbR^p$. \\

$\by$ & Training-label vector $\by:=(y_1,\ldots,y_{N_{\mathrm{tr}}})\T\in\bbR^{N_{\mathrm{tr}}}$. \\

$X$ & Row-oriented training feature matrix $X:=[\bx_1,\ldots,\bx_{N_{\mathrm{tr}}}]\T\in\bbR^{N_{\mathrm{tr}}\times p}$. \\

$X_{\mathrm L}$ & Column-oriented feature matrix used only when translating the notation of Ref.~\cite{louart2021spectral}; $X_{\mathrm L}:=X\T\in\bbR^{p\times N_{\mathrm{tr}}}$. \\

$\hat{\Sigma}$ & Sample covariance matrix; $\hat{\Sigma}:=X\T X/N_{\mathrm{tr}}=\sum_{i=1}^{N_{\mathrm{tr}}}\bx_i\bx_i\T/N_{\mathrm{tr}}\in\bbR^{p\times p}$. \\

$\Sigma$ & Population covariance matrix; $\Sigma:=\E_{\bs{u}\sim\mu}[\bx\bx\T]\in\bbR^{p\times p}$. \\

$k(\bs{u},\bs{u}')$ & Kernel function defined by $k(\bs{u},\bs{u}'):=\phi_a(\bs{u})\T\phi_a(\bs{u}') = \phi(\bs{u})\T\phi(\bs{u}')$. \\

$K$ & Kernel matrix $K:=XX\T\in\bbR^{N_{\mathrm{tr}}\times N_{\mathrm{tr}}}$, with entries $K_{ij}=k(\bs{u}_i,\bs{u}_j)$. \\

$\{(\xi_i,\psi_i)\}_{i=1}^p$ & Positive eigenvalues and corresponding orthonormal eigenfunctions in $L^2(\mu)$ of the kernel integral operator,\\
& ordered as $\xi_1 \geq \cdots \geq \xi_p > 0$, with $k(\bs{u},\bs{u}')=\sum_{i=1}^p\xi_i\psi_i(\bs{u})\psi_i(\bs{u}')$. \\

$\psi(\bs{u})$ & Vector of normalized eigenfunctions, $\psi(\bs{u}):=[\psi_1(\bs{u}),\ldots,\psi_p(\bs{u})]\T$. \\

$\bs{\beta}_*$ & Projected target vector in the kernel-eigenfunction coordinates; $f_*(\bs{u})=\bs{\beta}_*\T\psi(\bs{u})$ and $\bs{\beta}_* = \Sigma^{1/2}\bth_*$. \\

$\lambda > 0$ & $\ell_2$ Regularization parameter. \\

$\hat{\bth}$ & Ridge-regression estimator, $\hat{\bth}:=(\hat{\Sigma}+\lambda I_p)^{-1}X\T\by/N_{\mathrm{tr}}$. \\

$\hat{\bs{\alpha}}$ & Dual estimator, $\hat{\bs{\alpha}}:=(K+N_{\mathrm{tr}}\lambda I_{N_{\mathrm{tr}}})^{-1}\by$. \\

$\tilde{\gamma},\gamma$ & Finite-dimensional feature-to-sample aspect ratio $\tilde{\gamma}:=p/N_{\mathrm{tr}}$ and its limit $\gamma:=\lim_{p\to\infty}\tilde{\gamma}\in(0,\infty)$. \\

$\kappa_\lambda$ & Effective regularization parameter, defined as the unique positive solution of Eq.~\eqref{eq:self-consistent_eq}. \\

$\mathrm{df}_s(t)$ & Effective degrees of freedom of order $s$, $\mathrm{df}_s(t):=\Tr[\Sigma^s(\Sigma+tI_p)^{-s}]$, for $s\geq1$ and $t\geq0$. \\

$\eta(t),\eta_\lambda$ & Normalized second-order effective-DOF, $\eta(t):=\mathrm{df}_2(t)/N_{\mathrm{tr}}$ and $\eta_\lambda:=\eta(\kappa_\lambda)$. \\

$R_{\lambda,\hat{\Sigma}}$ & Conditional test risk $R_{\lambda,\hat{\Sigma}} :=\E_{(\bx_{\mathrm{ts}},y_{\mathrm{ts}})}\;\E_{\bs{\vep}_{\mathrm{tr}} \sim \calN(0, \sigma^2I_{N_{\mathrm{tr}}})} [(y_{\mathrm{ts}} - \hat{\bth}\T \bx_{\mathrm{ts}})^2]$. \\

$R^{\mathrm{DE}}_{\lambda,\Sigma}$ & Fixed-$\lambda$ deterministic equivalent of $R_{\lambda,\hat{\Sigma}}$ in Eq.~\eqref{eq:test_risk_DE}.

\end{longtable}
\newpage
\section{Correspondence with the Feature-Vector Generating Process in Ref.~\cite{louart2021spectral}}\label{appdx:sec:correspondence_feature_generation}
To apply the theoretical results from Ref.~\cite{louart2021spectral}, we must relate our feature-vector generating process (Assumption~\ref{assum:feature_generation}) to the mathematical framework in Ref.~\cite{louart2021spectral}. Specifically, Ref.~\cite{louart2021spectral} considers a feature-matrix generating map~$\Phi \colon (\bbR^{q'}, \|\cdot\|_2) \rightarrow (\bbR^{p_a \times N_{\mathrm{tr}}}, \|\cdot\|_F)$ and requires it to be Lipschitz continuous.

In our setting, we apply a common $C_1$-Lipschitz map~$\Omega_a \colon \bbR^q \to \bbR^{p_a}$ to each latent standard Gaussian vector~$\bs{z}_i \in \bbR^q$ to generate the feature vector~$\bx_i^a \in \bbR^{p_a}$. Thus, for any two latent standard Gaussian vectors~$\bs{z}_i, \bs{z}'_i \in \bbR^q$, we have
\begin{align}\label{eq:omega_lipschitz}
    \|\Omega_a(\bs{z}_i) - \Omega_a(\bs{z}'_i)\|_2 \leq C_1 \|\bs{z}_i - \bs{z}'_i\|_2 \,.
\end{align}

We can construct the map $\Phi$ by concatenating the inputs and outputs of $\Omega_a$. Let $Z = (\bs{z}_1\T, \bs{z}_2\T, \ldots, \bs{z}_{N_{\mathrm{tr}}}\T)\T \in \bbR^{q'}$ be the concatenated latent standard Gaussian vector, where $q' = q N_{\mathrm{tr}}$.
The map $\Phi$ produces the column-oriented matrix
\begin{align}
    X_{\mathrm{L},a}
    := \Phi(Z)
    = [\bx_1^a,\ldots,\bx_{N_{\mathrm{tr}}}^a]
    \in \bbR^{p_a\times N_{\mathrm{tr}}},
\end{align}
whose $i$-th column is $\bx_i^a = \Omega_a(\bs{z}_i)$.
The training feature matrix used in this paper is its transpose: $X_a = X_{\mathrm{L},a}\T$.

We now demonstrate that if the map $\Omega_a$ is $C_1$-Lipschitz with respect to the $\ell_2$ norm, the matrix-valued map $\Phi$ is also $C_1$-Lipschitz from the $\ell_2$ norm to the Frobenius norm.

Let $Z, Z' \in \bbR^{q'}$ be two concatenated latent standard Gaussian vectors. The squared $\ell_2$ distance between $Z$ and $Z'$ is the sum of the squared $\ell_2$ distances of their constituent vectors:
\begin{align}
    \|Z - Z'\|_2^2 = \sum_{i=1}^{N_{\mathrm{tr}}} \|\bs{z}_i - \bs{z}'_i\|_2^2 \,.
\end{align}
Similarly, because the squared Frobenius norm of a matrix is the sum of the squared $\ell_2$ norms of its columns, the squared Frobenius distance between the corresponding outputs is:
\begin{align}
    \|\Phi(Z) - \Phi(Z')\|_F^2 = \sum_{i=1}^{N_{\mathrm{tr}}} \|\Omega_a(\bs{z}_i) - \Omega_a(\bs{z}'_i)\|_2^2 \,.
\end{align}
Substituting the Lipschitz condition of $\Omega_a$ (Eq.~\eqref{eq:omega_lipschitz}) into the above expression yields:
\begin{align}
    \|\Phi(Z) - \Phi(Z')\|_F^2 \leq \sum_{i=1}^{N_{\mathrm{tr}}} C_1^2 \|\bs{z}_i - \bs{z}'_i\|_2^2 = C_1^2 \|Z - Z'\|_2^2 \,.
\end{align}
Taking the square root of both sides, we obtain:
\begin{align}
    \|\Phi(Z) - \Phi(Z')\|_F \leq C_1 \|Z - Z'\|_2 \,.
\end{align}
Since the Frobenius norm is invariant under transposition, $\|X_a - X_a'\|_F = \|X_{\mathrm{L},a} - X_{\mathrm{L},a}'\|_F$, the row-oriented map $Z \mapsto X_a = \Phi(Z)\T$ is also $C_1$-Lipschitz. Thus, the choice between the row- and column-oriented conventions does not affect the concentration property.
This confirms that the concentrated feature vectors generated under Assumption~\ref{assum:feature_generation} satisfy the global Lipschitz condition required by the framework in Ref.~\cite{louart2021spectral}.
The same argument applies even if we use different Lipschitz maps $\Omega_i^a$ with different Lipschitz constants $C_i < \infty$ for each $i$. In that case, the global map $\Phi$ would be $(\max_i C_i)$-Lipschitz.

\newpage
\section{Overparameterization in the Kernel Formulation}\label{appdx:sec:kernel_overparameterization}
In the primal formulation of ridge regression, the double descent phenomenon is governed by the ratio $\tilde{\gamma} = p/N_{\mathrm{tr}}$. The interpolation threshold occurs at $\tilde{\gamma} = 1$, or $p = N_{\mathrm{tr}}$.
However, as shown in Eq.~\eqref{eq:dual_estimator}, kernel ridge regression is solved by optimizing $N_{\mathrm{tr}}$ parameters ($\hat{\bs{\alpha}}$), regardless of the effective feature dimension $p$. This leads to a conceptual question: if the optimization problem always involves $N_{\mathrm{tr}}$ parameters, how does the model ``sense'' the overparameterization $p \gg N_{\mathrm{tr}}$?

The resolution lies in understanding that in the kernel formulation, the model's complexity is determined by the rank and the spectral properties of the kernel matrix $K$.
When $p < N_{\mathrm{tr}}$, the rank of $K$ is at most $p$, making $K$ singular.
Conversely, when $p \geq N_{\mathrm{tr}}$, assuming the features are generic, $K$ becomes full rank (rank $N_{\mathrm{tr}}$).
This affects the quality of the solution Eq.~\eqref{eq:dual_estimator}.

In the ridgeless limit $\lambda \to 0$ at fixed $p$ and $N_{\mathrm{tr}}$, the estimator simplifies to $\hat{\bs{\alpha}} = K^{-1}\by$ when $K$ is invertible. In this case, the model can perfectly interpolate the training dataset, achieving zero training error. However, if $K$ is singular, perfect interpolation is not generally possible.
In primal form, there are infinitely many solutions that can interpolate the training dataset in the overparameterized regime ($p > N_{\mathrm{tr}}$).
In the ridgeless limit $\lambda \to 0$ (again at fixed $p$ and $N_{\mathrm{tr}}$), the dual formulation selects the unique minimum-norm solution among all solutions that perfectly fit the training dataset:
\begin{align}
    \hat{\bth} = \underset{\bth \in \bbR^p}{\arg \min} \|\bth\|_2 \quad \text{s.t.} \quad \forall i \in [N_{\mathrm{tr}}], \; y_i = \bth\T \bx_i \,.
\end{align}
This selection principle constrains the model complexity, even in the absence of explicit regularization.

\section{Lipschitz Continuity of the Quantum Feature Map $\bs{u} \mapsto \bs{r}_a(\bs{u})$}\label{appdx:sec:lipschitz_qml_feature_map}
Using Theorem 2 from Ref.~\cite{kempkes2025double}, we prove that the quantum feature map $\bs{u} \mapsto \bs{r}_a(\bs{u}) := (r_i^a(\bs{u}))_{i=1}^{4^n}$ is a Lipschitz function of the classical input vector~$\bs{u} \in \bbR^d$ with respect to the $\ell_2$ norm.

By vectorizing $\rho(\bs{u})$ and $O$ as $\kett{\rho(\bs{u})}$ and $\kett{O}$, the QML model can be expressed as an inner product:
\begin{align}
    f(\bs{u}) = \evv{O}{\rho(\bs{u})} \,.
\end{align}

First, we can show that the quantum feature map $\bs{u} \mapsto \kett{\rho(\bs{u})}$ is a Lipschitz function of the classical input vector~$\bs{u} \in \bbR^d$ with respect to the $\ell_2$ norm.
This follows from the relation $\norm{\kett{\rho(\bs{u})} - \kett{\rho(\bs{v})}}_2 = \norm{\rho(\bs{u}) - \rho(\bs{v})}_F$, and the inequalities $\norm{AB}_F \leq \norm{A}_\infty \norm{B}_F$ and $\norm{\rho'}_F \leq 1$ for any density matrix $\rho'$, alongside the proof of Theorem 2 in Ref.~\cite{kempkes2025double}.
Thus, we have
\begin{align}
    \norm{\kett{\rho(\bs{u})} - \kett{\rho(\bs{v})}}_2
    = \norm{\rho(\bs{u}) - \rho(\bs{v})}_F
    \leq \sqrt{dL}\nu_{\max}\norm{\bs{u} - \bs{v}}_2 \,,
\end{align}
where $d$ is the dimension of the classical input vector, $L$ is the depth of the quantum circuit and $\nu_{\max}:=\max_{l,k}\|H_{lk}\|_\infty$, where $H_{lk}$ is the Hermitian generator of the corresponding rotation gate.

The coefficient vector $\bs{r}_a(\bs{u})$ and the vectorized density matrix $\kett{\rho(\bs{u})}$ are representations of the same state in different orthonormal bases. Indeed, by vectorizing both sides of $\rho(\bs{u}) = \sum_{i=1}^{4^n} r_i^a(\bs{u}) P_i$, we have $\kett{\rho(\bs{u})} = \sum_{i=1}^{4^n} r_i^a(\bs{u}) \kett{P_i}$. Considering a unitary matrix $M = (\kett{P_1}, \kett{P_2}, \ldots, \kett{P_{4^n}})$, we can express the relationship as
\begin{align}
    \kett{\rho(\bs{u})} = M \bs{r}_a(\bs{u}) \,.
\end{align}
Therefore, for any $\bs{u}, \bs{v} \in \bbR^d$, we have $\|\bs{r}_a(\bs{u})\|_2 = \|\kett{\rho(\bs{u})}\|_2 = \|\rho(\bs{u})\|_F$ and
\begin{align}
    \norm{\bs{r}_a(\bs{u}) - \bs{r}_a(\bs{v})}_2
    = \norm{M\bs{r}_a(\bs{u}) - M\bs{r}_a(\bs{v})}_2
    = \norm{\kett{\rho(\bs{u})} - \kett{\rho(\bs{v})}}_2
    \leq \sqrt{dL}\nu_{\max}\norm{\bs{u} - \bs{v}}_2 \,.
\end{align}

\newpage
\section{Effective Feature Dimension of the Tensor Product Ansatz}
\label{appdx:sec:effective_dim_TPA}
In this appendix, we characterize the effective feature dimension of the Tensor Product Ansatz (TPA) consisting only of single-qubit $R_X$ rotations and containing no entangling gates.
The key observation is that this circuit architecture restricts the Bloch vector of each qubit to the $YZ$-plane, eliminating the $X$-component of the Pauli expansion.
Although the circuit architecture restricts the Pauli feature vector to a subspace of dimension at most $3^n$, whether this upper bound is attained depends on the input distribution.

Since consecutive rotations about the same axis can be combined, the output state can be written as
\begin{align}
    \rho(\bs{u})
    &=
    \bigotimes_{j=1}^n
    \rho_j(\vartheta_j(\bs{u})),\\
    \rho_j(\vartheta)
    &:=
    R_X(\vartheta)\dyad{0}R_X^\dagger(\vartheta)
    =
    \frac{1}{2}
    (I-\sin\vartheta\,Y+\cos\vartheta\,Z) \,,
    \label{eq:TPA_single_qubit_state}
\end{align}
where $\vartheta_j(\bs{u})$ denotes the effective rotation angle on the $j$-th qubit.
For $a \in \{0,1,2\}$, define
\begin{alignat}{3}
    q_0 &:= \frac{I}{\sqrt{2}} \,, &\quad
    q_1 &:= \frac{Y}{\sqrt{2}} \,, &\quad
    q_2 &:= \frac{Z}{\sqrt{2}} \,,\\
    h_0(\vartheta) &:= \frac{1}{\sqrt{2}} \,, &\quad
    h_1(\vartheta) &:= -\frac{\sin\vartheta}{\sqrt{2}} \,, &\quad
    h_2(\vartheta) &:= \frac{\cos\vartheta}{\sqrt{2}} \,.
\end{alignat}
For each multi-index
$\bs{a} = (a_1,\ldots,a_n) \in \{0,1,2\}^n$, let
\begin{align}
    Q_{\bs{a}}
    :=
    \bigotimes_{j=1}^n q_{a_j} \,,
    \quad
    g_{\bs{a}}(\bs{u})
    :=
    \prod_{j=1}^n
    h_{a_j}(\vartheta_j(\bs{u})) \,.
\end{align}
Taking the tensor product in Eq.~\eqref{eq:TPA_single_qubit_state} gives
\begin{align}
    \rho(\bs{u})
    =
    \sum_{\bs{a} \in \{0,1,2\}^n}
    g_{\bs{a}}(\bs{u})\, Q_{\bs{a}} \,.
    \label{eq:TPA_Pauli_expansion}
\end{align}
Thus, every Pauli string containing at least one $X$ factor has an identically zero coefficient, and the Pauli feature vector is supported on the coordinate subspace corresponding to $\{I,Y,Z\}^{\otimes n}$.
For example, in the two-qubit case, the potentially nonzero
coordinates correspond to
\begin{align}
    II,\quad IY,\quad IZ,\quad
    YI,\quad YY,\quad YZ,\quad
    ZI,\quad ZY,\quad ZZ,
\end{align}
illustrating the architecture-induced count $3^2=9$.
This coordinate count alone, however, does not imply that the corresponding coefficient functions are linearly independent under an arbitrary input distribution.

We define
\begin{align}
    \bs{g}(\bs{u})
    := (g_{\bs{a}}(\bs{u}))_{\bs{a} \in \{0,1,2\}^n}
    \in \bbR^{3^n} \,,
    \quad
    G_\mu
    :=
    \E_{\bs{u}\sim\mu}[\bs{g}(\bs{u})\bs{g}(\bs{u})\T] \,.
\end{align}
Recall that we have defined the following quantities in Section~\ref{subsec:eigendecomposition} (with the subscript $\mu$ added to emphasize the dependence on the input distribution):
\begin{align}
    \Sigma_{a,\mu}
    :=
    \E_{\bs{u}\sim\mu}
    [\bs{r}_a(\bs{u})\bs{r}_a(\bs{u})\T] \,,
    \qquad
    p_\mu
    :=
    \rank(\Sigma_{a,\mu}) \,,
\end{align}
where $\bs{r}_a(\bs{u}) \in \bbR^{4^n}$ is the coefficient vector of
$\rho(\bs{u})$ in the normalized Pauli basis and we add the subscript $\mu$ to emphasize the dependence on the input distribution.

Let $J \in \bbR^{4^n\times3^n}$ denote the coordinate embedding from the coordinates indexed by $\{I,Y,Z\}^{\otimes n}$ into the full Pauli-coordinate space $\{I,X,Y,Z\}^{\otimes n}$, which has $1$ in the $(\bs{a},\bs{a})$ entry for each $\bs{a} \in \{0,1,2\}^n$ and $0$ elsewhere.
Then $J\T J = I_{3^n}$, and Eq.~\eqref{eq:TPA_Pauli_expansion} implies
\begin{align}
    \bs{r}_a(\bs{u}) = J\bs{g}(\bs{u}) \,.
\end{align}
Consequently,
\begin{align}
    \Sigma_{a,\mu}
    :=
    \E_{\bs{u}\sim\mu}[\bs{r}_a(\bs{u})\bs{r}_a(\bs{u})\T]
    = J
    \E_{\bs{u}\sim\mu}[\bs{g}(\bs{u})\bs{g}(\bs{u})\T]
    J\T
    = J G_\mu J\T \,.
\end{align}
Since $J$ is a coordinate embedding, $JG_\mu J\T$ is obtained
by placing the entries of $G_\mu$ in the coordinates indexed by
$\{I,Y,Z\}^{\otimes n}$ and setting all remaining rows and columns to zero. Thus, we have
\begin{align}
    p_\mu
    =
    \rank(\Sigma_{a,\mu})
    =
    \rank(G_\mu)
    \leq 3^n \,.
    \label{eq:TPA_effective_dimension_bound}
\end{align}

The upper bound is attained if and only if the coefficient
functions are linearly independent in $L^2(\mu)$:
\begin{align}
    p_\mu=3^n
    \quad\Longleftrightarrow\quad
    \{g_{\bs{a}} \,:\, \bs{a} \in \{0,1,2\}^n\}
    \text{ is linearly independent in }L^2(\mu) \,.
    \label{eq:TPA_full_rank_condition}
\end{align}
This follows from the following equivalence:
\begin{align}
    &\phantom{\iff}\;\; p_\mu = \rank(G_\mu) = 3^n \\[5pt]
    &\iff G_\mu \text{ is positive definite} \\[5pt]
    &\iff \forall \bs{c} \in \bbR^{3^n}, \quad
    \qty[
        \bs{c}\T G_\mu \bs{c} = 0
        \quad\Longrightarrow\quad
        \bs{c} = \bs{0}
    ] \\
    &\iff \forall \bs{c} \in \bbR^{3^n}, \quad
    \qty[
        \int_{\calU}\qty(
            \sum_{\bs{a} \in \{0,1,2\}^n} c_{\bs{a}}g_{\bs{a}}(\bs{u})
        )^2
        \ind\mu(\bs{u}) = 0
        \quad\Longrightarrow\quad
        \bs{c} = \bs{0}
    ] \\
    &\iff \forall \bs{c} \in \bbR^{3^n}, \quad
    \qty[
        \sum_{\bs{a} \in \{0,1,2\}^n}
        c_{\bs{a}}g_{\bs{a}}(\bs{u}) = 0 \quad \mu\text{-a.e.}
        \quad\Longrightarrow\quad
        \bs{c} = \bs{0}
    ] \\
    &\iff \{g_{\bs{a}} \,:\, \bs{a} \in \{0,1,2\}^n\}
    \text{ is linearly independent in }L^2(\mu) \,,
\end{align}
where we used the following identity for the third equivalence:
\begin{align}
    \bs{c}\T G_\mu\bs{c}
    =
    \int_{\calU}\qty(
        \sum_{\bs{a} \in \{0,1,2\}^n} c_{\bs{a}}g_{\bs{a}}(\bs{u})
    )^2
    \ind\mu(\bs{u}) \quad\text{for every}\quad \bs{c} \in \bbR^{3^n} \,.
    \label{eq:TPA_Gram_quadratic_form}
\end{align}

Therefore, the TPA architecture guarantees only the upper bound
$p_\mu\leq3^n$.
The equality $p_\mu=3^n$ additionally requires the effective-angle distribution induced by $\mu$ to be sufficiently nondegenerate, so that the coefficient functions $\{g_{\bs{a}} \,:\, \bs{a} \in \{0,1,2\}^n\}$ are linearly independent in $L^2(\mu)$.
Thus, the effective feature dimension is determined jointly by the circuit architecture and the input distribution.

\newpage
\section{Derivation of the Bias--Variance Decomposition of the Test Risk}\label{appdx:sec:qml_analysis}
Since the estimator of the quantum ridge regression is given by Eq.~\eqref{eq:ridge_regression_estimator}, the term $\bth_* - \hat{\bth}$, which appears in the test risk, can be expressed as follows:
\begin{align}
    \bth_* - \hat{\bth}
    &= \bth_* - (\hat{\Sigma} + \lambda I_p)^{-1} \frac{1}{N_{\mathrm{tr}}}X\T \by \nonumber\\
    &= \bth_* - (\hat{\Sigma} + \lambda I_p)^{-1} \frac{1}{N_{\mathrm{tr}}}X\T(X\bth_* + \bs{\vep}_{\mathrm{tr}}) \nonumber\\
    &= \bth_* - (\hat{\Sigma} + \lambda I_p)^{-1} \hat{\Sigma}\bth_* - (\hat{\Sigma} + \lambda I_p)^{-1} \frac{1}{N_{\mathrm{tr}}}X\T \bs{\vep}_{\mathrm{tr}} \nonumber\\
    &= \bth_* - (\hat{\Sigma} + \lambda I_p)^{-1}(\hat{\Sigma} + \lambda I_p - \lambda I_p)\bth_* - (\hat{\Sigma} + \lambda I_p)^{-1} \frac{1}{N_{\mathrm{tr}}}X\T \bs{\vep}_{\mathrm{tr}} \nonumber\\
    &=
    \underbrace{
        \lambda (\hat{\Sigma} + \lambda I_p)^{-1}\bth_*
    }_{\bth_* - \bar{\bth}}
    -
    \underbrace{
        (\hat{\Sigma} + \lambda I_p)^{-1} \frac{1}{N_{\mathrm{tr}}}X\T \bs{\vep}_{\mathrm{tr}}
    }_{\hat{\bth} - \bar{\bth}} \,.
    \label{eq:decomposition_of_theta_difference}
\end{align}

Substituting this decomposition into the test risk $R_{\lambda,\hat{\Sigma}}$ gives the following decomposition:
\begin{align}
    R_{\lambda,\hat{\Sigma}}
    &:= \E_{\bx_{\mathrm{ts}},\vep_{\mathrm{ts}},\bs{\vep}_{\mathrm{tr}}}
    \left[ (y_{\mathrm{ts}} - \hat{\bth}\T \bx_{\mathrm{ts}})^2 \right] \nonumber\\
    &= \E_{\bs{\vep}_{\mathrm{tr}} \sim \calN(0, \sigma^2I_{N_{\mathrm{tr}}})}
        [(\bth_* - \hat{\bth})\T \Sigma (\bth_* - \hat{\bth})]
    + \sigma^2 \nonumber\\
    &= (\bth_* - \bar{\bth})\T \Sigma (\bth_* - \bar{\bth})
    + \E_{\bs{\vep}_{\mathrm{tr}}} \left[ (\hat{\bth} - \bar{\bth})\T \Sigma (\hat{\bth} - \bar{\bth}) \right]
    + \sigma^2 \nonumber\\
    &= \lambda^2\, \bth_*\T(\hat{\Sigma} + \lambda I_p)^{-1}\Sigma(\hat{\Sigma} + \lambda I_p)^{-1}\bth_* \nonumber\\
    &\quad + \frac{1}{N_{\mathrm{tr}}}\E_{\bs{\vep}_{\mathrm{tr}} \sim \calN(0, \sigma^2I_{N_{\mathrm{tr}}})}
    \qty[
        \Tr[
            (\hat{\Sigma} + \lambda I_p)^{-1}\Sigma(\hat{\Sigma} + \lambda I_p)^{-1}\frac{1}{N_{\mathrm{tr}}}X\T \bs{\vep}_{\mathrm{tr}} \bs{\vep}_{\mathrm{tr}}\T X
        ]
    ]
    + \sigma^2 \nonumber\\
    &= \lambda^2\, \bth_*\T(\hat{\Sigma} + \lambda I_p)^{-1}\Sigma(\hat{\Sigma} + \lambda I_p)^{-1}\bth_*
    + \frac{\sigma^2}{N_{\mathrm{tr}}}
    \Tr[
            (\hat{\Sigma} + \lambda I_p)^{-1}\Sigma(\hat{\Sigma} + \lambda I_p)^{-1}\hat{\Sigma}
    ] + \sigma^2 \nonumber\\
    &= \lambda^2\, \bth_*\T(\hat{\Sigma} + \lambda I_p)^{-1}\Sigma(\hat{\Sigma} + \lambda I_p)^{-1}\bth_*
    + \frac{\sigma^2}{N_{\mathrm{tr}}}\Tr[(\hat{\Sigma} + \lambda I_p)^{-2}\hat{\Sigma}\Sigma]
    + \sigma^2 \nonumber\\
    &= \lambda^2\, \bs{\beta}_*\T\Sigma^{-1/2}(\hat{\Sigma} + \lambda I_p)^{-1}\Sigma(\hat{\Sigma} + \lambda I_p)^{-1}\Sigma^{-1/2}\bs{\beta}_*
    + \frac{\sigma^2}{N_{\mathrm{tr}}}\Tr[(\hat{\Sigma} + \lambda I_p)^{-2}\hat{\Sigma}\Sigma]
    + \sigma^2 \,.
\end{align}
\newpage
\section{Definitions and Assumptions in Ref.~\cite{louart2021spectral}}\label{appdx:sec:definitions_assumptions_louart2021}
To make this work self-contained, we review the definitions and assumptions in Ref.~\cite{louart2021spectral}.
A central object in this work is the resolvent of the sample covariance matrix $\hat{\Sigma} := \frac{1}{N_{\mathrm{tr}}}X\T X = \frac{1}{N_{\mathrm{tr}}}X_{\mathrm L}X_{\mathrm L}\T$ (Our matrix \(X\) stores samples as rows; it is the transpose of the column-oriented training feature matrix $X_{\mathrm L}$ used in Ref.~\cite{louart2021spectral}), defined for a complex parameter $z \notin \mathrm{Sp}(\hat{\Sigma})$ by
\begin{align}
    R(z) := \qty(\hat{\Sigma} - zI_p)^{-1} \,.
\end{align}
Because $R(z)$ is a random matrix, it is useful to approximate it by a deterministic equivalent $\widetilde{R}(z)$, such that the difference between $R(z)$ and $\widetilde{R}(z)$ vanishes in the high-dimensional limit.
In Ref.~\cite{louart2021spectral}, this deterministic equivalent is given by
\begin{align}
    \widetilde{R}(z) := \qty(\frac{1}{N_{\mathrm{tr}}}\sum_{i=1}^{N_{\mathrm{tr}}} \frac{z\Sigma_i}{\tilde{\Lambda}_i^z} - zI_p)^{-1} \,,
\end{align}
where $\Sigma_i = \E[\bx_i\bx_i\T]$ and the elements $(\tilde{\Lambda}_1^z, \tilde{\Lambda}_2^z, \ldots, \tilde{\Lambda}_{N_{\mathrm{tr}}}^z) =: \tilde{\Lambda}^z$ are the unique solutions to the fixed-point equation:
\begin{align}
    \forall i \in [N_{\mathrm{tr}}], \quad \tilde{\Lambda}_i^z = z - \frac{1}{N_{\mathrm{tr}}}\Tr\qty[\Sigma_i\qty(I_p - \frac{1}{N_{\mathrm{tr}}}\sum_{j=1}^{N_{\mathrm{tr}}} \frac{\Sigma_j}{\tilde{\Lambda}_j^z})^{-1}] \,.
    \label{eq:self_consistent_eq_louart}
\end{align}

\begin{assumption}\label{assum:0.3-0.7}
Ref.~\cite{louart2021spectral} establishes the convergence to this deterministic equivalent under the following five core assumptions regarding the training feature matrix $X$ (Assumption~0.4 is slightly modified from the original version since we consider the row-oriented training feature matrix $X = X_{\mathrm L}\T$):
\begin{itemize}
    \item \textbf{Assumption~0.3 (Dimensionality):} The feature dimension $p$ grows at most linearly with the number of samples $N_{\mathrm{tr}}$, i.e., there exists $K_1 > 0$ such that $\forall N_{\mathrm{tr}} \in \bbN, p \le K_1 N_{\mathrm{tr}}$.
    \item \textbf{Assumption~0.4 (Concentration):} The training feature matrix $X$ satisfies a concentration of measure property. Specifically, there exist $C, c > 0$ such that for any $1$-Lipschitz mapping $f \colon (\bbR^{N_{\mathrm{tr}} \times p}, \|\cdot\|_F) \to (\bbR, |\cdot|)$,
    \begin{align}
        \bbP(|f(X) - \E[f(X)]| \ge t) \le C e^{-(t/c)^2} \,.
        \label{eq:concentration_inequality}
    \end{align}
    \item \textbf{Assumption~0.5 (Independence):} The feature vectors $\bx_1, \dots, \bx_{N_{\mathrm{tr}}} \in \bbR^p$ are independent.
    \item \textbf{Assumption~0.6 (Bounded Means):} The expectations of the feature vectors are uniformly bounded, meaning there exists $K_2 > 0$ such that $\forall N_{\mathrm{tr}} \in \bbN,\; \forall i \in [N_{\mathrm{tr}}],\; \norm{\E[\bx_i]}_2 \le K_2$.
    \item \textbf{Assumption~0.7 (Bounded Covariances):} The covariance matrices are uniformly positive definite, meaning there exists $K_3 > 0$ such that $\forall N_{\mathrm{tr}} \in \bbN,\; \forall i \in [N_{\mathrm{tr}}],\; \Sigma_i \ge K_3 I_p$.
\end{itemize}
\end{assumption}

\begin{remark}
    Our feature-vector generating process, formalized in Assumption~\ref{assum:feature_generation}, is designed to satisfy the above assumptions.
    Since we consider $p/N_{\mathrm{tr}} \to \gamma \in (0, \infty)$ as $N_{\mathrm{tr}} \to \infty$, \textbf{Assumption~0.3} is satisfied.
    \textbf{Assumption~0.5} is satisfied because we generate the feature vectors by applying a common map $\Omega$ to independent latent standard Gaussian vectors $\{\bs{z}_i\}_{i=1}^{N_{\mathrm{tr}}}$.
    Furthermore, by proving that our effective quantum feature map is Lipschitz continuous and norm-bounded, we guarantee that the resulting feature vectors satisfy \textbf{Assumption~0.4} and \textbf{Assumption~0.6}.
\end{remark}

\begin{remark}\label{remark:assumption_0.7}
    We clarify that we removed \textbf{Assumption~0.7} in our setting.
    Since we apply a common map $\Omega$ for all samples in our framework, $\Sigma_i = \Sigma$ for all $i \in [N_{\mathrm{tr}}]$.
    Assumption~0.7 is only used in the proof of Proposition~6.5 to show $\norm*{\hat{\Lambda}^z - \tilde{\Lambda}^z} \leq \calO(\norm*{\hat{\Lambda}^z - I^z(\hat{\Lambda}^z)})$. However, since $\Sigma_i = \Sigma$ for all $i$, we have $\tilde{\Lambda}_i^z = \tilde{\Lambda}_1^z$ for all $i$, and thus $\sup_{i \in [N_{\mathrm{tr}}]}{\Im(\tilde{\Lambda}_i^z)} / \inf_{i \in [N_{\mathrm{tr}}]}{\Im(\tilde{\Lambda}_i^z)} = 1$, which allows us to bypass the need for Assumption~0.7.
\end{remark}

\begin{remark}
    As an example of a training feature matrix $X$ satisfying Assumption~0.4, Ref.~\cite{louart2021spectral} provides a random matrix~$X = \Phi(Z)$, where $Z \sim \calN(0, I_{q'})$ is a standard Gaussian vector and $\Phi \colon (\bbR^{q'}, \|\cdot\|_2) \to (\bbR^{N_{\mathrm{tr}} \times p}, \|\cdot\|_F)$ is a $C_1$-Lipschitz function.
    In this example, the constants in the concentration inequality~\eqref{eq:concentration_inequality} are $C = 2$ and $c = C_1\sqrt{2}$.
    In Appendix~\ref{appdx:sec:correspondence_feature_generation}, we show that our feature-vector generating process (Assumption~\ref{assum:feature_generation}) is a special case of this example, and thus satisfies the concentration property required by Assumption~0.4.
\end{remark}

\newpage
\section{From Theorem~0.9 in Ref.~\cite{louart2021spectral} to
Proposition~\ref{prop:resolvent_deterministic_equivalent}}
\label{appdx:sec:theorem_0.9_to_deterministic_equivalent}
In this section, we derive the two scalar interfaces to the resolvent concentration theorem~\ref{thm:0.9} that are used later. The normalized-trace interface is used for the variance and the bounded quadratic-form interface is used for the bias.

\begin{figure}[H]
    \centering
    \begin{tikzpicture}[font=\small, node distance=0pt]
        \tikzset{
            bx/.style={draw, rounded corners=2pt, align=center, inner sep=4pt},
            lb/.style={align=left, font=\scriptsize\itshape}
        }
        \node[bx] (t09) at (0,0)
            {Theorem~\ref{thm:0.9}: concentration at the level of functionals,\\
             measured by the Lipschitz constant};
        \node[bx] (prop) at (0,-1.8)
            {Proposition~\ref{prop:resolvent_deterministic_equivalent}: \;
            $Q_p(\lambda)\asymp_F\overline Q_p(\lambda)$};
        \node[bx] (bias) at (-3.4,-3.8)
            {$C_p=\bs{v}_p\bs{v}_p\T$\\ $\norm{C_p}_F=\norm{\bs{v}_p}_2^2$\\[2pt] \textbf{bias}};
        \node[bx] (var) at (3.4,-3.8)
            {$C_p=D_p/p$\\ $\norm{C_p}_F\leq\norm{D_p}_\infty/\sqrt p$\\[2pt] \textbf{variance}};
        \draw[->] (t09) -- (prop)
        node[lb, midway, right=4pt, xshift=2pt] {
            Lemma~\ref{lemma:borel_cantelli_convergence}: Borel--Cantelli\\
            Lemma~\ref{lemma:frobenius_dual_pairing}: $\operatorname{Lip}(u_p)=\norm{C_p}_F$
            };
        \draw[->] (prop) -- (bias);
        \draw[->] (prop) -- (var);
    \end{tikzpicture}
    \caption{Logical structure of the derivation. The final step branches according to the choice of test matrix and the evaluation of its Frobenius norm.}
    \label{fig:de_logical_chain}
\end{figure}

\subsection{Resolvent concentration and scalar interfaces}
Before stating Theorem~0.9, we define a seminorm used in Ref.~\cite{louart2021spectral}.
We define \(\calF(\bbC, \bbC^{p \times p})\) to be the space of matrix-valued functions \(f\colon\bbC\to\bbC^{p \times p}\).
\begin{definition}[Seminorm \(\norm{\cdot}_{F,S^{\epsilon}}\)]
    \label{def:seminorm_F_S_epsilon}
    Let \(\mu_1 \geq \mu_2 \geq \cdots \geq \mu_p\) be the eigenvalues of \(\hat{\Sigma}_p\), and define
    \begin{align}
        S
        &:=
        \set{\E[\mu_i]}[i\in[p]] \,.
    \end{align}
    For \(\epsilon > 0\), let
    \begin{align}
        S^{\epsilon}
        &:=
        \set{z \in \bbC}[
            \operatorname{dist}(z,S)\leq\epsilon
        ] \,.
    \end{align}
    For \(f \in \calF(\bbC,\bbC^{p\times p})\), define
    \begin{align}
        \norm{f}_{F,S^{\epsilon}}
        &:=
        \sup_{z \in \bbC \setminus S^{\epsilon}}
        \norm{f(z)}_F \,.
    \end{align}
    The dependence of \(S\) on \(p\) is left implicit.
\end{definition}

Because \(\hat{\Sigma}_p \succeq 0\), we have \(S \subset [0,\infty)\). Therefore, for any \(0 < \epsilon < \lambda\),
\begin{align}
    \operatorname{dist}(-\lambda,S)
    \geq \lambda > \epsilon
    \iff
    -\lambda \in \bbC \setminus S^{\epsilon}.
    \label{eq:minus_lambda_outside_S_epsilon}
\end{align}
Thus, for any \(f\in\calF(\bbC,\bbC^{p \times p})\) which is analytic on \(\bbC \setminus S^{\epsilon}\), we have \(\norm{f(-\lambda)}_F \leq \norm{f}_{F,S^{\epsilon}}\).

Ref.~\cite{louart2021spectral} gives the following concentration of the resolvent \(R(\cdot)\) around its deterministic approximation \(\widetilde{R}(\cdot)\).
\begin{theorem}[Theorem~0.9 in Ref.~\cite{louart2021spectral}]\label{thm:0.9}
    Under Assumption~\ref{assum:0.3-0.7}, there exist constants \(C,c>0\)
    such that, for every linear functional
    \begin{align}
        v\colon
        \bigl(
            \calF(\bbC,\bbC^{p\times p}),
            \norm{\cdot}_{F,S^{\epsilon}}
        \bigr)
        \longrightarrow
        (\bbR,|\cdot|)
    \end{align}
    with Lipschitz constant at most one, and every \(s > 0\),
    \begin{align}
        \bbP\left(
            |v(R - \widetilde{R})| \geq s
        \right)
        \leq
        C\exp(-cN_{\mathrm{tr}}s^2)
        +
        C\exp(-cN_{\mathrm{tr}}) \,.
        \label{eq:concentration_resolvent}
    \end{align}
\end{theorem}

We use the following Borel--Cantelli convergence criterion to convert the concentration inequality~\eqref{eq:concentration_resolvent} into an almost-sure convergence statement in Proposition~\ref{prop:uniform_linear_functional_resolvent_convergence} below.
\begin{lemma}[Borel--Cantelli convergence criterion]
    \label{lemma:borel_cantelli_convergence}
    Let \((\Omega,\calF,\bbP)\) be a probability space, and let
    \(\{Y_p\}_p\) be a sequence of real- or complex-valued random
    variables on this space. If, for every \( s>0\),
    \begin{align}
        \sum_{p=1}^{\infty}
        \bbP(|Y_p| \geq s)
        < \infty \,,
    \end{align}
    where $\bbP(|Y_p| \geq s) := \bbP(\set{\omega \in \Omega}[|Y_p(\omega)|\geq s])$, then
    \begin{align}
        Y_p
        \longrightarrow
        0
        \qquad
        \text{a.s. as } p \to \infty \,.
    \end{align}
    Here and below, ``a.s.'' means almost surely; specifically, the last
    conclusion means that
    \begin{align}
        \bbP\left(
            \set*{\omega \in \Omega}[
                Y_p(\omega) \longrightarrow 0
                \text{ as } p \to \infty
            ]
        \right)
        =
        1 \,.
    \end{align}
\end{lemma}

We first extract the common probabilistic engine from
Theorem~\ref{thm:0.9}.
\begin{proposition}[Almost-sure scalar resolvent convergence]
    \label{prop:uniform_linear_functional_resolvent_convergence}
    Let \(u_p\) be a linear functional satisfying
    \begin{align}
        \sup_p\operatorname{Lip}(u_p)
        \leq L
        < \infty \,.
        \label{eq:uniform_linear_functional_lipschitz_bound}
    \end{align}
    Under Assumption~\ref{assum:0.3-0.7},
    \begin{align}
        u_p(R_p - \widetilde{R}_p)
        \longrightarrow
        0
        \qquad
        \text{a.s. as } p \to \infty \,.
        \label{eq:uniform_linear_functional_resolvent_convergence}
    \end{align}
\end{proposition}

\begin{proof}
    If \(L = 0\), every \(u_p\) vanishes on differences and the conclusion is immediate. Suppose \(L > 0\). Since \(\tilde{u}_p := u_p/L\) has Lipschitz constant at most one, Theorem~\ref{thm:0.9} gives, for every fixed \(s > 0\),
    \begin{align}
        \bbP\left(
            |u_p(R_p - \widetilde{R}_p)| \geq s
        \right)
        \leq
        \bbP\left(
            |\tilde{u}_p(R_p - \widetilde{R}_p)| \geq \frac{s}{L}
        \right)
        \leq
        C\exp\!\left(
            -cN_{\mathrm{tr}}\frac{s^2}{L^2}
        \right)
        +
        C\exp(-cN_{\mathrm{tr}}) \,.
        \label{eq:uniform_linear_functional_concentration}
    \end{align}
    Assumption~0.3 in Ref.~\cite{louart2021spectral} implies that there exists a constant $k > 0$ such that $\forall p \in \bbN,\, kp \leq N_{\mathrm{tr}}$.
    Therefore,
    \begin{align}
        \bbP\left(
            |u_p(R_p - \widetilde{R}_p)|\geq s
        \right)
        &\leq
        C\exp\!\left(
            -ckp\frac{s^2}{L^2}
        \right)
        +
        C\exp(-ckp) \,.
    \end{align}
    Lemma~\ref{lemma:borel_cantelli_convergence} now gives Eq.~\eqref{eq:uniform_linear_functional_resolvent_convergence}.
\end{proof}

The class of observables directly covered by Proposition~\ref{prop:uniform_linear_functional_resolvent_convergence} includes those whose Lipschitz constants are uniformly bounded in \(p\).
The following lemma computes that constant for every trace-pairing functional $u_p$ and thereby identifies a sufficient uniformly bounded class of test matrices \(\{C_p\}\) for which the almost-sure convergence~\eqref{eq:uniform_linear_functional_resolvent_convergence} holds.
\begin{lemma}[Lipschitz constant of the trace-pairing functional]
    \label{lemma:frobenius_dual_pairing}
    Fix \(\lambda > \epsilon > 0\) and let \(C_p \in \bbR^{p\times p}\) be deterministic.
    Define the linear functional
    \begin{align}
        u_p(A)
        :=
        \Re\!\left(
            \Tr[C_p A(-\lambda)]
        \right) \,,
        \qquad
        A \in \calF(\bbC,\bbC^{p\times p}) \,.
        \label{eq:trace_pairing_functional}
    \end{align}
    Then, on the subspace of functions analytic on \(\bbC\setminus S^{\epsilon}\),
    \begin{align}
        \operatorname{Lip}(u_p)
        =
        \norm{C_p}_F \,,
        \label{eq:frobenius_dual_pairing}
    \end{align}
    the Lipschitz constant being taken with respect to \(\norm{\cdot}_{F,S^{\epsilon}}\).
\end{lemma}

\begin{proof}
    For the upper bound, the Cauchy--Schwarz inequality for the Frobenius inner
    product gives, for any \(A, B\),
    \begin{align}
        |u_p(A) - u_p(B)|
        &\leq
        |\Tr[C_p\{A(-\lambda) - B(-\lambda)\}]|
        \\
        &\leq
        \norm{C_p}_F
        \norm{A(-\lambda) - B(-\lambda)}_F
        \\
        &\leq
        \norm{C_p}_F
        \norm{A - B}_{F,S^{\epsilon}} \,,
    \end{align}
    where the last step uses \(-\lambda \in \bbC \setminus S^{\epsilon}\), which holds because \(\hat{\Sigma}_p \succeq 0\) implies \(S \subset [0,\infty)\) and hence \(\operatorname{dist}(-\lambda,S) \geq \lambda > \epsilon\).
    This is the only place in this lemma where the condition \(0 < \epsilon < \lambda\) is used.
    If \(C_p = 0\), the claim is immediate. Assume henceforth that \(C_p \neq 0\).
    To check the tightness of the bound, take \(B = 0\) and let \(A\) be the constant function \(A\equiv C_p\T/\norm{C_p}_F\), which is analytic on all of \(\bbC\) and satisfies \(\norm{A}_{F,S^{\epsilon}}=1\); then \(u_p(A) = \Tr[C_p C_p\T]/\norm{C_p}_F=\norm{C_p}_F\).
\end{proof}

Combining Lemma~\ref{lemma:frobenius_dual_pairing} with Proposition~\ref{prop:uniform_linear_functional_resolvent_convergence} yields the master interface.
\begin{corollary}[Frobenius-dual resolvent convergence]
    \label{cor:frobenius_dual_resolvent_convergence}
    Fix \(\lambda > \epsilon > 0\) and let \(C_p \in \bbR^{p\times p}\) be deterministic matrices with
    \begin{align}
        M_F
        :=
        \sup_p\norm{C_p}_F
        < \infty \,.
        \label{eq:frobenius_bounded_test_matrices}
    \end{align}
    Under Assumption~\ref{assum:0.3-0.7},
    \begin{align}
        \Tr\left[
            C_p
            \bigl\{
                R_p(-\lambda)-\widetilde{R}_p(-\lambda)
            \bigr\}
        \right]
        \longrightarrow
        0
        \qquad
        \text{a.s.}
        \label{eq:frobenius_dual_resolvent_convergence}
    \end{align}
    Equivalently, \(R_p(-\lambda)\asymp_F\widetilde{R}_p(-\lambda)\) and \(Q_p(\lambda)\asymp_F\overline Q_p(\lambda)\) in the sense of Definition~\ref{def:frobenius_deterministic_equivalent}.
\end{corollary}
\begin{proof}
    By Lemma~\ref{lemma:frobenius_dual_pairing} the functionals \(u_p\) of
    Eq.~\eqref{eq:trace_pairing_functional} satisfy
    \(\sup_p\operatorname{Lip}(u_p) = \sup_p\norm{C_p}_F=M_F < \infty\), so
    Proposition~\ref{prop:uniform_linear_functional_resolvent_convergence} gives
    \(u_p(R_p - \widetilde{R}_p)\to0\) a.s.
    Both \(R_p(-\lambda)\) and \(\widetilde{R}_p(-\lambda)\) are real symmetric, so taking
    the real part does not change the value of the trace, and
    Eq.~\eqref{eq:frobenius_dual_resolvent_convergence} follows.
    The final claim for \(Q_p\) follows because
    \(Q_p(\lambda) = \lambda R_p(-\lambda)\) and
    \(\overline Q_p(\lambda) = \lambda\widetilde{R}_p(-\lambda)\), so that the two traces
    differ only by the fixed factor \(\lambda\).
\end{proof}

The following table summarizes the two scalar interfaces obtained below by specializing Corollary~\ref{cor:frobenius_dual_resolvent_convergence}.

\begin{table}[H]
    \centering
    \setlength{\tabcolsep}{15pt}
    \begin{tabular}{llll}
    \toprule
    Interface & Deterministic input & Test matrix \(C_p\) & \(\norm{C_p}_F\) \\
    \midrule
    Normalized trace
        & \(D_p\), with \(\sup_p\norm{D_p}_\infty < \infty\)
        & \(D_p/p\)
        & \(\leq\norm{D_p}_\infty/\sqrt p\) \\
    Bounded bilinear form
        & \(u_p, v_p\), with \(\sup_p \norm{u_p}_2 \norm{v_p}_2 < \infty\)
        & \(v_p u_p\T\)
        & \(= \norm{u_p}_2 \norm{v_p}_2\) \\
    \bottomrule
    \end{tabular}
    \caption{Overview of the two scalar interfaces established below in Corollaries~\ref{cor:normalized_trace_resolvent_convergence} and~\ref{cor:bounded_bilinear_resolvent_convergence}, both obtained from the Frobenius-dual resolvent convergence in Corollary~\ref{cor:frobenius_dual_resolvent_convergence}. The normalized-trace interface is used in the variance analysis, whereas the quadratic specialization \(u_p = v_p\) of the second interface is used in the bias analysis.}
    \label{tab:resolvent_interfaces}
\end{table}

\begin{corollary}[Normalized-trace resolvent convergence]
    \label{cor:normalized_trace_resolvent_convergence}
    Fix \(\lambda > \epsilon > 0\), and let
    \(D_p \in \bbR^{p\times p}\) be deterministic matrices satisfying
    \begin{align}
        M_D
        &:=
        \sup_p\norm{D_p}_\infty
        < \infty \,.
    \end{align}
    Under Assumption~\ref{assum:0.3-0.7},
    \begin{align}
        \frac{1}{p}
        \Tr\left[
            D_p
            \bigl\{
                R_p(-\lambda)-\widetilde{R}_p(-\lambda)
            \bigr\}
        \right]
        \longrightarrow
        0
        \qquad
        \text{a.s.}
        \label{eq:normalized_trace_resolvent_convergence}
    \end{align}
\end{corollary}

\begin{proof}
    Apply Corollary~\ref{cor:frobenius_dual_resolvent_convergence} with \(C_p:=D_p/p\).
    Then, we have
    \begin{align}
        \norm{C_p}_F
        =
        \frac{\norm{D_p}_F}{p}
        \leq
        \frac{\sqrt p\,\norm{D_p}_\infty}{p}
        =
        \frac{\norm{D_p}_\infty}{\sqrt p}
        \leq
        \frac{M_D}{\sqrt p}
        \leq
        M_D \,.
        \label{eq:normalized_trace_functional_lipschitz}
    \end{align}
    Therefore, Eq.~\eqref{eq:frobenius_bounded_test_matrices} holds with \(M_F \leq M_D\), and Eq.~\eqref{eq:frobenius_dual_resolvent_convergence} yields Eq.~\eqref{eq:normalized_trace_resolvent_convergence}.
\end{proof}

\begin{corollary}[Bounded bilinear-form resolvent convergence]
    \label{cor:bounded_bilinear_resolvent_convergence}
    Fix \(\lambda > \epsilon > 0\), and let
    \(u_p,v_p \in \bbR^p\) be deterministic vectors satisfying
    \begin{align}
        M_{uv}
        &:=
        \sup_p
        \norm{u_p}_2\norm{v_p}_2
        < \infty \,.
    \end{align}
    Under Assumption~\ref{assum:0.3-0.7},
    \begin{align}
        u_p\T
        \bigl\{
            R_p(-\lambda)-\widetilde{R}_p(-\lambda)
        \bigr\}
        v_p
        \longrightarrow
        0
        \qquad
        \text{a.s.}
        \label{eq:bounded_bilinear_form_resolvent_convergence}
    \end{align}
\end{corollary}

\begin{proof}
    Apply
    Corollary~\ref{cor:frobenius_dual_resolvent_convergence}
    with the rank-one test matrix
    \(C_p:=v_pu_p\T\).
    Then
    \begin{align}
        \norm{C_p}_F
        &=
        \norm{v_pu_p\T}_F
        =
        \norm{u_p}_2\norm{v_p}_2
        \leq
        M_{uv} \,,
    \end{align}
    so Eq.~\eqref{eq:frobenius_bounded_test_matrices} holds with \(M_F \leq M_{uv}\), and \(\Tr[C_p M] = u_p\T Mv_p\) turns Eq.~\eqref{eq:frobenius_dual_resolvent_convergence} into Eq.~\eqref{eq:bounded_bilinear_form_resolvent_convergence}.
\end{proof}

\begin{remark}[Different Frobenius scales of the two interfaces]
    The two interfaces operate at different Frobenius scales.
    For the normalized-trace interface, \(C_p = D_p/p\) satisfies $\norm{C_p}_F \leq \norm{D_p}_\infty/\sqrt{p} = \calO(p^{-1/2})$, whereas for the bounded bilinear-form interface, \(C_p = v_p u_p\T\) satisfies $\norm{C_p}_F = \norm{u_p}_2 \norm{v_p}_2 = \calO(1)$.
    Thus, the canonical scales of the two test classes differ by up to a factor of \(\sqrt p\). In particular, representing the bilinear form \(u_p\T M v_p = \Tr[C_p M]\) through the normalized-trace interface \(\frac{1}{p}\Tr[D_p M]\) would require \(D_p = p\, v_p u_p\T\) and \(\norm{D_p}_\infty = p \norm{u_p}_2 \norm{v_p}_2\), whose operator norm is not guaranteed to remain uniformly bounded in \(p\) and can grow linearly with \(p\), so the uniformly bounded operator-norm assumption does not necessarily hold.
    This is why the bounded quadratic-form convergence needed for the bias does not follow from the normalized-trace interface alone.
\end{remark}

\subsection{Identically Distributed Feature Vectors and the Effective Regularization Parameter}
\begin{remark}[Specialization to identically distributed feature vectors]
    When $\Sigma_i = \Sigma$ for all $i \in [N_{\mathrm{tr}}]$, the fixed-point equation~\eqref{eq:self_consistent_eq_louart} reduces to the following form by replacing $z$ with $-\lambda$ and $\tilde{\Lambda}_i^z$ with $-\kappa_\lambda$:
    \begin{align}
        \kappa_\lambda
        = \lambda + \frac{1}{N_{\mathrm{tr}}}
        \Tr[
            \Sigma\left(I_p + \frac{\Sigma}{\kappa_\lambda}\right)^{-1}
        ] \,.
        \label{eq:fixed_point_equation}
    \end{align}
    In this case, the deterministic equivalent $\widetilde{R}(-\lambda)$ simplifies to
    \begin{align}
        \widetilde{R}(-\lambda)=\frac{\kappa_\lambda}{\lambda}(\Sigma+\kappa_\lambda I_p)^{-1} \,.
    \end{align}
    Applying
    Corollary~\ref{cor:frobenius_dual_resolvent_convergence}
    gives
    \begin{align}
        \lambda(\hat\Sigma+\lambda I_p)^{-1}
        \asymp_F
        \kappa_\lambda(\Sigma+\kappa_\lambda I_p)^{-1}.
    \end{align}
    This is the specialization stated in Proposition~\ref{prop:resolvent_deterministic_equivalent}.
    Note that we do not need Assumption~0.7 in Ref.~\cite{louart2021spectral} to establish this specialization, as explained in Remark~\ref{remark:assumption_0.7}.
\end{remark}

\begin{lemma}[Uniqueness of the Positive Solution to the Self-Consistent Equation]
    \label{lemma:unique_positive_fixed_point}
    Let $A \in \bbR^{p\times p}$ be positive definite, and let $a > 0$. Then the equation
    \begin{align}
        \frac{1}{N_{\mathrm{tr}}}\Tr\left[A(A+\kappa I_p)^{-1}\right]
        +\frac{a}{\kappa}
        =1
        \label{eq:general_positive_fixed_point}
    \end{align}
    has exactly one positive solution $\kappa > 0$.
    Moreover, this solution satisfies $\kappa > a$.
\end{lemma}
\begin{proof}
    Define
    \begin{align}
        f(\kappa)
        &:=
        \frac{1}{N_{\mathrm{tr}}}\Tr\left[A(A+\kappa I_p)^{-1}\right]
        +\frac{a}{\kappa},
        \qquad \kappa>0.
    \end{align}
    If $\alpha_1,\ldots,\alpha_p\geq0$ are the eigenvalues of $A$, then
    \begin{align}
        f'(\kappa)
        =
        -\frac{1}{N_{\mathrm{tr}}}\sum_{i=1}^p
        \frac{\alpha_i}{(\alpha_i+\kappa)^2}
        -\frac{a}{\kappa^2}
        <0.
    \end{align}
    Thus, $f$ is strictly decreasing. Furthermore,
    \begin{align}
        \lim_{\kappa\downarrow0}f(\kappa) = +\infty,
        \qquad
        \lim_{\kappa \to \infty}f(\kappa) = 0.
    \end{align}
    Hence, by continuity, $f(\kappa) = 1$ has exactly one positive solution.

    Finally, rearranging the fixed-point equation gives
    \begin{align}
        \kappa
        =
        a + \frac{\kappa}{N_{\mathrm{tr}}}
        \Tr\left[A(A+\kappa I_p)^{-1}\right]
        > a.
    \end{align}
\end{proof}

\begin{remark}\label{remark:lipschitz_constant}
    In Assumption~\ref{assum:feature_generation}, we consider a Lipschitz constant $C_1$ for the feature-generating map $\Omega$.
    In Theorem~\ref{thm:0.9}, the exponent $c$ in the concentration inequality~\eqref{eq:concentration_resolvent} depends on the Lipschitz constant $C_1$ as $c \propto 1/C_1^2$.
    In the quantum setting, the Lipschitz constant $C_1$ scales as $\sqrt{dL}\nu_{\max}$, as shown in Appendix~\ref{appdx:sec:lipschitz_qml_feature_map}, if we assume $\operatorname{Lip}(\chi)$ is constant.
    Since \(n = \order{\log p} = \order{\log{N_{\mathrm{tr}}}}\) by Assumption~\ref{assum:exponential_effective_dimension} and we assume $d,L \propto \poly(n)$, $C_1 \propto \poly(\log{N_{\mathrm{tr}}})$.
    While this choice results in a Lipschitz constant that grows with $N_{\mathrm{tr}}$, the subsequent results (Corollary~\ref{cor:normalized_trace_resolvent_convergence} and Corollary~\ref{cor:bounded_bilinear_resolvent_convergence}) still hold.
    This is because the concentration inequality \ref{eq:concentration_resolvent} in Theorem~\ref{thm:0.9} is simply adjusted by replacing the exponent $N_{\mathrm{tr}}$ with $\order{N_{\mathrm{tr}}/C_1^2} = \order{N_{\mathrm{tr}}/\poly(\log{N_{\mathrm{tr}}})}$ and the Borel--Cantelli lemma still applies. Thus, the growth of $C_1$ only slows the concentration rate while preserving almost-sure convergence.
\end{remark}
\newpage
\section{Derivative Trick for the Deterministic Equivalent of the Test Risk}
\label{appdx:sec:derivative_method}

In this section, we derive the deterministic equivalents of the bias and variance terms in Theorem~\ref{thm:test_risk_DE}. These terms are given by
\begin{align}
    \calB_{\lambda,\hat{\Sigma}}
    &=
    \lambda^2\bth_*\T
    (\hat{\Sigma}+\lambda I_p)^{-1}
    \Sigma
    (\hat{\Sigma}+\lambda I_p)^{-1}
    \bth_* \,,
    \label{eq:derivative_trick_bias_definition}
    \\
    \calV_{\lambda,\hat{\Sigma}}
    &=
    \frac{\sigma^2}{N_{\mathrm{tr}}}
    \Tr\left[
        \Sigma\hat{\Sigma}
        (\hat{\Sigma}+\lambda I_p)^{-2}
    \right] \,.
    \label{eq:derivative_trick_variance_definition}
\end{align}
Both expressions contain second-order resolvents. The derivative trick generates them by differentiating a first-order resolvent with respect to an auxiliary perturbation parameter or the regularization parameter~\cite{zavatone-vethlecture,atanasov2024scaling}.
The principal technical issue is that pointwise convergence of analytic functions does not, in general, imply convergence of their derivatives.
To rigorously justify the convergence of the derivatives, our proof proceeds in the following four steps.
\begin{enumerate}[leftmargin=20pt, label=\textbf{Step \arabic*:}]
    \item \textbf{Reduction to derivatives.}
    Both $\calB_{\lambda,\hat{\Sigma}}$ and $\calV_{\lambda,\hat{\Sigma}}$ are second-order in the resolvent, whereas the only available result is the first-order statement of Proposition~\ref{prop:resolvent_deterministic_equivalent}. We therefore represent each of them as the derivative of a first-order resolvent observable, with respect to an auxiliary perturbation parameter $J$ for the bias (Eq.~\eqref{eq:bias_as_derivative}) and with respect to the regularization parameter for the variance (Eq.~\eqref{eq:variance_resolvent_derivative_identity}).
    \item \textbf{General analytic tools (Appendix~\ref{appdx:subsec:reusable_tools}).}
    Passing from convergence of a function to convergence of its derivative requires Vitali's convergence theorem (Lemma~\ref{lemma:vitali}), which needs local uniform boundedness together with pointwise convergence on a set with a limit point. Since our convergence statements hold almost surely, we record the pathwise form of this argument in Proposition~\ref{prop:pathwise_analytic_derivative_transfer}. We also record a quantitative implicit-branch lemma based on Rouch\'e's theorem (Lemma~\ref{lemma:uniform_analytic_implicit_branch}).
    \item \textbf{Uniform analytic extensions (Appendices~\ref{appdx:subsec:vitali_justification}--\ref{appdx:subsec:analyticity_kappa}).}
    Vitali's theorem requires the limit candidate to be analytic as well, so the effective regularization parameter must be continued holomorphically in $J$ and in the complex regularization parameter. Moreover, the domain of this continuation must be \emph{uniform} in $p$. Lemma~\ref{lemma:uniform_analytic_implicit_branch} supplies such branches; the source of the $p$-uniformity is the stability bound $1 - \eta_{\lambda,p} \geq a_\lambda$ of Proposition~\ref{prop:effective_regularization_stability}. Table~\ref{tab:uniform_implicit_branch_correspondence} lists the correspondence between the lemma and its two instantiations.
    \item \textbf{Derivative calculation and convergence (Appendices~\ref{appdx:subsec:bias_derivation}--\ref{appdx:subsec:variance_derivation}).}
    On these domains we establish almost-sure scalar convergence on the real axis, using the quadratic specialization of the bounded-bilinear-form interface (Corollary~\ref{cor:bounded_bilinear_resolvent_convergence}) for the bias and the normalized-trace interface (Corollary~\ref{cor:normalized_trace_resolvent_convergence}) for the variance, apply Vitali's theorem, and finally evaluate the derivatives at $J = 0$ and at $z = \lambda$.
\end{enumerate}

Throughout this section, \(\lambda > 0\) is fixed.
We repeatedly apply the matrix inverse derivative identity
\begin{align}
    \frac{\ind}{\ind t}A(t)^{-1}
    =
    -A(t)^{-1}A'(t)A(t)^{-1} \,.
    \label{eq:matrix_inverse_derivative_identity}
\end{align}
To display the dependence on the effective feature dimension, we temporarily write
\begin{align}
    N_p := N_{\mathrm{tr}}(p),
    \qquad
    \hat{\Sigma}_p,
    \qquad
    \Sigma_p,
    \qquad
    \bth_{*,p} \,.
\end{align}
The subscript \(p\) is suppressed again in the final formulas.

\subsection{Reusable Analytic Tools}
\label{appdx:subsec:reusable_tools}
We first isolate the complex-analytic arguments that will be instantiated in both derivative calculations.
We use the following standard form of Vitali's theorem; see, for example, Ref.~\cite{bai2010spectral}.
\begin{lemma}[Vitali's convergence theorem]\label{lemma:vitali}
    Let \(D\subset\bbC\) be a connected open set, and let \(\{f_p\}_p\) be a sequence of analytic functions on \(D\).
    Assume that the sequence is locally uniformly bounded: for every compact set \(K\subset D\),
    \begin{align}
        \sup_p\sup_{z \in K}|f_p(z)|
        < \infty \,.
        \label{eq:vitali_local_uniform_bound}
    \end{align}
    Suppose that \(f_p(z)\) converges for every \(z\) in a subset \(E\subset D\) having a limit point in \(D\). Then there exists an analytic function \(f\) on \(D\) such that, for every compact set \(K\subset D\),
    \begin{align}
        \sup_{z \in K}
        |f_p(z)-f(z)|
        \longrightarrow
        0 \,,
        \qquad
        \sup_{z \in K}
        |f_p'(z)-f'(z)|
        \longrightarrow
        0 \,.
        \label{eq:vitali_local_uniform_convergence}
    \end{align}
    % Consequently, the derivative sequence is uniformly bounded on every compact
    % subset \(K \subset D\):
    % \begin{align}
    %     \sup_p\sup_{z \in K}|f_p'(z)|
    %     < \infty \,.
    % \end{align}
\end{lemma}

The following proposition records the pathwise form used in this section.
It is not an additional complex-analytic theorem: it is obtained by first constructing one probability-one event on which the pointwise convergence holds simultaneously on a countable set, and then applying Lemma~\ref{lemma:vitali} and the identity theorem on each sample path.

\begin{proposition}[Almost-sure analytic derivative transfer]
\label{prop:pathwise_analytic_derivative_transfer}
    Let \((\Omega,\calF,\bbP)\) be the underlying probability space.
    Let \(D\subset\bbC\) be a connected open set, and let \(E\subset D\) be a countable set having a limit point in \(D\).
    Let \(h_p\colon\Omega\times D\to\bbC\) be random functions, and write \(h_p(z):=h_p(\cdot,z)\) for the corresponding random variable.
    Suppose that for every \(\omega \in \Omega\),
    \begin{enumerate}[label=(\roman*),leftmargin=25pt]
        \item \(h_p(\omega,\cdot)\) is analytic on \(D\) for every \(p\);
        \item for every compact set \(K\subset D\),
        \begin{align}
            \sup_p\sup_{z \in K}
            |h_p(\omega,z)|
            < \infty \,.
            \label{eq:pathwise_analytic_local_uniform_bound}
        \end{align}
    \end{enumerate}
    Assume further that, for every fixed \(z \in E\),
    \begin{align}
        h_p(z)
        \longrightarrow
        0
        \quad
        \text{a.s. as }p \to \infty \,.
        \quad
        \Longleftrightarrow
        \quad
        \bbP\left(
            \set*{\omega \in \Omega}[
                h_p(\omega,z) \longrightarrow 0
                \text{ as }p \to \infty
            ]
        \right)
        =
        1 \,.
        \label{eq:analytic_transfer_pointwise_assumption}
    \end{align}
    Then there exists an event \(\Omega_0\) with \(\bbP(\Omega_0) = 1\) such that, for every \(\omega \in \Omega_0\) and every compact set \(K\subset D\),
    \begin{align}
        \sup_{z \in K}
        |h_p(\omega,z)|
        &\longrightarrow
        0 \,,
        \label{eq:analytic_transfer_local_uniform_conclusion}
        \\
        \sup_{z \in K}
        |h_p'(\omega,z)|
        &\longrightarrow
        0 \,.
        \label{eq:analytic_transfer_derivative_local_uniform_conclusion}
    \end{align}
    In particular, for every fixed \(z_0 \in D\),
    \begin{align}
        h_p'(z_0)
        \longrightarrow
        0
        \qquad
        \text{a.s. as }p \to \infty \,.
        \label{eq:analytic_transfer_derivative_conclusion}
    \end{align}
\end{proposition}

\begin{proof}
    For every fixed \(z \in E\), define
    \begin{align}
        \Omega_z
        &:=
        \set{\omega \in \Omega}[
            h_p(\omega,z) \longrightarrow 0
            \text{ as }p \to \infty
        ] \,.
        \label{eq:analytic_transfer_pointwise_event}
    \end{align}
    Equation~\eqref{eq:analytic_transfer_pointwise_assumption} means that \(\bbP(\Omega_z) = 1\) for every \(z \in E\). Since \(E\) is countable,
    define
    \begin{align}
        \Omega_0
        &:=
        \bigcap_{z\in E}\Omega_z \,.
        \label{eq:analytic_transfer_common_event}
    \end{align}
    This event has probability one. Explicitly, De Morgan's law and countable subadditivity give
    \begin{align}
        \bbP(\Omega_0^c)
        =
        \bbP\left(
            \bigcup_{z \in E}\Omega_z^c
        \right)
        \leq
        \sum_{z \in E}\bbP(\Omega_z^c)
        =
        0 \,.
        \label{eq:analytic_transfer_common_event_probability}
    \end{align}

    Fix \(\omega \in \Omega_0\). By construction,
    \begin{align}
        h_p(\omega,z)
        \longrightarrow
        0
        \qquad
        \text{for every }z \in E \,.
        \label{eq:analytic_transfer_pathwise_pointwise_convergence}
    \end{align}
    Moreover, the functions \(h_p(\omega,\cdot)\) are analytic on \(D\) and locally uniformly bounded there by Eq.~\eqref{eq:pathwise_analytic_local_uniform_bound}.
    Thus, Lemma~\ref{lemma:vitali} applies to the deterministic analytic sequence \(\{h_p(\omega,\cdot)\}_p\). It gives an analytic function \(h_\omega\) on \(D\) such that
    \begin{align}
        h_p(\omega,\cdot)
        &\longrightarrow
        h_\omega,
        &
        h_p'(\omega,\cdot)
        &\longrightarrow
        h_\omega'
    \end{align}
    locally uniformly on \(D\).

    For every \(z \in E\), Eq.~\eqref{eq:analytic_transfer_pathwise_pointwise_convergence} implies \(h_\omega(z) = 0\). Since \(E\) has a limit point in the connected open set \(D\), the identity theorem yields
    \begin{align}
        h_\omega
        \equiv
        0
        \qquad
        \text{on }D \,.
    \end{align}
    Consequently, \(h_\omega'\equiv0\), and the locally uniform convergence provided by Vitali's theorem becomes precisely Eqs.~\eqref{eq:analytic_transfer_local_uniform_conclusion} and \eqref{eq:analytic_transfer_derivative_local_uniform_conclusion}.
    Because this reasoning holds for every \(\omega\in\Omega_0\) and \(\bbP(\Omega_0)=1\), the conclusions hold almost surely. Taking \(K = \{z_0\}\) gives Eq.~\eqref{eq:analytic_transfer_derivative_conclusion}.
\end{proof}

\begin{lemma}[Rouch\'e's theorem]\label{lemma:rouche}
    Fix \(\zeta_0 \in \bbC\) and \(R > 0\), and define
    \begin{align}
        \bbD
        &:=
        \set{\zeta \in \bbC}[|\zeta - \zeta_0| < R] \,,
        \\
        \overline{\bbD}
        &:=
        \set{\zeta \in \bbC}[|\zeta - \zeta_0| \leq R] \,,
        \\
        \partial\bbD
        &:=
        \set{\zeta \in \bbC}[|\zeta - \zeta_0| = R] \,.
    \end{align}
    Let \(\mathfrak{u}\) and \(\mathfrak{v}\) be analytic on an open neighborhood of \(\overline{\bbD}\). If
    \begin{align}
        |\mathfrak{v}(\zeta)|
        <
        |\mathfrak{u}(\zeta)|
        \qquad
        \text{for every }\zeta \in \partial\bbD \,,
    \end{align}
    then \(\mathfrak{u}\) and \(\mathfrak{u} + \mathfrak{v}\) have the same number of zeros in \(\bbD\), counted with multiplicity.
\end{lemma}

\begin{lemma}[Uniform analytic implicit branch]
\label{lemma:uniform_analytic_implicit_branch}
    Let \(\kappa_p \in \bbC\), and suppose that \(H_p(t,\kappa)\) is jointly analytic on an open neighborhood of the closed polydisc in $\bbC^2$ defined by
    \begin{align}
        |t| \leq \tau_0,
        \qquad
        |\kappa - \kappa_p| \leq \rho_0 \,,
    \end{align}
    where \(\tau_0,\rho_0 > 0\) are independent of \(p\). Assume that
    \begin{align}
        H_p(0,\kappa_p)
        &=0,
        &
        \inf_p
        |\partial_\kappa H_p(0,\kappa_p)|
        &\geq a>0,
        \label{eq:implicit_branch_base_conditions}
    \end{align}
    and that, throughout the indicated polydisc,
    \begin{align}
        |\partial_\kappa^2 H_p(0,\kappa)|
        &\leq B,
        &
        |\partial_t H_p(t,\kappa)|
        &\leq L
        \label{eq:implicit_branch_derivative_bounds}
    \end{align}
    for constants \(B,L < \infty\) independent of \(p\).
    Let \(\rho \in (0,\rho_0]\) and \(\tau \in (0,\tau_0]\) satisfy
    \begin{align}
        \frac{B\rho}{2}
        \leq
        \frac{a}{4},
        \qquad
        L\tau
        \leq
        \frac{a\rho}{4} \,.
        \label{eq:uniform_implicit_branch_radius_choice}
    \end{align}
    Then, for every \(|t|<\tau\), the equation \(H_p(t,\kappa)=0\) has exactly one zero \(\kappa_p(t)\) in
    \begin{align}
        |\kappa - \kappa_p| < \rho \,.
    \end{align}
    This zero is simple, \(t \mapsto \kappa_p(t)\) is analytic on \(\{|t| < \tau\}\), and
    \begin{align}
        \kappa_p'(0)
        =
        -
        \frac{
            \partial_t H_p(0,\kappa_p)
        }{
            \partial_\kappa H_p(0,\kappa_p)
        } \,.
        \label{eq:uniform_implicit_branch_derivative}
    \end{align}
\end{lemma}

\begin{proof}
    On \(|\kappa - \kappa_p| = \rho\), compare \(H_p(0,\kappa)\) with
    \begin{align}
        \ell_p(\kappa)
        &:=
        \partial_\kappa H_p(0,\kappa_p)
        (\kappa - \kappa_p) \,.
    \end{align}
    From Eq.~\eqref{eq:implicit_branch_base_conditions}, \(|\ell_p(\kappa)| \geq a \rho\) on the boundary.
    Taylor's formula and Eq.~\eqref{eq:implicit_branch_derivative_bounds} give
    \begin{align}
        |H_p(0,\kappa)-\ell_p(\kappa)|
        \leq
        \frac{B\rho^2}{2}
        \leq
        \frac{a\rho}{4}
        <
        |\ell_p(\kappa)| \,.
    \end{align}
    Rouch\'e's theorem shows that \(H_p(0,\cdot)\) has exactly one zero in the disk, counted with multiplicity. Moreover, on the boundary,
    \begin{align}
        |H_p(0,\kappa)|
        \geq
        |\ell_p(\kappa)| - \frac{a\rho}{4}
        \geq
        \frac{3a\rho}{4} \,.
    \end{align}
    For \(|t|<\tau\), Eq.~\eqref{eq:implicit_branch_derivative_bounds} gives
    \begin{align}
        |H_p(t,\kappa) - H_p(0,\kappa)|
        \leq
        L|t|
        <
        \frac{a\rho}{4}
        <
        |H_p(0,\kappa)| \,.
    \end{align}
    A second application of Rouch\'e's theorem therefore shows that \(H_p(t,\cdot)\) also has exactly one zero in the disk, counted with multiplicity. Since its total multiplicity is one, this zero is simple.
    Thus, denoting it by \(\kappa_p(t)\),
    \begin{align*}
        \partial_\kappa H_p(t,\kappa_p(t))
        \neq
        0 \,.
    \end{align*}

    Fix \(t_0\) with \(|t_0| < \tau\), and set \(\kappa_0 := \kappa_p(t_0)\). The complex-analytic implicit-function theorem states that if a jointly analytic function \(H(t,\kappa)\) satisfies
    \begin{align*}
        H(t_0,\kappa_0)=0,
        \qquad
        \partial_\kappa H(t_0,\kappa_0)\neq0,
    \end{align*}
    then, on some neighborhood \(U\) of \(t_0\), the equation \(H(t,\kappa) = 0\) determines \(\kappa\) uniquely near \(\kappa_0\) as an analytic function \(\varphi(t)\) satisfying \(\varphi(t_0) = \kappa_0\).
    Applying the theorem to \(H_p\), and shrinking \(U\) if necessary, we may ensure that \(U \subset \{t \in \bbC : |t| < \tau\}\) and \(|\varphi(t) - \kappa_p| < \rho\) for every \(t\in U\). The second Rouch\'e argument shows that, for each such \(t\), there is only one zero in this disk.
    Hence \(\varphi(t) = \kappa_p(t)\) on \(U\). Since \(t_0\) was arbitrary, \(t \mapsto \kappa_p(t)\) is analytic throughout \(\{|t| < \tau\}\); equivalently, the local implicit branches agree on overlaps because they all select the same zero in the common disk.

    Finally, \(\kappa_p(0) = \kappa_p\), and differentiating \(H_p(t,\kappa_p(t)) = 0\) at \(t = 0\) gives
    \begin{align*}
        0
        =
        \partial_t H_p(0,\kappa_p)
        +
        \partial_\kappa H_p(0,\kappa_p)\kappa_p'(0) \,.
    \end{align*}
    The coefficient of \(\kappa_p'(0)\) is nonzero by Eq.~\eqref{eq:implicit_branch_base_conditions}; solving for \(\kappa_p'(0)\) gives Eq.~\eqref{eq:uniform_implicit_branch_derivative}.
\end{proof}

\subsection{Scalar Analytic Framework}
\label{appdx:subsec:vitali_justification}
We begin by recording bounds that hold uniformly along the high-dimensional sequence. Since \(p/N_p \to \gamma \in (0,\infty)\), there exists a constant \(\Gamma < \infty\) such that
\begin{align}
    \sup_p\frac{p}{N_p} \leq \Gamma \,.
    \label{eq:derivative_trick_aspect_ratio_bound}
\end{align}
We also use
\begin{align}
    M_\Sigma
    &:=
    \sup_p \norm{\Sigma_p}_\infty < \infty,
    \qquad
    C_*
    :=
    \sup_p \norm{\bth_{*,p}}_2 < \infty \,.
    \label{eq:derivative_trick_uniform_bounds}
\end{align}
The second bound is Assumption~\ref{assum:target_regularity}. For the first one, Assumption~\ref{assum:feature_generation} and the Gaussian Poincar\'e inequality give, for every deterministic unit vector
\(v\in\bbR^p\),
\begin{align}
    v\T\Sigma_pv
    &=
    \Var(v\T\bx)+\bigl\{\E[v\T\bx]\bigr\}^2
    \leq C_1^2+C_2^2 \,.
\end{align}
Thus, one may take \(M_\Sigma=C_1^2+C_2^2\) when the constants in Assumption~\ref{assum:feature_generation} are uniform in \(p\).
For the normalized quantum feature vectors considered in this work, the sharper bound \(\norm{\Sigma_p}_\infty\leq\E\norm{\bx}_2^2\leq1\) holds.

Let \(\kappa_{\lambda,p}\) denote the positive solution of the self-consistent equation
\begin{align}
    \frac{1}{N_p}
    \Tr\left[
        \Sigma_p(\Sigma_p+\kappa_{\lambda,p}I_p)^{-1}
    \right]
    +
    \frac{\lambda}{\kappa_{\lambda,p}}
    =1 \,,
    \label{eq:derivative_trick_base_fixed_point}
\end{align}
and define
\begin{align}
    \eta_{\lambda,p}
    &:=
    \frac{1}{N_p}
    \Tr\left[
        \Sigma_p^2
        (\Sigma_p+\kappa_{\lambda,p}I_p)^{-2}
    \right] \,.
    \label{eq:derivative_trick_eta}
\end{align}
\begin{proposition}[Effective-regularization stability]
    \label{prop:effective_regularization_stability}
    For every fixed \(\lambda > 0\), the quantities in Eqs.~\eqref{eq:derivative_trick_base_fixed_point} and \eqref{eq:derivative_trick_eta} satisfy the uniform bounds
    \begin{align}
        \lambda
        \leq
        \kappa_{\lambda,p}
        \leq
        K_\lambda
        &:=
        \lambda+\Gamma M_\Sigma \,,
        \label{eq:derivative_trick_kappa_bounds}
        \\
        1-\eta_{\lambda,p}
        \geq
        a_\lambda
        &:=
        \frac{\lambda}{K_\lambda}
        >0 \,.
        \label{eq:uniform_stability_eta}
    \end{align}
\end{proposition}

\begin{proof}
    Multiplying Eq.~\eqref{eq:derivative_trick_base_fixed_point} by \(\kappa_{\lambda,p}\) and rearranging gives
    \begin{align}
        \kappa_{\lambda,p}
        =
        \lambda
        +
        \frac{\kappa_{\lambda,p}}{N_p}
        \Tr\left[
            \Sigma_p(\Sigma_p+\kappa_{\lambda,p}I_p)^{-1}
        \right] \,.
    \end{align}
    Since \(\Sigma_p\) is positive definite and $\kappa_{\lambda,p} > 0$ by Lemma~\ref{lemma:unique_positive_fixed_point}
    \begin{align}
        \lambda
        <
        \kappa_{\lambda,p}
        \leq
        \lambda+\frac{1}{N_p}\Tr\Sigma_p
        \leq
        \lambda+\Gamma M_\Sigma
        =
        K_\lambda \,,
    \end{align}
    which proves Eq.~\eqref{eq:derivative_trick_kappa_bounds}.

    Each eigenvalue of \(\Sigma_p(\Sigma_p+\kappa_{\lambda,p}I_p)^{-1}\) belongs to \([0,1]\).
    Using \(x^2\leq x\) on \([0,1]\), we obtain
    \begin{align}
        \eta_{\lambda,p}
        &\leq
        \frac{1}{N_p}
        \Tr\left[
            \Sigma_p(\Sigma_p+\kappa_{\lambda,p}I_p)^{-1}
        \right]
        =
        1-\frac{\lambda}{\kappa_{\lambda,p}} \,.
        \label{eq:eta_upper_bound_from_fixed_point}
    \end{align}
    Combining Eq.~\eqref{eq:eta_upper_bound_from_fixed_point} with the upper bound on \(\kappa_{\lambda,p}\) in Eq.~\eqref{eq:derivative_trick_kappa_bounds}, we obtain
    \begin{align}
        1-\eta_{\lambda,p}
        \geq
        \frac{\lambda}{\kappa_{\lambda,p}}
        \geq
        \frac{\lambda}{K_\lambda}
        =:a_\lambda
        >0 \,.
    \end{align}
\end{proof}

Proposition~\ref{prop:effective_regularization_stability} is the common stability condition for both analytic branches below. In particular, it provides their uniform non-degeneracy; the derivative formulas are obtained only after the branches have been constructed.

\subsection{Uniform Holomorphic Continuation of the Effective Regularization}
\label{appdx:subsec:analyticity_kappa}
We next construct the holomorphic branches of the effective regularization parameter that enter the bias and variance calculations. The essential point is that the domains of these branches are uniform in \(p\).

For orientation, Table~\ref{tab:uniform_implicit_branch_correspondence} gives the complete correspondence between Lemma~\ref{lemma:uniform_analytic_implicit_branch} and its two applications.  Each analytic and quantitative assumption is verified in the corresponding subsubsection below.
\begin{table}[H]
    \centering
    \setlength{\tabcolsep}{15pt}
    \begin{tabular}{lll}
    \toprule
    Lemma symbol & Bias instance & Variance instance \\
    \midrule
    \(H_p(t,\kappa)\)
        & \(H_p(J,\kappa)\)
        & \(\widehat H_p(t,\kappa)
        =H_p^{\mathrm{var}}(\lambda+t,\kappa)\) \\
    \(t\) & \(J\) & \(z-\lambda\) \\
    \(\kappa_p\) & \(\kappa_{\lambda,p}\) & \(\kappa_{\lambda,p}\) \\
    \(\tau_0\) & \(r_0\) & \(\lambda/2\) \\
    \(\rho_0\) & \(\rho_0\) & \(\rho_0\) \\
    \(a\) & \(a_\lambda\) & \(a_\lambda\) \\
    \(B\) & \(B_\lambda\) & \(B_\lambda\) \\
    \(L\) & \(L_\lambda\) & \(1\) \\
    \(\tau\) & \(r\) & \(\delta\) \\
    \(\rho\) & \(\rho\) & \(\rho\) \\
    \(\kappa_p(t)\)
        & \(\kappa_{\lambda,J,p}\)
        & \(\kappa_{\lambda+t,p}=\kappa_{z,p}\) \\
    \bottomrule
    \end{tabular}
    \caption{Instantiation of the uniform analytic implicit-branch lemma.}
    \label{tab:uniform_implicit_branch_correspondence}
\end{table}

\subsubsection{Perturbation parameter for the bias}
For \((J,\kappa) \in \bbC^2\), define
\begin{align}
    H_p(J,\kappa)
    &:=
    \kappa-\lambda
    -
    \frac{\kappa}{N_p}
    \Tr\left[
        \Sigma_p
        \bigl\{\Sigma_p+\kappa(I_p + J\Sigma_p)\bigr\}^{-1}
    \right] \,.
    \label{eq:complex_fixed_point_bias}
\end{align}
At \(J = 0\), the equation \(H_p(0,\kappa)=0\) is equivalent to
Eq.~\eqref{eq:derivative_trick_base_fixed_point}; hence
\begin{align}
    H_p(0,\kappa_{\lambda,p})=0 \,.
\end{align}
To evaluate the derivative with respect to \(\kappa\), set
\begin{align}
    m_p(\kappa)
    &:=
    \frac{1}{N_p}
    \Tr\left[
        \Sigma_p(\Sigma_p+\kappa I_p)^{-1}
    \right] \,.
\end{align}
Since
\(H_p(0,\kappa)=\kappa-\lambda-\kappa m_p(\kappa)\), direct
differentiation gives
\begin{align}
    \partial_\kappa H_p(0,\kappa_{\lambda,p})
    &=
    1-m_p(\kappa_{\lambda,p})
    +
    \frac{\kappa_{\lambda,p}}{N_p}
    \Tr\left[
        \Sigma_p
        (\Sigma_p+\kappa_{\lambda,p}I_p)^{-2}
    \right] \,.
\end{align}
The identity
\begin{align}
    m_p(\kappa_{\lambda,p})-\eta_{\lambda,p}
    =
    \frac{\kappa_{\lambda,p}}{N_p}
    \Tr\left[
        \Sigma_p
        (\Sigma_p+\kappa_{\lambda,p}I_p)^{-2}
    \right]
\end{align}
therefore yields
\begin{align}
    \partial_\kappa H_p(0,\kappa_{\lambda,p})
    =
    1-\eta_{\lambda,p}
    \geq
    a_\lambda \,.
    \label{eq:complex_fixed_point_kappa_derivative}
\end{align}

We now verify the uniform analyticity and derivative bounds required by
Lemma~\ref{lemma:uniform_analytic_implicit_branch}. First choose preliminary
radii \(\rho_0,r_0>0\), independently of \(p\), such that
\begin{align}
    \rho_0
    +
    r_0M_\Sigma(K_\lambda+\rho_0)
    \leq
    \frac{\lambda}{2} \,.
    \label{eq:complex_neighborhood_initial_choice}
\end{align}
In particular, since the second term on the left-hand side of Eq.~\eqref{eq:complex_neighborhood_initial_choice} is nonnegative,
\begin{align}
    0<\rho_0\leq\frac{\lambda}{2} \,.
    \label{eq:rho0_upper_bound}
\end{align}
If \(|\kappa - \kappa_{\lambda,p}| \leq \rho_0\),
\(|J| \leq r_0\), and \(\xi\) is an eigenvalue of \(\Sigma_p\), then
\begin{align}
    |\xi + \kappa(1 + J\xi)|
    &=
    |\xi + \kappa_{\lambda,p} + \kappa - \kappa_{\lambda,p} + \kappa J\xi| \\
    &\geq
    \xi + \kappa_{\lambda,p}
    - |\kappa - \kappa_{\lambda,p}|
    - |\kappa J\xi|
    \\
    &\geq
    \lambda
    -
    \rho_0
    -
    r_0M_\Sigma(K_\lambda+\rho_0)
    \geq
    \frac{\lambda}{2} \,.
    \label{eq:complex_denominator_lower_bound}
\end{align}
Here we used Eqs.~\eqref{eq:derivative_trick_uniform_bounds} and
\eqref{eq:derivative_trick_kappa_bounds}. Diagonalizing \(\Sigma_p\) gives
\begin{align}
    H_p(J,\kappa)
    =
    \kappa-\lambda
    -
    \frac{\kappa}{N_p}
    \sum_{i=1}^p
    \frac{\xi_i}
    {\xi_i+\kappa(1 + J\xi_i)} \,.
\end{align}
Equation~\eqref{eq:complex_denominator_lower_bound} therefore shows that
\(H_p\) is jointly analytic on an open neighborhood of the resulting closed
polydisc.

The relevant derivatives are uniformly bounded there. In particular, there
exist \(B_\lambda,L_\lambda<\infty\), independent of \(p\), such that
\begin{align}
    \sup_{|\kappa - \kappa_{\lambda,p}|\leq\rho_0}
    |\partial_\kappa^2 H_p(0,\kappa)|
    &\leq
    B_\lambda \,,
    \label{eq:H_second_derivative_bound}
    \\
    \sup_{\substack{
        |J| \leq r_0 \\
        |\kappa - \kappa_{\lambda,p}| \leq \rho_0}}
    |\partial_J H_p(J,\kappa)|
    &\leq
    L_\lambda \,.
    \label{eq:H_J_derivative_bound}
\end{align}
Indeed,
\begin{align}
    \partial_\kappa^2H_p(0,\kappa)
    &=
    \frac{2}{N_p}
    \Tr\left[
        \Sigma_p(\Sigma_p+\kappa I_p)^{-2}
    \right]
    -
    \frac{2\kappa}{N_p}
    \Tr\left[
        \Sigma_p(\Sigma_p+\kappa I_p)^{-3}
    \right] \,,
    \\
    \partial_JH_p(J,\kappa)
    &=
    \frac{\kappa^2}{N_p}
    \Tr\left[
        \Sigma_p^2
        \bigl\{\Sigma_p+\kappa(I_p + J\Sigma_p)\bigr\}^{-2}
    \right] \,,
    \label{eq:H_J_derivative}
\end{align}
and the asserted bounds follow from
Eqs.~\eqref{eq:derivative_trick_aspect_ratio_bound},
\eqref{eq:derivative_trick_uniform_bounds}, and
\eqref{eq:complex_denominator_lower_bound}.

Choose final radii \(\rho \in (0,\rho_0]\) and \(r \in (0,r_0]\),
independently of \(p\), such that
\begin{align}
    \frac{B_\lambda\rho}{2}
    \leq
    \frac{a_\lambda}{4},
    \qquad
    L_\lambda r
    \leq
    \frac{a_\lambda\rho}{4},
    \qquad
    rM_\Sigma
    \leq
    \frac{1}{2} \,.
    \label{eq:rouche_radius_choice}
\end{align}
From \eqref{eq:rho0_upper_bound}, we have
\begin{align}
    0<\rho\leq\frac{\lambda}{2} \,.
    \label{eq:rho_upper_bound}
\end{align}
Define the family of \(\kappa\)-disks with common radius \(\rho\), and the
common parameter disk,
\begin{align}
    \bbD_{\kappa,p}
    &:=
    \set{\kappa \in \bbC}[
        |\kappa-\kappa_{\lambda,p}|<\rho
    ] \,,
    \label{eq:kappa_rouche_disk}
    \\
    D_J
    &:=
    \set{J \in \bbC}[|J|<r] \,.
    \label{eq:bias_complex_disk}
\end{align}
The Bias column of
Table~\ref{tab:uniform_implicit_branch_correspondence} is verified as
follows:
\begin{enumerate}[label=(\roman*),leftmargin=25pt]
    \item the base-point identity is
    \(H_p(0,\kappa_{\lambda,p})=0\), established after
    Eq.~\eqref{eq:complex_fixed_point_bias};
    \item joint analyticity on the preliminary polydisc follows from
    Eq.~\eqref{eq:complex_denominator_lower_bound};
    \item uniform non-degeneracy is
    Eq.~\eqref{eq:complex_fixed_point_kappa_derivative};
    \item the derivative bounds are
    Eqs.~\eqref{eq:H_second_derivative_bound} and
    \eqref{eq:H_J_derivative_bound};
    \item the final radii satisfy the lemma's inequalities by
    Eq.~\eqref{eq:rouche_radius_choice}.
\end{enumerate}
Thus all assumptions of
Lemma~\ref{lemma:uniform_analytic_implicit_branch} hold.  Hence, for every
\(J \in D_J\), there is a
unique simple zero \(\kappa_{\lambda,J,p} \in \bbD_{\kappa,p}\) of
\(H_p(J,\cdot)\), and \(J \mapsto \kappa_{\lambda,J,p}\) is analytic on
\(D_J\). In particular,
\begin{align}
    |\kappa_{\lambda,J,p} - \kappa_{\lambda,p}|
    &< \rho,
    \label{eq:bias_kappa_branch_bound}
    \\
    \partial_\kappa H_p(J,\kappa_{\lambda,J,p})
    &\neq 0 \,.
    \label{eq:perturbed_fixed_point_simple_zero}
\end{align}

For real \(J \in (-r,r)\),
\begin{align}
    H_p(J,\overline{\kappa})
    =
    \overline{H_p(J,\kappa)} \,.
\end{align}
Since \(H_p(J,\kappa_{\lambda,J,p}) = 0\), \(H_p(J,\overline{\kappa_{\lambda,J,p}}) = 0\) as well.
Moreover, since \(\kappa_{\lambda,p}\) is real, \(|\overline{\kappa_{\lambda,J,p}} - \kappa_{\lambda,p}| = |\kappa_{\lambda,J,p} - \kappa_{\lambda,p}| < \rho\). Hence, \(\overline{\kappa_{\lambda,J,p}}\) is also a zero of \(H_p(J,\cdot)\) in \(\bbD_{\kappa,p}\).
Uniqueness then implies \(\overline{\kappa_{\lambda,J,p}}=\kappa_{\lambda,J,p}\), so \(\kappa_{\lambda,J,p}\) is real.
Furthermore, Eqs.~\eqref{eq:derivative_trick_kappa_bounds},
\eqref{eq:rho_upper_bound}, and
\eqref{eq:bias_kappa_branch_bound} give
\begin{align}
    \kappa_{\lambda,J,p}
    >
    \kappa_{\lambda,p}-\rho
    \geq
    \lambda-\rho
    \geq
    \frac{\lambda}{2}
    >0 \,.
\end{align}

It remains to identify this branch with the effective regularization parameter of the transformed model.
The implicit-branch construction guarantees only that $\kappa_{\lambda,J,p}$ is the unique zero of $H_p(J,\cdot)$ \emph{inside the disk} $\bbD_{\kappa,p}$; it does not exclude further zeros outside that disk, and in particular it does not by itself identify $\kappa_{\lambda,J,p}$ with the effective regularization parameter of the transformed model, which is characterized as the unique zero \emph{on the positive real axis} by Lemma~\ref{lemma:unique_positive_fixed_point}. The following argument closes this gap.

For real $J \in (-r,r)$, every eigenvalue of $\Sigma_{p,J} = \Sigma_p(I_p + J\Sigma_p)^{-1}$ equals $\xi_i/(1 + J\xi_i)>0$ because $1 + J\xi_i \geq 1 - rM_\Sigma \geq 1/2$; hence $\Sigma_{p,J}$ is positive definite.
Dividing $H_p(J,\kappa) = 0$ by $\kappa$, which is legitimate since $\kappa_{\lambda,J,p}\geq\lambda/2>0$, and using Eq.~\eqref{eq:perturbed_population_resolvent_identity}, we see that the equation $H_p(J,\kappa) = 0$ is precisely Eq.~\eqref{eq:general_positive_fixed_point} with $A = \Sigma_{p,J}$ and $a = \lambda$.
We have shown that $\kappa_{\lambda,J,p}$ is real and positive.
Lemma~\ref{lemma:unique_positive_fixed_point}, which gives uniqueness on the whole positive real axis, therefore implies that $\kappa_{\lambda,J,p}$ is the unique positive solution of the self-consistent equation associated with the transformed model. Hence it is exactly the effective regularization parameter appearing in the deterministic resolvent of that model, as required for the identification in Eq.~\eqref{eq:bias_deterministic_transform_identity} below.

\subsubsection{Complex regularization parameter for the variance}
For complex \(z\) near the fixed real number \(\lambda\), define
\begin{align}
    H_p^{\mathrm{var}}(z,\kappa)
    &:=
    \kappa-z
    -
    \frac{\kappa}{N_p}
    \Tr\left[
        \Sigma_p(\Sigma_p+\kappa I_p)^{-1}
    \right] \,.
    \label{eq:complex_fixed_point_variance}
\end{align}
At \((z,\kappa)=(\lambda,\kappa_{\lambda,p})\),
\begin{align}
    H_p^{\mathrm{var}}(\lambda,\kappa_{\lambda,p})
    &=0,
    \\
    \partial_\kappa
    H_p^{\mathrm{var}}(\lambda,\kappa_{\lambda,p})
    &=
    1-\eta_{\lambda,p}
    \geq a_\lambda \,.
    \label{eq:H_variance_base_point}
\end{align}
The first identity follows from
Eq.~\eqref{eq:derivative_trick_base_fixed_point}, while the derivative
identity is Eq.~\eqref{eq:complex_fixed_point_kappa_derivative}. Moreover,
Eqs.~\eqref{eq:complex_fixed_point_bias} and
\eqref{eq:complex_fixed_point_variance} give the exact relation
\begin{align}
    H_p^{\mathrm{var}}(z,\kappa)
    =
    H_p(0,\kappa)-(z-\lambda) \,.
    \label{eq:H_variance_as_constant_perturbation}
\end{align}
Set
\begin{align}
    \widehat H_p(t,\kappa)
    &:=
    H_p^{\mathrm{var}}(\lambda+t,\kappa)
    =
    H_p(0,\kappa)-t \,.
    \label{eq:variance_shifted_implicit_function}
\end{align}
By Eq.~\eqref{eq:complex_denominator_lower_bound} at \(J = 0\), this
function is jointly analytic on an open neighborhood of
\begin{align}
    |t|\leq\frac{\lambda}{2},
    \qquad
    |\kappa-\kappa_{\lambda,p}|\leq\rho_0 \,.
\end{align}
Choose \(\delta > 0\), independently of \(p\), so that
\begin{align}
    0<\delta
    <
    \min\!\left\{
        \frac{\lambda}{2},
        \frac{a_\lambda\rho}{4}
    \right\} \,,
    \label{eq:variance_disk_radius_choice}
\end{align}
and define
\begin{align}
    D_\lambda
    &:=
    \set{z \in \bbC}[|z - \lambda| < \delta] \,.
    \label{eq:variance_complex_disk}
\end{align}
For the Variance column of
Table~\ref{tab:uniform_implicit_branch_correspondence}, the base-point and
non-degeneracy conditions are Eq.~\eqref{eq:H_variance_base_point}, the
second-\(\kappa\)-derivative bound is
Eq.~\eqref{eq:H_second_derivative_bound}, and
\begin{align}
    |\partial_t\widehat H_p(t,\kappa)| = 1 \,.
\end{align}
Finally, Eqs.~\eqref{eq:rouche_radius_choice} and
\eqref{eq:variance_disk_radius_choice} verify the two inequalities in
Eq.~\eqref{eq:uniform_implicit_branch_radius_choice}. Thus all assumptions
of Lemma~\ref{lemma:uniform_analytic_implicit_branch} hold. Hence, for every
\(z \in  D_\lambda\), the equation
\(H_p^{\mathrm{var}}(z,\kappa)=0\) has a unique simple zero
\(\kappa_{z,p} \in \bbD_{\kappa,p}\), the map
\(z\mapsto\kappa_{z,p}\) is analytic on \(D_\lambda\), and
\begin{align}
    |\kappa_{z,p} - \kappa_{\lambda,p}|
    <\rho \,.
    \label{eq:variance_kappa_branch_bound}
\end{align}

For real \(z \in (\lambda-\delta,\lambda+\delta)\),
\begin{align}
    H_p^{\mathrm{var}}(z,\overline{\kappa})
    =
    \overline{H_p^{\mathrm{var}}(z,\kappa)} \,.
\end{align}
Since \(H_p^{\mathrm{var}}(z,\kappa_{z,p}) = 0\), \(H_p^{\mathrm{var}}(z,\overline{\kappa_{z,p}}) = 0\) as well.
Moreover, since \(\kappa_{\lambda,p}\) is real, \(|\overline{\kappa_{z,p}} - \kappa_{\lambda,p}| = |\kappa_{z,p} - \kappa_{\lambda,p}| < \rho\). Hence, \(\overline{\kappa_{z,p}}\) is also a zero of \(H_p^{\mathrm{var}}(z,\cdot)\) in \(\bbD_{\kappa,p}\).
Uniqueness then implies \(\overline{\kappa_{z,p}}=\kappa_{z,p}\), so \(\kappa_{z,p}\) is real.
Moreover, Eqs.~\eqref{eq:derivative_trick_kappa_bounds}, \eqref{eq:rho_upper_bound}, and \eqref{eq:variance_kappa_branch_bound} give
\begin{align}
    \kappa_{z,p}
    >
    \kappa_{\lambda,p}-\rho
    \geq
    \lambda-\rho
    \geq
    \frac{\lambda}{2}
    >0 \,.
\end{align}

As in the bias case, we must still show that this branch coincides with the effective regularization parameter in Proposition~\ref{prop:resolvent_deterministic_equivalent}, since Rouch\'e's theorem gives uniqueness only inside $\bbD_{\kappa,p}$.
For real $z \in (\lambda-\delta,\lambda+\delta)$ we have $z > 0$, and dividing by $\kappa\geq\lambda/2 > 0$ shows that $H_p^{\mathrm{var}}(z,\kappa)=0$ is Eq.~\eqref{eq:general_positive_fixed_point} with $A=\Sigma_p$ and $a = z$.
Since $\kappa_{z,p}$ has been shown to be real and positive, Lemma~\ref{lemma:unique_positive_fixed_point} identifies it with the unique positive solution of the self-consistent equation in Proposition~\ref{prop:resolvent_deterministic_equivalent}.

Finally, if \(z \in D_\lambda\), then \(|z - \lambda| < \delta\).  Since
Eq.~\eqref{eq:variance_disk_radius_choice} gives
\(\delta<\lambda/2\),
\begin{align}
    \operatorname{Re}z
    >
    \frac{\lambda}{2}
    \qquad
    \text{for every }z \in D_\lambda \,.
    \label{eq:variance_disk_positive_real_part}
\end{align}

\subsection{Derivation of the Bias Term}\label{appdx:subsec:bias_derivation}
We now apply the preceding analytic framework to the bias. For real
\(J \in (-r,r)\), define
\begin{align}
    S_{p,J}
    &:=
    (I_p + J\Sigma_p)^{-1/2},
    \qquad
    \bx_{i,J}
    :=
    S_{p,J}\bx_i \,.
    \label{eq:bias_feature_perturbation}
\end{align}
The choice in Eq.~\eqref{eq:rouche_radius_choice} ensures
\(rM_\Sigma\leq1/2\). By the spectral theorem and
Eq.~\eqref{eq:derivative_trick_uniform_bounds}, every eigenvalue of
\(I_p + J\Sigma_p\) is at least \(1 - rM_\Sigma\geq1/2\). Thus,
\(I_p + J\Sigma_p\) is positive definite and
\begin{align}
    \sup_p\sup_{|J| < r}
    \norm{S_{p,J}}_\infty
    \leq
    (1 - rM_\Sigma)^{-1/2}
    \leq
    \sqrt{2} \,.
    \label{eq:bias_transform_uniform_bound}
\end{align}
Consequently, the linear map $\bx \mapsto (I + J\Sigma)^{-1/2}\bx$ is $\sqrt{2}$-Lipschitz. Since the composition of Lipschitz functions remains Lipschitz, $\bx_J$ also satisfies Assumption~\ref{assum:feature_generation} with the modified constants $C_1' = \sqrt{2} C_1$ and $C_2' = \sqrt{2} C_2$.

Let
\begin{align}
    X_{p,J}
    &:=
    [\bx_{1,J},\ldots,\bx_{N_p,J}]\T \,.
\end{align}
The sample and population covariance matrices of the transformed features are
\begin{align}
    \hat{\Sigma}_{p,J}
    &=
    \frac{1}{N_p}X_{p,J}\T X_{p,J}
    =
    S_{p,J}\hat{\Sigma}_pS_{p,J} \,,
    \\
    \Sigma_{p,J}
    &=
    \E[\bx_{i,J}\bx_{i,J}\T]
    =
    S_{p,J}\Sigma_pS_{p,J}
    =
    \Sigma_p(I_p + J\Sigma_p)^{-1} \,.
    \label{eq:perturbed_covariances}
\end{align}
Because \(\Sigma_p\) commutes with every analytic function of itself,
\begin{align}
    \Sigma_{p,J}
    (\Sigma_{p,J}+\kappa I_p)^{-1}
    =
    \Sigma_p
    \bigl\{\Sigma_p+\kappa(I_p + J\Sigma_p)\bigr\}^{-1} \,.
    \label{eq:perturbed_population_resolvent_identity}
\end{align}
Substituting Eq.~\eqref{eq:perturbed_population_resolvent_identity} into the fixed-point equation~\eqref{eq:fixed_point_equation} shows that the fixed-point equation for the transformed model is \(H_p(J,\kappa) = 0\).
Its positive real solution is \(\kappa_{\lambda,J,p}\).

For real \(J \in (-r,r)\), the identity
\begin{align}
    S_{p,J}
    (\hat{\Sigma}_{p,J}+\lambda I_p)^{-1}
    S_{p,J}
    &=
    (\hat{\Sigma}_p + \lambda I_p + \lambda J\Sigma_p)^{-1}
    \label{eq:perturbed_resolvent_identity}
\end{align}
follows from \(\hat{\Sigma}_{p,J} = S_{p,J}\hat{\Sigma}_pS_{p,J}\) and \(S_{p,J}^{-2} = I_p + J\Sigma_p\).
For \(J \in D_J\), let
\begin{align}
    M_p(J) := \hat{\Sigma}_p + \lambda I_p + \lambda J \Sigma_p = (\hat{\Sigma}_p + \lambda I_p)\big\{I_p + \lambda J(\hat{\Sigma}_p + \lambda I_p)^{-1}\Sigma_p\big\} \,.
    \label{eq:bias_random_resolvent_factorization}
\end{align}
Since \(\hat{\Sigma}_p\) is positive semidefinite, \(\norm*{(\hat{\Sigma}_p + \lambda I_p)^{-1}}_\infty\leq1/\lambda\).
Combining this fact with Eq.~\eqref{eq:derivative_trick_uniform_bounds} and \(rM_\Sigma\leq1/2\), we obtain
\begin{align}
    \norm*{
        \lambda J
        (\hat{\Sigma}_p + \lambda I_p)^{-1}\Sigma_p
    }_\infty
    \leq
    |J|M_\Sigma
    \leq
    r M_\Sigma
    \leq
    \frac{1}{2} \,.
    \label{eq:bias_neumann_perturbation_bound}
\end{align}
The Neumann-series theorem states that, if \(\norm{B}_\infty < 1\), then \(I_p + B\) is invertible and \(\norm{(I_p + B)^{-1}}_\infty \leq (1 - \norm{B}_\infty)^{-1}\).
Hence, from Eqs.~\eqref{eq:bias_random_resolvent_factorization} and
\eqref{eq:bias_neumann_perturbation_bound}, for every \(J \in D_J\), \(M_p(J)\) is invertible and
\begin{align}
    \norm*{
        M_p(J)^{-1}
    }_\infty
    \leq
    \frac{2}{\lambda} \,.
\end{align}
We can therefore define an analytic function:
\begin{align}
    g_p(J)
    &:=
    -\lambda\bth_{*,p}^{\T} M_p(J)^{-1} \bth_{*,p},
    \qquad J \in D_J \,,
    \label{eq:bias_random_function}
\end{align}
which is bounded as follows:
\begin{align}
    \sup_p\sup_{J \in D_J}|g_p(J)|
    \leq
    2C_*^2 \,.
    \label{eq:bias_random_uniform_bound}
\end{align}
Since \(M_p'(J)=\lambda\Sigma_p\), the inverse-derivative identity Eq.~\eqref{eq:matrix_inverse_derivative_identity} gives
\begin{align}
    g_p'(0)
    =
    \lambda^2\bth_{*,p}^{\T}
    (\hat{\Sigma}_p + \lambda I_p)^{-1}
    \Sigma_p
    (\hat{\Sigma}_p + \lambda I_p)^{-1}
    \bth_{*,p}
    =
    \calB_{\lambda,\hat{\Sigma}_p}.
    \label{eq:bias_as_derivative}
\end{align}

We next establish scalar convergence of \(g_p(J)\) for each fixed real \(J \in (-r,r)\) by applying
Corollary~\ref{cor:bounded_bilinear_resolvent_convergence}.

Fix a real \(J \in (-r,r)\).
Let \(R_{p,J}\) and \(\widetilde{R}_{p,J}\) denote, respectively, the
resolvent of \(\hat{\Sigma}_{p,J}\) and its deterministic resolvent in
Theorem~\ref{thm:0.9}, normalized so that
\begin{align}
    \lambda\widetilde{R}_{p,J}(-\lambda)
    =
    \kappa_{\lambda,J,p}
    (\Sigma_{p,J}+\kappa_{\lambda,J,p}I_p)^{-1} \,.
    \label{eq:bias_deterministic_resolvent_normalization}
\end{align}
For every \(J \in D_J\), define the analytic candidate
\begin{align}
    \widetilde{g}_p(J)
    &:=
    -\bth_{*,p}\T
    \kappa_{\lambda,J,p}
    \bigl\{
        \Sigma_p
        +
        \kappa_{\lambda,J,p}(I_p + J\Sigma_p)
    \bigr\}^{-1}
    \bth_{*,p} \,.
    \label{eq:bias_deterministic_function}
\end{align}
This definition is valid on all of \(D_J\) because Eq.~\eqref{eq:complex_denominator_lower_bound} keeps every denominator away from zero.  For real \(J \in (-r,r)\), its identification with the
deterministic resolvent follows from
\begin{align}
    S_{p,J}
    \kappa_{\lambda,J,p}
    (\Sigma_{p,J} + \kappa_{\lambda,J,p}I_p)^{-1}
    S_{p,J}
    =
    \kappa_{\lambda,J,p}
    \bigl\{
        \Sigma_p
        +
        \kappa_{\lambda,J,p}(I_p + J\Sigma_p)
    \bigr\}^{-1} \,.
    \label{eq:bias_deterministic_transform_identity}
\end{align}

The correspondence with Corollary~\ref{cor:bounded_bilinear_resolvent_convergence} is
\begin{align}
    v_p
    &\longleftrightarrow
    v_{p,J}
    :=
    S_{p,J}\bth_{*,p} \,.
    \label{eq:bias_quadratic_form_correspondence}
\end{align}
Equations~\eqref{eq:bias_transform_uniform_bound} and
\eqref{eq:derivative_trick_uniform_bounds} give
\begin{align}
    \sup_p\sup_{J \in (-r,r)}
    \norm{v_{p,J}}_2
    \leq
    \frac{C_*}{\sqrt{1-rM_\Sigma}}
    < \infty \,.
    \label{eq:bias_quadratic_vector_bound}
\end{align}
The assumptions of Corollary~\ref{cor:bounded_bilinear_resolvent_convergence} are therefore satisfied.
Moreover, \(p/N_p \to \gamma \in (0,\infty)\) implies \(N_p \geq c_0p\) for all sufficiently large \(p\) and some \(c_0 > 0\).
Consequently, for the fixed real \(J\), the corollary yields
\begin{align}
    v_{p,J}\T
    \bigl\{
        R_{p,J}(-\lambda)-\widetilde{R}_{p,J}(-\lambda)
    \bigr\}
    v_{p,J}
    \longrightarrow
    0
    \qquad
    \text{a.s.}
\end{align}
Equations~\eqref{eq:perturbed_resolvent_identity},
\eqref{eq:bias_deterministic_resolvent_normalization}, and
\eqref{eq:bias_deterministic_transform_identity} identify
\(-\lambda\) times this quadratic form with
\(g_p(J) - \widetilde{g}_p(J)\). Hence, for every fixed real
\(J \in (-r,r)\),
\begin{align}
    g_p(J) - \widetilde{g}_p(J)
    \longrightarrow
    0
    \qquad
    \text{a.s. as }p \to \infty \,.
    \label{eq:bias_pointwise_convergence}
\end{align}

Eqs.~\eqref{eq:derivative_trick_kappa_bounds} and
\eqref{eq:bias_kappa_branch_bound} give
\(|\kappa_{\lambda,J,p}|\leq K_\lambda+\rho\), while
Eq.~\eqref{eq:complex_denominator_lower_bound} bounds the inverse in
Eq.~\eqref{eq:bias_deterministic_function} by \(2/\lambda\). Therefore,
\begin{align}
    \sup_p\sup_{J \in D_J}|\widetilde{g}_p(J)|
    \leq
    \frac{2(K_\lambda+\rho)}{\lambda}C_*^2 \,.
    \label{eq:bias_deterministic_uniform_bound}
\end{align}

Let
\begin{align}
    E_J
    &:=
    \bbQ\cap(-r,r) \,,
    \label{eq:bias_vitali_countable_set}
    \\
    f_p(J)
    &:=
    g_p(J)-\widetilde{g}_p(J) \,.
    \label{eq:bias_vitali_difference_function}
\end{align}
The functions \(f_p\) are analytic on \(D_J\): analyticity of \(g_p\)
follows from Eqs.~\eqref{eq:bias_random_resolvent_factorization} and
\eqref{eq:bias_neumann_perturbation_bound}, while analyticity of
\(\widetilde{g}_p\) follows from the analytic branch
\(J\mapsto\kappa_{\lambda,J,p}\) and
Eq.~\eqref{eq:complex_denominator_lower_bound}. Equations
\eqref{eq:bias_random_uniform_bound} and
\eqref{eq:bias_deterministic_uniform_bound} give uniform boundedness on
\(D_J\), and Eq.~\eqref{eq:bias_pointwise_convergence} gives pointwise
almost-sure convergence to zero on \(E_J\).

The correspondence with
Proposition~\ref{prop:pathwise_analytic_derivative_transfer} is
\begin{align}
    \begin{aligned}
        D \longleftrightarrow D_J,
        \quad
        E \longleftrightarrow E_J,
        \quad
        h_p \longleftrightarrow f_p,
        \quad
        z_0 \longleftrightarrow0 \,.
    \end{aligned}
    \label{eq:bias_derivative_transfer_correspondence}
\end{align}
The set \(E_J\) is countable and has \(0\) as a limit point in \(D_J\).
The proposition therefore gives
\begin{align}
    f_p'(0)
    \longrightarrow
    0
    \qquad
    \text{a.s.}
    \label{eq:bias_vitali_derivative_convergence}
\end{align}
Using Eqs.~\eqref{eq:bias_vitali_difference_function},
\eqref{eq:bias_vitali_derivative_convergence}, and
\eqref{eq:bias_as_derivative}, we obtain
\begin{align}
    \calB_{\lambda,\hat{\Sigma}_p}
    -
    \widetilde{g}_p'(0)
    \longrightarrow0
    \qquad\text{a.s.}
    \label{eq:bias_DE_derivative}
\end{align}

It remains to compute \(\widetilde{g}_p'(0)\).
The derivative formula in Lemma~\ref{lemma:uniform_analytic_implicit_branch}, under the Bias correspondence in Table~\ref{tab:uniform_implicit_branch_correspondence}, gives
\begin{align}
    \left.
    \frac{\partial\kappa_{\lambda,J,p}}{\partial J}
    \right|_{J=0}
    &=
    -
    \frac{
        \partial_JH_p(0,\kappa_{\lambda,p})
    }{
        \partial_\kappa H_p(0,\kappa_{\lambda,p})
    } \,.
\end{align}
By Eqs.~\eqref{eq:H_J_derivative}, \eqref{eq:derivative_trick_eta}, and
\eqref{eq:complex_fixed_point_kappa_derivative},
\begin{align}
    \left.
    \frac{\partial\kappa_{\lambda,J,p}}{\partial J}
    \right|_{J=0}
    =
    -
    \frac{
        \kappa_{\lambda,p}^2\eta_{\lambda,p}
    }{
        1-\eta_{\lambda,p}
    } \,.
    \label{eq:kappa_J_derivative}
\end{align}

For brevity, write
\begin{align}
    \kappa_p
    :=
    \kappa_{\lambda,p},
    \qquad
    \dot{\kappa}_p
    :=
    \left.
    \frac{\partial\kappa_{\lambda,J,p}}{\partial J}
    \right|_{J=0},
    \qquad
    Q_p
    :=
    (\Sigma_p+\kappa_pI_p)^{-1} \,.
    \label{eq:bias_shorthand_definitions}
\end{align}
Apply the product rule and the inverse-derivative identity
in Eq.~\eqref{eq:matrix_inverse_derivative_identity} to
Eq.~\eqref{eq:bias_deterministic_function}. Using the notation in
Eq.~\eqref{eq:bias_shorthand_definitions} then yields
\begin{align}
    \widetilde{g}_p'(0)
    &=
    \bth_{*,p}\T
    \bigl[
        -\dot{\kappa}_pQ_p
        +
        \kappa_pQ_p
        (\dot{\kappa}_pI_p+\kappa_p\Sigma_p)
        Q_p
    \bigr]
    \bth_{*,p} \,.
    \label{eq:bias_deterministic_derivative_expanded}
\end{align}
Multiplying the defining relation
\(Q_p(\Sigma_p+\kappa_pI_p)=I_p\) by \(Q_p\), and using the commutativity of
\(Q_p\) and \(\Sigma_p\), gives
\begin{align}
    Q_p-\kappa_pQ_p^2
    =
    Q_p^2\Sigma_p \,.
    \label{eq:bias_Q_identity}
\end{align}
Substituting Eq.~\eqref{eq:bias_Q_identity} into
Eq.~\eqref{eq:bias_deterministic_derivative_expanded}, we obtain
\begin{align}
    \widetilde{g}_p'(0)
    &=
    (\kappa_p^2-\dot{\kappa}_p)
    \bth_{*,p}\T Q_p^2\Sigma_p\bth_{*,p} \,.
    \label{eq:bias_deterministic_derivative_reduced}
\end{align}
Equation~\eqref{eq:kappa_J_derivative} gives
\begin{align}
    \kappa_p^2-\dot{\kappa}_p
    =
    \frac{\kappa_p^2}{1-\eta_{\lambda,p}} \,.
    \label{eq:bias_kappa_prefactor_identity}
\end{align}
Combining Eqs.~\eqref{eq:bias_deterministic_derivative_reduced} and
\eqref{eq:bias_kappa_prefactor_identity}, and restoring
\(\kappa_p=\kappa_{\lambda,p}\), gives
\begin{align}
    \widetilde{g}_p'(0)
    =
    \frac{\kappa_{\lambda,p}^2}{1-\eta_{\lambda,p}}
    \bth_{*,p}\T
    (\Sigma_p+\kappa_{\lambda,p}I_p)^{-2}
    \Sigma_p
    \bth_{*,p} \,.
    \label{eq:bias_DE_theta}
\end{align}
By construction, \(\Sigma_p\) is positive definite on the effective feature
space introduced in Section~\ref{subsec:eigendecomposition}, and
\(\bth_{*,p}=\Sigma_p^{-1/2}\bs{\beta}_{*,p}\). Since every matrix in
Eq.~\eqref{eq:bias_DE_theta} is a function of \(\Sigma_p\),
\begin{align}
    \bth_{*,p}\T
    (\Sigma_p+\kappa_{\lambda,p}I_p)^{-2}
    \Sigma_p
    \bth_{*,p}
    =
    \bs{\beta}_{*,p}\T
    (\Sigma_p+\kappa_{\lambda,p}I_p)^{-2}
    \bs{\beta}_{*,p} \,.
\end{align}
Combining this identity with Eq.~\eqref{eq:bias_DE_derivative}, we conclude
that
\begin{align}
    \calB_{\lambda,\hat{\Sigma}_p}
    \asymp
    \frac{\kappa_{\lambda,p}^2}{1-\eta_{\lambda,p}}
    \bs{\beta}_{*,p}\T
    (\Sigma_p+\kappa_{\lambda,p}I_p)^{-2}
    \bs{\beta}_{*,p} \,.
    \label{eq:bias_DE_final}
\end{align}

\subsection{Derivation of the Variance Term}
\label{appdx:subsec:variance_derivation}

We next derive the deterministic equivalent of
Eq.~\eqref{eq:derivative_trick_variance_definition}. For \(z \in D_\lambda\),
define
\begin{align}
    T_p(z)
    &:=
    \frac{\sigma^2}{N_p}
    \Tr\left[
        \Sigma_p
        z(\hat{\Sigma}_p + zI_p)^{-1}
    \right] \,.
    \label{eq:variance_random_function}
\end{align}
Applying Eq.~\eqref{eq:matrix_inverse_derivative_identity} with
\(A(z) = \hat{\Sigma}_p + zI_p\), followed by the algebraic identity
\(I_p - z(\hat{\Sigma}_p + zI_p)^{-1}
=\hat{\Sigma}_p(\hat{\Sigma}_p + zI_p)^{-1}\), gives
\begin{align}
    \frac{\partial}{\partial z}
    \bigl\{
        z(\hat{\Sigma}_p + zI_p)^{-1}
    \bigr\}
    =
    (\hat{\Sigma}_p + zI_p)^{-1}
    -
    z(\hat{\Sigma}_p + zI_p)^{-2}
    =
    \hat{\Sigma}_p
    (\hat{\Sigma}_p + zI_p)^{-2} \,.
    \label{eq:variance_resolvent_derivative_identity}
\end{align}
Substituting Eq.~\eqref{eq:variance_resolvent_derivative_identity} into
Eq.~\eqref{eq:variance_random_function} and setting \(z=\lambda\), we obtain
\begin{align}
    T_p'(\lambda)
    =
    \calV_{\lambda,\hat{\Sigma}_p} \,.
    \label{eq:variance_as_derivative}
\end{align}

The corresponding deterministic function is
\begin{align}
    \widetilde{T}_p(z)
    &:=
    \frac{\sigma^2\kappa_{z,p}}{N_p}
    \Tr\left[
        \Sigma_p(\Sigma_p+\kappa_{z,p}I_p)^{-1}
    \right] \,.
    \label{eq:variance_deterministic_function}
\end{align}
For every fixed real \(z \in (\lambda - \delta,\lambda + \delta)\),
Proposition~\ref{prop:resolvent_deterministic_equivalent} identifies the
corresponding deterministic resolvent as
\begin{align}
    z(\hat{\Sigma}_p + zI_p)^{-1}
    \asymp_F
    \kappa_{z,p}
    (\Sigma_p + \kappa_{z,p}I_p)^{-1} \,.
    \label{eq:variance_real_axis_resolvent_DE}
\end{align}
Choose \(0 < \epsilon < \lambda - \delta\), which is possible by
Eq.~\eqref{eq:variance_disk_radius_choice}.  The correspondence with
Corollary~\ref{cor:normalized_trace_resolvent_convergence} is
\begin{align}
    \lambda\text{ in the corollary} & \longleftrightarrow z \,,
    \qquad
    D_p & \longleftrightarrow \Sigma_p \,.
    \label{eq:variance_normalized_trace_correspondence}
\end{align}
The test matrices are admissible by
Eq.~\eqref{eq:derivative_trick_uniform_bounds}.  The corollary therefore
gives
\begin{align}
    \Delta_p(z)
    :=
    \frac{1}{p}
    \Tr\left[
        \Sigma_p
        \left\{
            z(\hat{\Sigma}_p + zI_p)^{-1}
            -
            \kappa_{z,p}
            (\Sigma_p+\kappa_{z,p}I_p)^{-1}
        \right\}
    \right]
    \longrightarrow
    0
    \qquad
    \text{a.s.}
    \label{eq:variance_normalized_trace_convergence}
\end{align}
for every such fixed real \(z\). By the definitions of \(T_p\) and
\(\widetilde{T}_p\),
\begin{align}
    T_p(z)-\widetilde{T}_p(z)
    =
    \sigma^2\frac{p}{N_p}
    \Delta_p(z) \,.
\end{align}
Since \(p/N_p\) is uniformly bounded by
Eq.~\eqref{eq:derivative_trick_aspect_ratio_bound}, it follows that
\begin{align}
    T_p(z)-\widetilde{T}_p(z)
    \longrightarrow0
    \qquad
    \text{a.s.}
    \label{eq:variance_pointwise_convergence}
\end{align}
for every such fixed real \(z\).

We now establish the uniform bounds needed for Vitali's theorem. By the
spectral theorem, if \(\mu \geq 0\) is an eigenvalue of
\(\hat{\Sigma}_p\), then
\(|\mu+z| \geq \operatorname{Re}(\mu+z)
\geq\operatorname{Re}z\). Together with
Eq.~\eqref{eq:variance_disk_positive_real_part}, this gives
\begin{align}
    \norm{(\hat{\Sigma}_p + zI_p)^{-1}}_\infty
    \leq
    \frac{1}{\operatorname{Re}z}
    \leq
    \frac{2}{\lambda} \,.
    \label{eq:variance_random_resolvent_bound}
\end{align}
Furthermore, Eqs.~\eqref{eq:variance_disk_radius_choice} and
\eqref{eq:variance_complex_disk} imply
\( |z| \leq \lambda + \delta < 3\lambda/2 \). Therefore,
\begin{align}
    \norm{z(\hat{\Sigma}_p + zI_p)^{-1}}_\infty
    \leq3 \,.
    \label{eq:variance_scaled_resolvent_bound}
\end{align}
For a positive semidefinite matrix \(A\) and an arbitrary matrix \(B\),
H\"older's inequality for Schatten norms gives
\(|\Tr[AB]| \leq \Tr[A]\norm{B}_\infty\).
Applying it with \(A=\Sigma_p\), using
Eq.~\eqref{eq:variance_scaled_resolvent_bound}, and then using
Eqs.~\eqref{eq:derivative_trick_aspect_ratio_bound} and
\eqref{eq:derivative_trick_uniform_bounds}, we obtain
\begin{align}
    |T_p(z)|
    &\leq
    \frac{3\sigma^2}{N_p}\Tr\Sigma_p
    \leq
    3\sigma^2\Gamma M_\Sigma \,.
    \label{eq:variance_random_uniform_bound}
\end{align}

The fixed-point equation
\(H_p^{\mathrm{var}}(z,\kappa_{z,p})=0\) can be divided by
\(\kappa_{z,p}\).
Indeed, Eqs.~\eqref{eq:derivative_trick_kappa_bounds},
\eqref{eq:rho_upper_bound}, and
\eqref{eq:variance_kappa_branch_bound} give
\begin{align}
    |\kappa_{z,p}|
    \geq
    \kappa_{\lambda,p}
    -
    |\kappa_{z,p} - \kappa_{\lambda,p}|
    >
    \lambda - \rho
    \geq \frac{\lambda}{2} >0 \,.
\end{align}

Using the definition Eq.~\eqref{eq:complex_fixed_point_variance}, we obtain
\begin{align}
    \frac{1}{N_p}
    \Tr\left[
        \Sigma_p(\Sigma_p+\kappa_{z,p}I_p)^{-1}
    \right]
    +
    \frac{z}{\kappa_{z,p}}
    =1 \,.
    \label{eq:variance_complex_fixed_point_normalized}
\end{align}
Substituting Eq.~\eqref{eq:variance_complex_fixed_point_normalized} into
Eq.~\eqref{eq:variance_deterministic_function}, we find
\begin{align}
    \widetilde{T}_p(z)
    =
    \sigma^2(\kappa_{z,p}-z) \,.
    \label{eq:variance_deterministic_simplified}
\end{align}
Equations~\eqref{eq:derivative_trick_kappa_bounds} and
\eqref{eq:variance_kappa_branch_bound} give
\(|\kappa_{z,p}| \leq K_\lambda + \rho\). Together with
\(|z| < 3\lambda/2\), established before
Eq.~\eqref{eq:variance_scaled_resolvent_bound}, this proves
\begin{align}
    \sup_p\sup_{z \in D_\lambda}
    |\widetilde{T}_p(z)|
    \leq
    \sigma^2
    \left(
        K_\lambda + \rho + \frac{3\lambda}{2}
    \right)
    < \infty \,.
    \label{eq:variance_deterministic_uniform_bound}
\end{align}

Let
\begin{align}
    E_\lambda
    &:=
    \bbQ\cap(\lambda - \delta,\lambda + \delta) \,,
    \label{eq:variance_vitali_countable_set}
    \\
    F_p(z)
    &:=
    T_p(z) - \widetilde{T}_p(z) \,.
    \label{eq:variance_vitali_difference_function}
\end{align}
The functions \(F_p\) are analytic on \(D_\lambda\): analyticity of \(T_p\) follows from Eq.~\eqref{eq:variance_random_resolvent_bound}, while analyticity of \(\widetilde{T}_p\) follows from \(z\mapsto\kappa_{z,p}\) and Eq.~\eqref{eq:complex_denominator_lower_bound} at \(J = 0\). Equations~\eqref{eq:variance_random_uniform_bound} and \eqref{eq:variance_deterministic_uniform_bound} give uniform boundedness, and Eq.~\eqref{eq:variance_pointwise_convergence} gives pointwise almost-sure convergence to zero on \(E_\lambda\).

The correspondence with Proposition~\ref{prop:pathwise_analytic_derivative_transfer} is
\begin{align}
    \begin{aligned}
        D \longleftrightarrow D_\lambda,
        \quad
        E \longleftrightarrow E_\lambda,
        \quad
        h_p \longleftrightarrow F_p,
        \quad
        z_0 \longleftrightarrow\lambda \,.
    \end{aligned}
    \label{eq:variance_derivative_transfer_correspondence}
\end{align}
The set \(E_\lambda\) is countable and has \(\lambda\) as a limit point in \(D_\lambda\). The proposition therefore gives
\begin{align}
    F_p'(\lambda)
    \longrightarrow
    0
    \qquad
    \text{a.s.}
    \label{eq:variance_vitali_derivative_convergence}
\end{align}
Using Eqs.~\eqref{eq:variance_vitali_difference_function}, \eqref{eq:variance_vitali_derivative_convergence}, and \eqref{eq:variance_as_derivative}, we obtain
\begin{align}
    \calV_{\lambda,\hat{\Sigma}_p}
    -
    \widetilde{T}_p'(\lambda)
    \longrightarrow0
    \qquad\text{a.s.}
    \label{eq:variance_DE_derivative}
\end{align}

It remains to differentiate the effective regularization parameter.
The derivative formula in Lemma~\ref{lemma:uniform_analytic_implicit_branch}, applied to \(\widehat H_p(t,\kappa) = H_p^{\mathrm{var}}(\lambda + t,\kappa)\), gives
\begin{align}
    \left.
    \frac{\partial\kappa_{z,p}}{\partial z}
    \right|_{z=\lambda}
    &=
    -
    \frac{
        \partial_z H_p^{\mathrm{var}}
        (\lambda,\kappa_{\lambda,p})
    }{
        \partial_\kappa H_p^{\mathrm{var}}
        (\lambda,\kappa_{\lambda,p})
    } \,.
\end{align}
Here Eq.~\eqref{eq:complex_fixed_point_variance} gives
\(\partial_zH_p^{\mathrm{var}}=-1\), while
Eq.~\eqref{eq:H_variance_base_point} gives
\(\partial_\kappa H_p^{\mathrm{var}}=1-\eta_{\lambda,p}\). Therefore,
\begin{align}
    \left.
    \frac{\partial\kappa_{z,p}}{\partial z}
    \right|_{z=\lambda}
    =
    \frac{1}{1-\eta_{\lambda,p}} \,.
    \label{eq:kappa_derivative}
\end{align}
Differentiating Eq.~\eqref{eq:variance_deterministic_simplified} and then
using Eq.~\eqref{eq:kappa_derivative} yields
\begin{align}
    \widetilde{T}_p'(\lambda)
    =
    \sigma^2
    \left(
        \frac{1}{1-\eta_{\lambda,p}}-1
    \right)
    =
    \sigma^2
    \frac{\eta_{\lambda,p}}{1-\eta_{\lambda,p}} \,.
    \label{eq:variance_deterministic_derivative}
\end{align}
Combining Eqs.~\eqref{eq:variance_DE_derivative} and
\eqref{eq:variance_deterministic_derivative}, we obtain
\begin{align}
    \calV_{\lambda,\hat{\Sigma}_p}
    \asymp
    \sigma^2
    \frac{\eta_{\lambda,p}}{1-\eta_{\lambda,p}} \,.
    \label{eq:variance_DE_final}
\end{align}

Suppressing the subscript \(p\) in
Eqs.~\eqref{eq:bias_DE_final} and \eqref{eq:variance_DE_final}, we arrive at
\begin{align}
    \calB_{\lambda,\hat{\Sigma}}
    &\asymp
    \frac{\kappa_\lambda^2}{1-\eta_\lambda}
    \bs{\beta}_*\T
    (\Sigma+\kappa_\lambda I_p)^{-2}
    \bs{\beta}_* \,,
    \\
    \calV_{\lambda,\hat{\Sigma}}
    &\asymp
    \sigma^2
    \frac{\eta_\lambda}{1-\eta_\lambda} \,,
\end{align}
which are the deterministic equivalents stated in
Theorem~\ref{thm:test_risk_DE}.

\begin{remark}[Fixed positive regularization]
\label{remark:fixed_positive_regularization}
    The proof is local around an arbitrary fixed \(\lambda > 0\). The construction
    of the common complex neighborhoods relies on
    \begin{align}
        1-\eta_{\lambda,p}
        \geq
        \frac{\lambda}{\kappa_{\lambda,p}} \,.
    \end{align}
    If \(\lambda=\lambda_p\to0\) simultaneously with \(p \to \infty\), this lower
    bound may vanish, and the radii of the \(p\)-uniform complex neighborhoods may shrink with \(p\).
    Quantitatively, \(a_\lambda = \lambda/K_\lambda\),
    \(\rho = \order{a_\lambda}\) by Eq.~\eqref{eq:rouche_radius_choice}, and
    \(r = \order{a_\lambda\rho}\), so both radii shrink to zero as \(\lambda \downarrow 0\).
    The argument above therefore does not, by itself, establish a
    deterministic equivalent that is uniform in a simultaneous ridgeless limit.
\end{remark}

\subsection{Proof of Corollary~\ref{cor:fixed_ridge_normalized_quantum_kernel}}\label{appdx:subsec:fixed_ridge_normalized_quantum_kernel}
\begin{proof}[Proof of Corollary~\ref{cor:fixed_ridge_normalized_quantum_kernel}]
    Multiplying the self-consistent equation~\eqref{eq:self-consistent_eq} by $\kappa_{\lambda,p}$ and using $\Tr\Sigma_p=1$ gives
    \begin{align}
        \kappa_{\lambda,p} - \lambda
        &=
        \frac{\kappa_{\lambda,p}}{N_p}
        \Tr\!\left[
            \Sigma_p(\Sigma_p + \kappa_{\lambda,p}I_p)^{-1}
        \right]
        \leq
        \frac{\Tr\Sigma_p}{N_p}
        =
        \frac{1}{N_p}.
    \end{align}
    Moreover, by \eqref{eq:eta_upper_bound_from_fixed_point},
    \begin{align}
        0 \leq \eta_{\lambda,p}
        \leq 1 - \frac{\lambda}{\kappa_{\lambda,p}}
        = \frac{\kappa_{\lambda,p} - \lambda}{\kappa_{\lambda,p}}
        \leq \frac{1}{N_p\lambda}.
    \end{align}
    Since $N_p \to \infty$, the stated limits follow for every fixed $\lambda > 0$.
\end{proof}

\subsection{Training Error and Generalized Cross-Validation}\label{appdx:sec:training_error_DE}
We close this section by presenting the deterministic equivalent of the training error and the generalized cross-validation (GCV) risk estimator~\cite{atanasov2024scaling,MisiakiewiczSaeed2024}, which can be derived using techniques similar to those used for the test risk.
The training error is defined as
\begin{align}
    \hat{R}^{\mathrm{train}}_{\lambda,\hat{\Sigma}}
    &:=
    \frac{1}{N_{\mathrm{tr}}} \norm*{\by - X\hat{\bth}}^2 \,.
\end{align}
The deterministic equivalent of the training error is
\begin{align}
    \hat{R}^{\mathrm{train}}_{\lambda,\hat{\Sigma}}
    \asymp
    \frac{\lambda^2}{\kappa_\lambda^2}
    R^{\mathrm{DE}}_{\lambda,\Sigma} \,.
\end{align}

The GCV risk estimator is defined as
\begin{align}
    R_{\lambda,\hat{\Sigma}}^{\mathrm{GCV}}
    &:=
    \frac{
        \hat{R}^{\mathrm{train}}_{\lambda,\hat{\Sigma}}
    }{
        (1 - \frac{1}{N_{\mathrm{tr}}}\operatorname{Tr}[\hat{\Sigma}(\hat{\Sigma} + \lambda I_p)^{-1}])^2
    } \,.
\end{align}
We define $G := XX\T/N_{\mathrm{tr}}$.
Noting that
\begin{align}
    \hat{R}^{\mathrm{train}}_{\lambda,\hat{\Sigma}}
    =
    \frac{\lambda^2}{N_{\mathrm{tr}}}
    \by\T
    \qty(
        G + \lambda I_{N_{\mathrm{tr}}}
    )^{-2}
    \by \,,
    \quad
    1 - \frac{1}{N_{\mathrm{tr}}}\operatorname{Tr}[\hat{\Sigma}(\hat{\Sigma} + \lambda I_p)^{-1}]
    =
    \frac{\lambda}{N_{\mathrm{tr}}}
    \operatorname{Tr}[(G + \lambda I_{N_{\mathrm{tr}}})^{-1}] \,,
\end{align}
we obtain the following expression for the GCV risk estimator:
\begin{align}
    R_{\lambda,\hat{\Sigma}}^{\mathrm{GCV}}
    &:=
    \frac{
        \frac{1}{N_{\mathrm{tr}}}
        \by\T(G + \lambda I_{N_{\mathrm{tr}}})^{-2}\by
    }{
    \qty(
        \frac{1}{N_{\mathrm{tr}}}
        \operatorname{Tr}
        [(G + \lambda I_{N_{\mathrm{tr}}})^{-1}]
    )^2
    } \,.
\end{align}

In Definition~\ref{def:frobenius_deterministic_equivalent}, we only consider the case where $\{B_p\}_p$ are deterministic. More generally, the comparison sequence \(B_p\) may itself be random.
Whether $B_p$ is deterministic or not, if Eq.~\eqref{eq:def:frobenius_deterministic_equivalent} holds, we say that $B_p$ is an \textbf{asymptotic equivalent} of $A_p$, and we also denote it as $A_p \asymp_F B_p$.

Under the same assumptions as in Theorem~\ref{thm:test_risk_DE}, the GCV estimator is expected to be asymptotically equivalent to the test risk, i.e., $R_{\lambda,\hat{\Sigma}}^{\mathrm{GCV}} \asymp R_{\lambda,\hat{\Sigma}}$.
Here we drop the subscript $F$ in the notation $\asymp_F$ since these two quantities are scalars as stated after Definition~\ref{def:frobenius_deterministic_equivalent}.
The key point is that the GCV risk estimator depends only on the training samples. This avoids the estimation of the population covariance matrix $\Sigma$ as in Section~\ref{subsec:estimation_sigma_beta}, allowing us to estimate the test-risk behavior more efficiently.
The GCV risk estimator and the test-risk DE $R^{\mathrm{DE}}_{\lambda,\Sigma}$ are compared in Fig.~\ref{fig:gcv_synthetic_5qubits_TPA} for the synthetic dataset with 5-qubit TPA.

\begin{figure}[H]
    \centering
    \includegraphics[width=10cm]{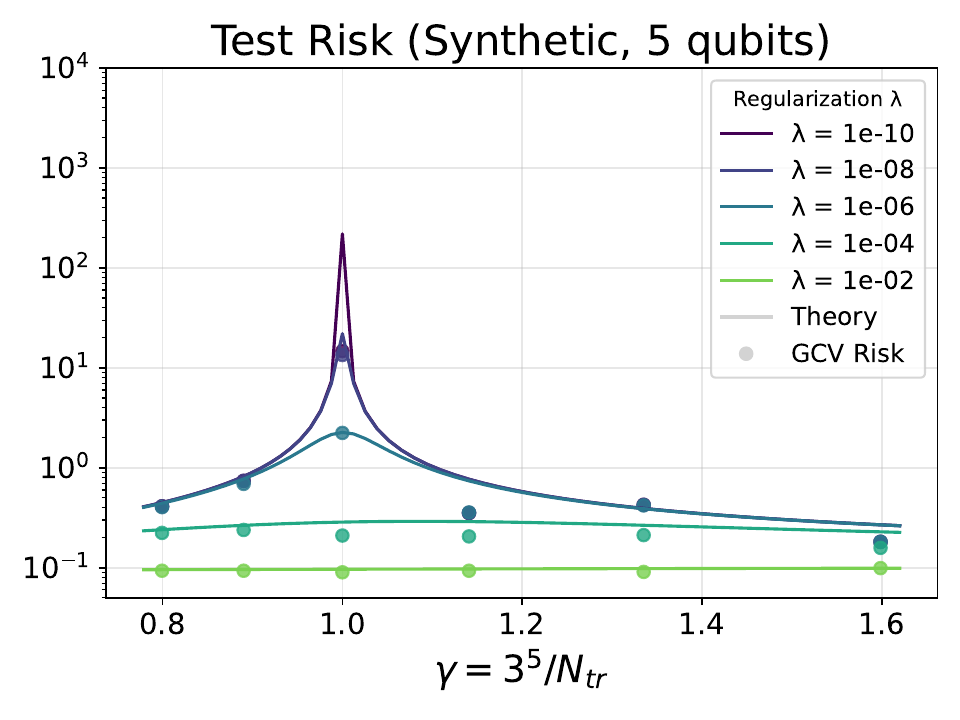}
    \caption{Comparison of the finite-dimensional test-risk DE $R^{\mathrm{DE}}_{\lambda,\Sigma}$ and the GCV risk estimator $R_{\lambda,\hat{\Sigma}}^{\mathrm{GCV}}$ for the synthetic dataset with 5-qubit TPA.}
    \label{fig:gcv_synthetic_5qubits_TPA}
\end{figure}

\newpage
\section{Properties of the Effective Regularization Parameter $\kappa_\lambda$}\label{appdx:sec:properties_of_kappa}

This appendix studies the properties of the effective regularization parameter $\kappa_\lambda$ as $\lambda \downarrow 0$ that are shown in Section~\ref{subsec:asymptotic_test_risk}.
Throughout this section, the effective feature dimension $p$ and the number of training samples $N_{\mathrm{tr}}$ are \emph{held fixed}, and we use the self-consistent equation~\eqref{eq:self-consistent_eq} as a relation between the two scalars $\lambda$ and $\kappa_\lambda$.
The constants appearing below may depend on $p$ and on the spectrum of $\Sigma$.
We emphasize that these statements do \emph{not} assert a deterministic equivalent in a simultaneous limit $p,N_{\mathrm{tr}} \to \infty$ with $\lambda = \lambda_p \downarrow 0$, for which see Remark~\ref{remark:fixed_positive_regularization}.

We write $\Sigma = \diag(\xi_1,\ldots,\xi_p)$ with $\xi_1 \geq \cdots \geq \xi_p > 0$ and set $\xi_{\min} := \xi_p$; positive definiteness holds by construction in the effective feature space (Section~\ref{subsec:test_risk_decomposition}).
We use the effective DOF of Eq.~\eqref{eq:effective_DOF} and define $\eta$ for a general scalar $\kappa \geq 0$:
\begin{align}
    \mathrm{df}_s(\kappa)
    :=
    \Tr[\Sigma^s(\Sigma + \kappa I_p)^{-s}]
    =
    \sum_{i=1}^p\left(\frac{\xi_i}{\xi_i + \kappa}\right)^{\!s} \,,
    \qquad
    \eta(\kappa)
    :=
    \frac{1}{N_{\mathrm{tr}}}\mathrm{df}_2(\kappa) \,.
\end{align}
Thus, the quantity denoted by $\eta_\lambda$ in Theorem~\ref{thm:test_risk_DE} is $\eta(\kappa_\lambda)$ in this appendix.
We record the elementary properties of $\mathrm{df}_1$ that we use repeatedly.
\begin{enumerate}[label=(P\arabic*),leftmargin=32pt]
    \item\label{item:df1_range} $\mathrm{df}_1$ is continuous on $[0,\infty)$ with $\mathrm{df}_1(0)=p$ and $\mathrm{df}_1(\kappa)\to0$ as $\kappa \to \infty$.
    \item\label{item:df1_decreasing} $\mathrm{df}_1'(\kappa)=-\sum_{i=1}^p\xi_i/(\xi_i+\kappa)^2<0$, so $\mathrm{df}_1$ is strictly decreasing.
    \item\label{item:df1_jensen1} The map $x\mapsto 1/(1 + \kappa x)$ is strictly convex on $(0,\infty)$ for $\kappa > 0$, since its second derivative is $2\kappa^2/(1 + \kappa x)^3 > 0$.
    Define $\bar{\xi}_{\mathrm{inv}} := \frac{1}{p}\sum_{i=1}^p \xi_i^{-1} = \frac{1}{p}\Tr[\Sigma^{-1}]$.
    Jensen's inequality ($\E[f(X)] \geq f(\E[X])$) over the empirical eigenvalue distribution therefore gives
    \begin{align}\label{eq:df_inequality1}
        \frac{1}{p}\mathrm{df}_1(\kappa)
        \geq
        \frac{1}{1 + \kappa \bar{\xi}_{\mathrm{inv}}} \,.
    \end{align}
    \item\label{item:df1_jensen2} The map $x \mapsto x/(x+\kappa)$ is strictly concave on $(0,\infty)$ for $\kappa > 0$, since its second derivative is $-2\kappa/(x + \kappa)^3 < 0$.
    Define $\bar{\xi} := \frac{1}{p}\sum_{i=1}^p\xi_i = \Tr[\Sigma]/p$.
    Jensen's inequality ($\E[f(X)] \leq f(\E[X])$) over the empirical eigenvalue distribution therefore gives
    \begin{align}\label{eq:df_inequality2}
        \frac{1}{p}\mathrm{df}_1(\kappa)
        \leq
        \frac{\bar{\xi}}{\kappa + \bar{\xi}} \,.
    \end{align}
\end{enumerate}

\subsection{The ridgeless effective regularization parameter $\kappa_0$}\label{appdx:subsec:kappa_zero}
For $\lambda > 0$, the self-consistent equation~\eqref{eq:self-consistent_eq} is equivalent to
\begin{align}
    \lambda = F(\kappa_\lambda),
    \qquad
    F(\kappa)
    := \kappa\left(1 - \frac{\mathrm{df}_1(\kappa)}{N_{\mathrm{tr}}}\right).
    \label{eq:self_consistent_equation2}
\end{align}
Differentiating Eq.~\eqref{eq:self_consistent_equation2} and using
$-\kappa\,\mathrm{df}_1'(\kappa) = \mathrm{df}_1(\kappa)-\mathrm{df}_2(\kappa)$ yields
\begin{align}\label{eq:lambda_prime_of_kappa}
    F'(\kappa)
    =
    1-\frac{\mathrm{df}_2(\kappa)}{N_{\mathrm{tr}}}
    =
    1-\eta(\kappa) \,.
\end{align}

Note that $F(0) = 0$ and $F(\kappa)\to\infty$ as $\kappa\to\infty$.
The latter follows from the fact that $\mathrm{df}_1(\kappa) \to 0$ as $\kappa \to \infty$.
For $\kappa \geq 0$, the zero equation has the two possible forms
\begin{align}\label{eq:lambda_zero_roots}
    F(\kappa) = 0
    \iff
    \kappa = 0
    \quad\text{or}\quad
    \left[\kappa > 0\ \text{and}\ \mathrm{df}_1(\kappa)=N_{\mathrm{tr}}\right] \,.
\end{align}
We define the largest nonnegative root by
\begin{align}\label{eq:kappa_zero_definition}
    \kappa_e
    :=
    \max\set{\kappa \geq 0}[F(\kappa) = 0] \,.
\end{align}
This maximum exists because the zero set is nonempty, as it contains $0$; it is closed, by continuity of $F$; and it is bounded, because $F(\kappa) \to \infty$. Since $\mathrm{df}_1$ is strictly decreasing from $p$ to $0$, the positive alternative in Eq.~\eqref{eq:lambda_zero_roots} has no solution when $\tilde{\gamma} \leq 1$ and has exactly one solution when $\tilde{\gamma} > 1$. Thus the zero set is $\{0\}$ for $\tilde{\gamma} \leq 1$ and $\{0,\kappa_e\}$ with $\kappa_e > 0$ for $\tilde{\gamma} > 1$.
In the latter case, since $\mathrm{df}_1$ is strictly decreasing and $\mathrm{df}_1(\kappa_e) = N_{\mathrm{tr}}$, the sign of $F$ is given by
\begin{align}
    \begin{cases}
        F(\kappa) < 0, & 0 < \kappa < \kappa_e,\\
        F(\kappa) = 0, & \kappa = 0 \text{ or } \kappa = \kappa_e,\\
        F(\kappa) > 0, & \kappa > \kappa_e.
    \end{cases}
    \label{eq:sign_of_F}
\end{align}
Consequently, for every $\lambda > 0$, any solution of $F(\kappa) = \lambda$ must satisfy $\kappa > \kappa_e$.
Thus, the root $\kappa = 0$ is separated from all positive-$\lambda$ solutions and cannot be the $\lambda \downarrow 0$ endpoint of the positive-$\lambda$ solution branch.

For every $\kappa > \kappa_e$, we have $F(\kappa) > 0$, or equivalently $\mathrm{df}_1(\kappa) < N_{\mathrm{tr}}$. Since $\mathrm{df}_2(\kappa)\leq\mathrm{df}_1(\kappa)$, Eq.~\eqref{eq:lambda_prime_of_kappa} gives
\begin{align}
    F'(\kappa)>0
    \qquad\text{for every }\kappa > \kappa_e \,.
\end{align}
Consequently, the restriction $F|_{[\kappa_e,\infty)}$ is a strictly increasing continuous bijection from $[\kappa_e,\infty)$ onto $[0,\infty)$. For every $\lambda > 0$, the unique positive solution is therefore
\begin{align}\label{eq:kappa_inverse_branch}
    \kappa_\lambda
    =
    \left(F|_{[\kappa_e,\infty)}\right)^{-1}(\lambda) \,.
\end{align}
Thus, the map $\lambda\mapsto\kappa_\lambda$ is strictly increasing. The inverse is continuous at $\lambda = 0$, so the ridgeless endpoint exists. We define it separately by
\begin{align}\label{eq:kappa_zero_limit_definition}
    \kappa_0
    :=
    \lim_{\lambda\downarrow0}\kappa_\lambda \,.
\end{align}
Because $F|_{[\kappa_e,\infty)}$ is a continuous strictly increasing bijection onto $[0,\infty)$, its inverse is continuous and \emph{strictly} increasing on $[0,\infty)$. Hence, taking the limit $\lambda \downarrow 0$ in Eq.~\eqref{eq:kappa_inverse_branch} gives
\begin{align}\label{eq:kappa_zero_equals_kappa_e}
    \kappa_0
    &=
    \lim_{\lambda\downarrow0}
    \qty(F|_{[\kappa_e,\infty)})^{-1}(\lambda)
    =
    \qty(F|_{[\kappa_e,\infty)})^{-1}(0)
    =
    \kappa_e.
\end{align}

Thus, $\kappa_e$ and $\kappa_0$ have distinct definitions: $\kappa_e$ is the largest algebraic root of the continuously extended zero equation, whereas $\kappa_0$ is the endpoint selected by the positive-$\lambda$ solution branch. Their equality is a consequence of the branch analysis above.

When $\tilde{\gamma} > 1$, it obeys the two-sided estimate
\begin{align}
    \frac{p}{\Tr[\Sigma^{-1}]}(\tilde{\gamma}-1)
    \;\leq\;
    \kappa_0
    \;\leq\;
    \frac{\Tr[\Sigma]}{p}\,(\tilde{\gamma}-1) \,.
\end{align}
The lower bound in Eq.~\eqref{eq:kappa_zero_two_sided_bounds} follows by inserting $\mathrm{df}_1(\kappa_0) = N_{\mathrm{tr}}$ into Eq.~\eqref{eq:df_inequality1} and rearranging.
The upper bound follows by inserting $\mathrm{df}_1(\kappa_0) = N_{\mathrm{tr}}$ into Eq.~\eqref{eq:df_inequality2} and rearranging.
The bounds are tight when all eigenvalues are equal, $\xi_1 = \cdots = \xi_p$.

For $\lambda > 0$, the inverse-function theorem gives
\begin{align}\label{eq:kappa_derivative_appendix}
    \frac{\ind\kappa_\lambda}{\ind\lambda}
    =
    \frac{1}{F'(\kappa_\lambda)}
    =
    \frac{1}{1 - \eta(\kappa_\lambda)} \,.
\end{align}
This agrees with Eq.~\eqref{eq:kappa_derivative}.
It is useful to record the following exact decomposition of the denominator $1 - \eta(\kappa_\lambda)$ in Eq.~\eqref{eq:kappa_derivative_appendix}.
Combining $\mathrm{df}_1(\kappa)-\mathrm{df}_2(\kappa)=\kappa\sum_{i=1}^p\xi_i/(\xi_i+\kappa)^2$ with $\mathrm{df}_1(\kappa_\lambda)=N_{\mathrm{tr}}(1-\lambda/\kappa_\lambda)$ from Eq.~\eqref{eq:self_consistent_equation2} gives
\begin{align}\label{eq:eta_exact_identity}
    1-\eta(\kappa_\lambda)
    =
    \frac{\lambda}{\kappa_\lambda}
    +
    \frac{\kappa_\lambda}{N_{\mathrm{tr}}}\Tr[\Sigma(\Sigma + \kappa_\lambda I_p)^{-2}] \,.
\end{align}
Both terms on the right-hand side are positive for $\lambda > 0$.
When $\lambda = 0$, the first term is $1 - \tilde{\gamma}$ (shown in Eq.~\eqref{eq:kappa_leading_underparameterized}) and the second term is zero if $\tilde{\gamma} \leq 1$ and strictly positive if $\tilde{\gamma} > 1$.
Thus, even when $\lambda = 0$, we have $1 - \eta(\kappa_\lambda)$ strictly positive for $\tilde{\gamma} \neq 1$:
\begin{align}
    1 - \eta(\kappa_0)
    =
    \begin{cases}
        1 - \tilde{\gamma} > 0 \,, & \tilde{\gamma} < 1 \,,\\
        0 \,, & \tilde{\gamma} = 1 \,,\\
        \frac{\kappa_0}{N_{\mathrm{tr}}}\Tr[\Sigma(\Sigma + \kappa_0 I_p)^{-2}] > 0 \,, & \tilde{\gamma} > 1 \,.
    \end{cases}
\end{align}

To summarize, we have the following proposition.
\begin{proposition}[Strict positivity of the Ridgeless Endpoint $\kappa_0$]
    \label{prop:kappa_zero_positive}
    Let $p,N_{\mathrm{tr}}$ be fixed and $\Sigma\succ0$. Then
    \begin{align}\label{eq:kappa_zero_piecewise}
        \kappa_0 = \kappa_e
        =
        \begin{cases}
            0 \,,
            & \tilde{\gamma} \leq 1 \,,\\
            \text{the unique positive solution of }
            \mathrm{df}_1(\kappa) = N_{\mathrm{tr}} \,,
            & \tilde{\gamma} > 1 \,,
        \end{cases}
    \end{align}
    and $\kappa_\lambda$ is strictly increasing in $\lambda \in [0,\infty)$.
    When $\tilde{\gamma} > 1$, it obeys the two-sided estimate
    \begin{align}\label{eq:kappa_zero_two_sided_bounds}
        \frac{p}{\Tr[\Sigma^{-1}]}(\tilde{\gamma}-1)
        \;\leq\;
        \kappa_0
        \;\leq\;
        \frac{\Tr[\Sigma]}{p}\,(\tilde{\gamma}-1) \,.
    \end{align}
    Furthermore,
    \begin{align}\label{eq:one_minus_eta_at_kappa_zero}
        1 - \eta(\kappa_0)
        =
        \begin{cases}
            1 - \tilde{\gamma} > 0 \,, & \tilde{\gamma} < 1 \,,\\
            0 \,, & \tilde{\gamma} = 1 \,,\\
            \frac{\kappa_0}{N_{\mathrm{tr}}}\Tr[\Sigma(\Sigma + \kappa_0 I_p)^{-2}] > 0 \,, & \tilde{\gamma} > 1 \,.
        \end{cases}
    \end{align}
\end{proposition}

Proposition~\ref{prop:kappa_zero_positive} is the precise characterization of the \emph{ridgeless endpoint} $\kappa_0$ discussed in Section~\ref{subsec:asymptotic_test_risk}: in the overparameterized regime, the test-risk DE behaves as if a strictly positive ridge $\kappa_0=\kappa_e$ were present even without explicit regularization. Equation~\eqref{eq:one_minus_eta_at_kappa_zero} shows that the denominator $1 - \eta(\kappa_\lambda)$ in the test-risk DE of Theorem~\ref{thm:test_risk_DE} stays bounded away from zero as $\lambda \downarrow 0$ for $\tilde{\gamma} \neq 1$.
Note that only $\Sigma \succ 0$ is used; no assumption on the shape of the spectrum is needed.
The lower bound in Eq.~\eqref{eq:kappa_zero_two_sided_bounds} can nevertheless be very loose in the quantum setting, since highly expressive circuits produce exponentially small eigenvalues and hence a large $\Tr[\Sigma^{-1}]$; the \emph{existence} of $\kappa_0 > 0$ is unaffected.

\subsection{Small-Regularization Asymptotics}
\label{appdx:subsec:kappa_asymptotics}
Since $p - \mathrm{df}_1(\kappa) = \kappa\Tr[(\Sigma + \kappa I_p)^{-1}]$ for $\kappa \geq 0$, we can rewrite $F(\kappa)$ as
\begin{align}\label{eq:lambda_of_kappa}
    F(\kappa)
    =
    \kappa\left[
        (1 - \tilde{\gamma})
        +
        \frac{\kappa}{N_{\mathrm{tr}}}\Tr[(\Sigma + \kappa I_p)^{-1}]
    \right] \,.
\end{align}
We now read off the behavior of $\kappa_\lambda$ as $\lambda\downarrow0$ from Eq.~\eqref{eq:lambda_of_kappa} in the three regimes.
Related analyses can be found in Ref.~\cite{bach2024highdimensional}.
The results are summarized in Table~\ref{tab:kappa_regimes}.

\begin{table}[H]
    \centering
    \setlength{\tabcolsep}{15pt}
    \begin{tabular}{llll}
    \toprule
    Regime & Ridgeless endpoint $\kappa_0$ & Small-$\lambda$ behavior of $\kappa_\lambda$ & Small-$\lambda$ behavior of $1/[1 - \eta(\kappa_\lambda)]$ \\
    \midrule
    $\tilde{\gamma} < 1$
        & $0$
        & $\dfrac{\lambda}{1 - \tilde{\gamma}}+\calO(\lambda^2)$
        & $\to\dfrac{1}{1 - \tilde{\gamma}}$ \\[2.2ex]
    $\tilde{\gamma} = 1$
        & $0$
        & $\sqrt{\lambda/c_2}$, \; $c_2=\Tr[\Sigma^{-1}]/N_{\mathrm{tr}}$
        & $\sim\dfrac{1}{2\sqrt{c_2\lambda}} \to \infty$ \\[2.2ex]
    $\tilde{\gamma} > 1$
        & $>0$, \; $\mathrm{df}_1(\kappa_0)=N_{\mathrm{tr}}$
        & $\kappa_0+\dfrac{\lambda}{1 - \eta(\kappa_0)}+\calO(\lambda^2)$
        & $\to\dfrac{1}{1 - \eta(\kappa_0)} < \infty$ \\
    \bottomrule
    \end{tabular}
    \caption{Behavior of the effective regularization parameter $\kappa_\lambda$ as the explicit regularization becomes small, at fixed $p$ and $N_{\mathrm{tr}}$. Here $\kappa_0=\lim_{\lambda\downarrow0}\kappa_\lambda = \kappa_e$. Only the interpolation threshold $\tilde{\gamma} = 1$ produces a diverging $1/[1 - \eta(\kappa_\lambda)]$, which is the origin of the double-descent peak in the test-risk DE.}
    \label{tab:kappa_regimes}
\end{table}

\paragraph{Underparameterized regime ($\tilde{\gamma} < 1$).}
Proposition~\ref{prop:kappa_zero_positive} gives $\kappa_0 = \kappa_e = 0$.
Solving Eq.~\eqref{eq:lambda_of_kappa} for $\kappa_\lambda$ gives the exact relation
\begin{align}\label{eq:kappa_exact_underparameterized}
    \kappa_\lambda
    =
    \frac{\lambda}{
        (1-\tilde{\gamma})
        +
        \frac{\kappa_\lambda}{N_{\mathrm{tr}}}\Tr[(\Sigma+\kappa_\lambda I_p)^{-1}]
    } \,,
\end{align}
and since $\kappa_\lambda \downarrow 0$ and $\Tr[(\Sigma+\kappa I_p)^{-1}]\to\Tr[\Sigma^{-1}] < \infty$, we obtain
\begin{align}\label{eq:kappa_leading_underparameterized}
    \frac{\kappa_\lambda}{\lambda}
    \longrightarrow
    \frac{1}{1 - \tilde{\gamma}}
    \qquad (\lambda \downarrow 0) \,.
\end{align}
If in addition $\kappa_\lambda < \xi_{\min}$, so that the geometric series
$\Tr[(\Sigma+\kappa I_p)^{-1}] = \Tr[\Sigma^{-1}]-\kappa\Tr[\Sigma^{-2}]+\calO(\kappa^2)$ converges, then Eq.~\eqref{eq:lambda_of_kappa} becomes
\begin{align}
    \lambda
    =
    (1 - \tilde{\gamma})\kappa_\lambda
    +
    \frac{\Tr[\Sigma^{-1}]}{N_{\mathrm{tr}}}\kappa_\lambda^2
    -
    \frac{\Tr[\Sigma^{-2}]}{N_{\mathrm{tr}}}\kappa_\lambda^3
    +
    \calO(\kappa_\lambda^4) \,,
\end{align}
and inverting this series order by order gives the second-order correction
\begin{align}\label{eq:kappa_second_order_underparameterized}
    \kappa_\lambda
    =
    \frac{\lambda}{1-\tilde{\gamma}}
    -
    \frac{\Tr[\Sigma^{-1}]}{N_{\mathrm{tr}}}
    \frac{\lambda^2}{(1-\tilde{\gamma})^3}
    +
    \calO(\lambda^3) \,.
\end{align}
To see the inversion explicitly, write $a:=1-\tilde{\gamma}$ and $b:=\Tr[\Sigma^{-1}]/N_{\mathrm{tr}}$, and substitute $\kappa_\lambda = A\lambda + B\lambda^2+\calO(\lambda^3)$ into $\lambda = a\kappa_\lambda + b\kappa_\lambda^2+\calO(\kappa_\lambda^3)$. Comparing the coefficients of $\lambda$ and $\lambda^2$ gives $A = a^{-1}$ and $B=-ba^{-3}$, which yields Eq.~\eqref{eq:kappa_second_order_underparameterized}.
At fixed $p$, the condition $\kappa_\lambda<\xi_{\min}$ is eventually automatic because $\xi_{\min}>0$ and $\kappa_\lambda\to0$. Using Eq.~\eqref{eq:kappa_leading_underparameterized}, its onset is roughly $\lambda\ll(1-\tilde{\gamma})\xi_{\min}$. This condition need not be uniform in $p$: if $\xi_{\min}$ is exponentially small for an expressive quantum feature map, the second-order expansion is informative only in a correspondingly narrow window, whereas Eqs.~\eqref{eq:kappa_exact_underparameterized}--\eqref{eq:kappa_leading_underparameterized} do not require this expansion.

Since $\eta$ is continuous on $[0,\infty)$ and $\kappa_\lambda \to 0$ as $\lambda \downarrow 0$, we have
\begin{align}
    1 - \eta(\kappa_\lambda)
    &\longrightarrow
    1 - \eta(0)
    =
    1 - \frac{p}{N_{\mathrm{tr}}}
    =
    1 - \tilde{\gamma} > 0
    \quad
    \implies
    \quad
    \frac{1}{1 - \eta(\kappa_\lambda)}
    \longrightarrow
    \frac{1}{1 - \tilde{\gamma}} \,.
\end{align}

\paragraph{At the interpolation threshold ($\tilde{\gamma} = 1$).}
Here $\kappa_0 = \kappa_e = 0$, and $1 - \tilde{\gamma} = 0$ in Eq.~\eqref{eq:lambda_of_kappa}, leaving the \emph{exact} relation
\begin{align}
    \lambda
    =
    \frac{\kappa_\lambda^2}{N_{\mathrm{tr}}}\Tr[(\Sigma+\kappa_\lambda I_p)^{-1}] \,.
\end{align}
Letting $\lambda\downarrow0$ and using $\Tr[(\Sigma+\kappa I_p)^{-1}]\to\Tr[\Sigma^{-1}]$,
\begin{align}\label{eq:kappa_sqrt_scaling}
    \kappa_\lambda
    \sim
    \sqrt{\frac{\lambda}{c_2}} \,,
    \qquad
    c_2
    :=
    \frac{\Tr[\Sigma^{-1}]}{N_{\mathrm{tr}}}
    =
    \frac{1}{N_{\mathrm{tr}}}\sum_{i=1}^p\frac{1}{\xi_i}
    >0 \,,
\end{align}
which recovers the scaling reported in Refs.~\cite{defilippis2024dimension,bach2024highdimensional}.
To obtain the corresponding behavior of $1-\eta(\kappa_\lambda)$, we expand $\mathrm{df}_2$ directly. Since $p=N_{\mathrm{tr}}$,
\begin{align}
    \mathrm{df}_2(\kappa)
    =
    \sum_{i=1}^p\left(1+\frac{\kappa}{\xi_i}\right)^{-2}
    =
    N_{\mathrm{tr}}
    -2\kappa\Tr[\Sigma^{-1}]
    +\calO(\kappa^2) \,.
\end{align}
Hence
\begin{align}
    1 - \eta(\kappa_\lambda)
    =
    1 - \frac{\mathrm{df}_2(\kappa_\lambda)}{N_{\mathrm{tr}}}
    =
    2c_2\kappa_\lambda+\calO(\kappa_\lambda^2)
    \sim
    2\sqrt{c_2\lambda} \,.
\end{align}
Therefore,
\begin{align}\label{eq:one_minus_eta_gamma1}
    \frac{1}{1 - \eta(\kappa_\lambda)}
    \sim
    \frac{1}{2\sqrt{c_2\lambda}} \,.
\end{align}
Thus, the variance term $\sigma^2\eta(\kappa_\lambda)/(1 - \eta(\kappa_\lambda))$ of Theorem~\ref{thm:test_risk_DE} grows as $\lambda^{-1/2}$, while the bias prefactor $\kappa_\lambda^2/(1 - \eta(\kappa_\lambda))$ decays like $\lambda^{1/2}$; this is the double-descent peak discussed in Section~\ref{subsec:asymptotic_test_risk}.
The expansion requires $\kappa_\lambda$ to be small relative to the relevant spectral scales. A sufficient condition is $\kappa_\lambda\ll\xi_{\min}$, which, using Eq.~\eqref{eq:kappa_sqrt_scaling}, corresponds roughly to $\lambda\ll c_2\xi_{\min}^2$. Small eigenvalues can therefore make the valid window narrow; the value of $c_2$ alone does not determine where the $\sqrt{\lambda}$ regime begins.

\paragraph{Overparameterized regime ($\tilde{\gamma} > 1$).}
Proposition~\ref{prop:kappa_zero_positive} gives $\kappa_0 = \kappa_e > 0$, and Eq.~\eqref{eq:one_minus_eta_at_kappa_zero} gives $1 - \eta(\kappa_0) > 0$.
Since $F'(\kappa_0) = 1 - \eta(\kappa_0) > 0$, the inverse function is differentiable at $\lambda = 0$. Equation~\eqref{eq:kappa_derivative_appendix} therefore gives
\begin{align}\label{eq:kappa_expansion_overparameterized}
    \kappa_\lambda
    =
    \kappa_0
    +
    \frac{\lambda}{1 - \eta(\kappa_0)}
    +
    \calO(\lambda^2) \,.
\end{align}
In particular, $\kappa_\lambda \to \kappa_0$ as $\lambda \downarrow 0$. Hence, by continuity of $\eta$ at $\kappa_0 > 0$,
\begin{align}
    \frac{1}{1-\eta(\kappa_\lambda)}
    &\longrightarrow
    \frac{1}{1-\eta(\kappa_0)}
    <\infty.
\end{align}

\subsection{Isotropic Covariance: Closed Form}
\label{appdx:subsec:kappa_isotropic}
For $\Sigma = I_p$ the self-consistent equation~\eqref{eq:self-consistent_eq} is quadratic in $\kappa_\lambda$ and can be solved in closed form,
\begin{align}\label{eq:kappa_isotropic_closed_form}
    \kappa_\lambda
    =
    \frac{\tilde{\gamma}-1 + \lambda + \sqrt{(\tilde{\gamma} - 1 + \lambda)^2 + 4\lambda}}{2} \,.
\end{align}
For $\lambda = 0$, this gives the ridgeless endpoint $\kappa_0 = (\tilde{\gamma} - 1 + |\tilde{\gamma} - 1|)/2 = \max\{\tilde{\gamma} - 1,0\}$.
This is in accordance with Proposition~\ref{prop:kappa_zero_positive} for $\Sigma = I_p$.
Figure~\ref{fig:kappa_isotropic} shows the resulting behavior of $\kappa_\lambda$ as a function of $\lambda$ and $\tilde{\gamma}$, where the three scaling regimes of Table~\ref{tab:kappa_regimes} are clearly visible.

\begin{figure}[t]
    \centering
    \includegraphics[width=10cm]{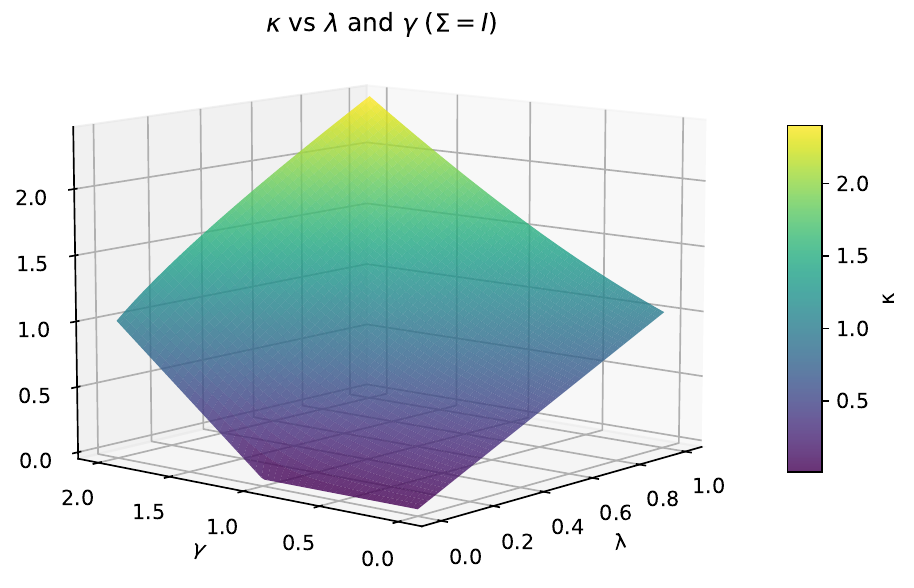}
    \caption{Behavior of $\kappa_\lambda$ as a function of $\lambda$ and $\tilde{\gamma}$ for isotropic covariance $\Sigma = I_p$.}
    \label{fig:kappa_isotropic}
\end{figure}

\clearpage
\newpage
\section{Estimation of the Population Covariance Matrix $\Sigma$ and the Projected Target Vector $\bs{\beta}_*$}\label{appdx:sec:estimation_Sigma_beta}

To plot the test-risk DE curves from Eq.~\eqref{eq:test_risk_DE}, we need to estimate the population covariance matrix $\Sigma = \diag(\xi_1,\ldots,\xi_p) \in \bbR^{p\times p}$ and the projected target vector $\bs{\beta}_* = (\beta_{*,1},\ldots,\beta_{*,p})\T\in\bbR^p$. Let $\calD_{\mathrm{est}} = \{(\bs{u}_i,y_i)\}_{i=1}^{N_{\mathrm{est}}}$ be an estimation dataset independent of the training and test datasets, and let $\calU_{\mathrm{est}}=\{\bs{u}_i\}_{i=1}^{N_{\mathrm{est}}}$ denote its inputs. The set $\calU_{\mathrm{est}}$ is used to estimate the eigenvalues of $\Sigma$. To estimate $\bs{\beta}_*$, we use the available noise-free teacher values $f_*(\bs{u}_i)$ in the synthetic experiments and the observed class labels $y_i$ in the Fashion-MNIST experiments. For Fashion-MNIST, this substitution is an empirical approximation because the noise-free target values are unavailable.

\subsection{Estimation of the Population Spectrum}
Let $X_{\mathrm{est}}\in\bbR^{N_{\mathrm{est}}\times p}$ be the effective feature matrix whose $i$-th row is $\phi(\bs{u}_i)\T$, and let $K_{\mathrm{est}}:=X_{\mathrm{est}}X_{\mathrm{est}}\T\in\bbR^{N_{\mathrm{est}}\times N_{\mathrm{est}}}$ be the corresponding kernel matrix. For every finite $N_{\mathrm{est}}$, the nonzero eigenvalues of the sample covariance matrix $X_{\mathrm{est}}\T X_{\mathrm{est}}/N_{\mathrm{est}}$ are exactly the same as those of $K_{\mathrm{est}}/N_{\mathrm{est}}$. Moreover, for each fixed feature map, $X_{\mathrm{est}}\T X_{\mathrm{est}}/N_{\mathrm{est}}$ converges to $\Sigma$ as $N_{\mathrm{est}}$ increases. We therefore use the nonzero eigenvalues of $K_{\mathrm{est}}/N_{\mathrm{est}}$ computed from a large estimation dataset $N_{\mathrm{est}}$ as empirical approximations of the eigenvalues of $\Sigma$.

Let $\xi^{(\mathrm{est})}_1\geq\cdots\geq\xi^{(\mathrm{est})}_{N_{\mathrm{est}}}\geq0$ denote the eigenvalues of $K_{\mathrm{est}}/N_{\mathrm{est}}$.
We then use the compact rank-$p$ eigendecomposition
\begin{align}
    \frac{1}{N_{\mathrm{est}}}K_{\mathrm{est}}
    = H_p \Lambda_p^{(\mathrm{est})} H_p\T
    = \sum_{i=1}^{p}\xi_i^{(\mathrm{est})}\bs{h}_i\bs{h}_i\T \,.
\end{align}
Here $\Lambda_p^{(\mathrm{est})} := \diag(\xi_1^{(\mathrm{est})},\ldots,\xi_p^{(\mathrm{est})})$ and $H_p := [\bs{h}_1,\ldots,\bs{h}_p]\in\bbR^{N_{\mathrm{est}}\times p}$, where $\bs{h}_k\T\bs{h}_l = \delta_{kl}$. Because $\Sigma$ is usually unknown, we use $\Lambda_p^{(\mathrm{est})}$ as its empirical approximation in Eq.~\eqref{eq:test_risk_DE}.
The remaining $N_{\mathrm{est}} - p$ Gram-matrix eigenvalues are numerically null and do not enter the $p$-dimensional test-risk DE formula.

\subsection{Approximation of Eigenfunction Values at the Data Samples}
To estimate the projected target vector $\bs{\beta}_*$, we require approximations of the eigenfunction values at the data samples, i.e., $\psi_k(\bs{u}_j)$ for $j \in [N_{\mathrm{est}}]$ and $k \in [p]$.
We employ the Nystr\"om method, which discretizes the integral eigenvalue equation $\int_{\calU} k(\bs{u}, \bs{u}') \psi_k(\bs{u}') \ind\mu(\bs{u}') = \xi_k \psi_k(\bs{u})$ using the empirical measure $\frac{1}{N_{\mathrm{est}}} \sum_{j=1}^{N_{\mathrm{est}}} \delta_{\bs{u}_j}$.

Applying this approximation to the inputs in $\calU_{\mathrm{est}}$ yields the matrix equation:
\begin{align}
    \forall i \in [N_{\mathrm{est}}], \quad
    \frac{1}{N_{\mathrm{est}}} \sum_{j=1}^{N_{\mathrm{est}}} k(\bs{u}_i, \bs{u}_j) \psi^{(\mathrm{est})}_k(\bs{u}_j)
    &= \xi^{(\mathrm{est})}_k \psi^{(\mathrm{est})}_k(\bs{u}_i) \,.
\end{align}
Let $\bs{\psi}^{(\mathrm{est})}_k = [\psi^{(\mathrm{est})}_k(\bs{u}_1), \dots, \psi^{(\mathrm{est})}_k(\bs{u}_{N_{\mathrm{est}}})]\T$ be the approximate vector of eigenfunction evaluations. The equation above can be rewritten in matrix form as
\begin{align}
    \frac{1}{N_{\mathrm{est}}} K_{\mathrm{est}} \bs{\psi}^{(\mathrm{est})}_k = \xi^{(\mathrm{est})}_k \bs{\psi}^{(\mathrm{est})}_k \,.
\end{align}
Comparing this with the eigendecomposition $\frac{1}{N_{\mathrm{est}}} K_{\mathrm{est}} \bs{h}_k = \xi^{(\mathrm{est})}_k \bs{h}_k$, we see that $\bs{\psi}^{(\mathrm{est})}_k$ is proportional to the eigenvector $\bs{h}_k$. The scaling factor is determined by the orthonormality condition, i.e., $\int \psi_k(\bs{u})^2 \ind\mu(\bs{u}) = 1$. Its empirical counterpart is
\begin{align}
    \frac{1}{N_{\mathrm{est}}} \sum_{j=1}^{N_{\mathrm{est}}} \psi^{(\mathrm{est})}_k(\bs{u}_j)^2 = 1
    \iff \frac{1}{N_{\mathrm{est}}} \|\bs{\psi}^{(\mathrm{est})}_k\|_2^2 = 1
    \iff \|\bs{\psi}^{(\mathrm{est})}_k\|_2 = \sqrt{N_{\mathrm{est}}} \,.
\end{align}
Since the eigenvectors $\{\bs{h}_k\}_{k=1}^p$ are normalized such that $\|\bs{h}_k\|_2 = 1$, we define the estimated eigenfunction values by
\begin{align}
    \bs{\psi}^{(\mathrm{est})}_k := \sqrt{N_{\mathrm{est}}}\,\bs{h}_k \,.
\end{align}

\subsection{Estimation of the Projected Target Vector $\bs{\beta}_*$}
The entries of the projected target vector $\beta_{*,k}$ are defined as the projections of the target function $f_*$ onto the orthonormal eigenfunctions $\{\psi_k\}_{k=1}^p$:
\begin{align}
    \beta_{*,k} = \int_{\calU} f_*(\bs{u}) \psi_k(\bs{u}) \ind\mu(\bs{u}) \,.
\end{align}
Since $f_*$ is known for the synthetic experiments, we use its noise-free values rather than the noisy responses.
However, for the Fashion-MNIST experiments, we use the observed class labels $y_j$ as empirical approximations of the unknown noise-free values $f_*(\bs{u}_j)$.
In both cases, the actual values of the eigenfunctions $\psi_k(\bs{u}_j)$ are unavailable, so we use $\psi^{(\mathrm{est})}_k(\bs{u}_j)$ instead.
Define $\bs{t}_{\mathrm{est}} := (t_1,\ldots,t_{N_{\mathrm{est}}})\T$, where $t_j := f_*(\bs{u}_j)$ for synthetic data and $t_j := y_j$ for Fashion-MNIST.
For each $k\in[p]$, we define $\beta_{*,k}^{(\mathrm{est})}$, the empirical approximation of the projected target coefficient $\beta_{*,k}$, by
\begin{align}
    \beta^{(\mathrm{est})}_{*,k}
    &
    := \frac{1}{N_{\mathrm{est}}} \sum_{j=1}^{N_{\mathrm{est}}} t_j \psi^{(\mathrm{est})}_k(\bs{u}_j) \nonumber\\
    &
    = \frac{1}{N_{\mathrm{est}}} \sum_{j=1}^{N_{\mathrm{est}}} t_j \qty[\sqrt{N_{\mathrm{est}}}(\bs{h}_k)_j] \nonumber\\
    &
    = \frac{1}{\sqrt{N_{\mathrm{est}}}}\bs{h}_k\T\bs{t}_{\mathrm{est}} \,.
\end{align}
Collecting these empirical coefficients gives the empirical estimator
\begin{align}
    \bs{\beta}^{(\mathrm{est})}_*
    := \frac{1}{\sqrt{N_{\mathrm{est}}}}H_p\T\bs{t}_{\mathrm{est}}
    \in \bbR^p \,.
\end{align}
Finally, we use $\Lambda_p^{(\mathrm{est})}$ and $\bs{\beta}_*^{(\mathrm{est})}$ in place of $\Sigma$ and $\bs{\beta}_*$ in Eq.~\eqref{eq:test_risk_DE}.

\section{Plot of the Normalized Effective Degrees of Freedom $\eta_\lambda$}\label{appdx:sec:effective_DOF}
\begin{figure}[H]
    \centering
    \includegraphics[width=10cm]{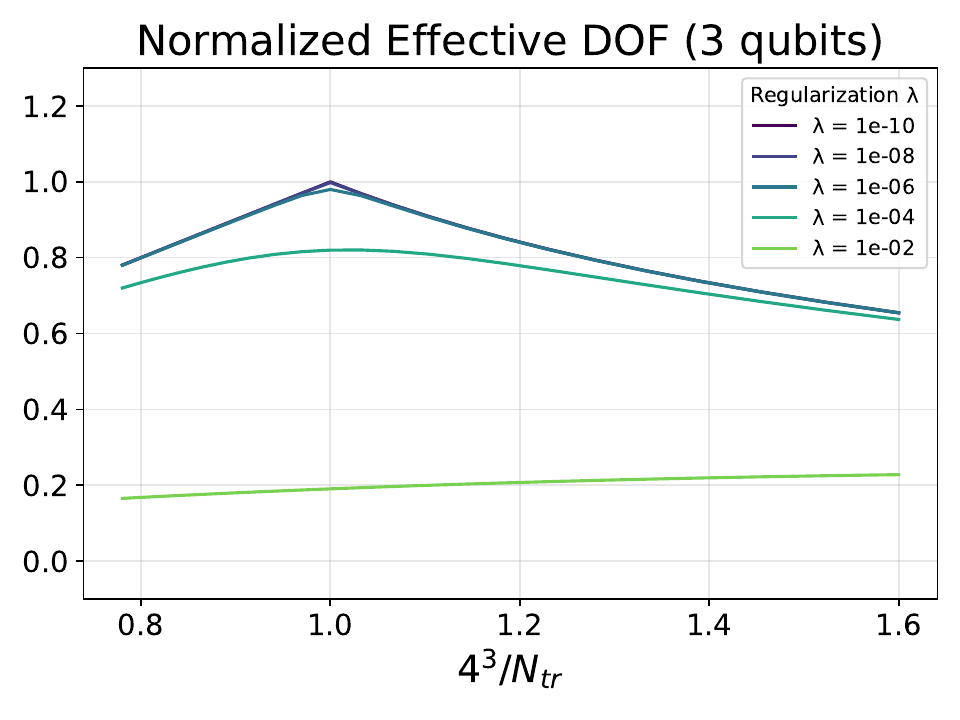}
    \caption{Plot of the normalized effective degrees of freedom $\eta_\lambda$ as a function of $\tilde{\gamma} = 4^3/N_{\mathrm{tr}}$ for the synthetic dataset with the 3-qubit HEA.}
    \label{fig:effective_DOF_synthetic_3qubits}
\end{figure}

\end{document}